\documentclass[11pt,letterpaper]{article}

\usepackage[bookmarks,colorlinks,breaklinks]{hyperref}
\hypersetup{urlcolor=blue, colorlinks=true, citecolor=green!50!black, linkcolor=blue}
\usepackage[letterpaper, left=1in, right=1in, top=0.9in, bottom=0.9in]{geometry}

\usepackage[T1]{fontenc}
\usepackage[utf8]{inputenc}
\usepackage[american]{babel}
\usepackage[normalem]{ulem}
\usepackage{amsmath, amssymb, cases, amsthm}
\usepackage{thmtools}
\usepackage{thm-restate}
\usepackage[shortlabels]{enumitem}

\usepackage{mdframed}
\usepackage{bbm}
\usepackage{bm}
\usepackage{microtype}
\usepackage{tabularx}
\usepackage{subcaption}
\usepackage{xcolor}
\usepackage{makecell}
\usepackage{mathtools}
\usepackage{float}
\usepackage{multicol}

\usepackage{comment}
\usepackage{multirow}
\usepackage{natbib}
\usepackage[capitalize,noabbrev]{cleveref}
\usepackage{graphics}

\usepackage{tikz}
\usetikzlibrary{arrows.meta, positioning, shapes.geometric}

\usepackage[nottoc]{tocbibind}

\declaretheorem[numberwithin=section,refname={Theorem,Theorems},Refname={Theorem,Theorems}]{theorem}

\declaretheorem[numberlike=theorem]{lemma}

\declaretheorem[numberlike=theorem]{corollary}
\declaretheorem[numberlike=theorem,style=definition]{definition}
\declaretheorem[numberlike=theorem]{claim}

\declaretheorem[numberlike=theorem,style=remark]{remark}

\declaretheorem[numberlike=theorem, refname={Observation,Observations},Refname={Observation,Observations},name={Observation}]{observation}



\def\final{0}  
\def\iflong{\iffalse}
\ifnum\final=0  
\newcommand{\danupon}[1]{{\bf \color{green!70!black} Danupon: #1}}
\newcommand{\martin}[1]{{\bf \color{blue!50!black} Martin: #1}}
\newcommand{\simon}[1]{{\bf \color{red!75!black} Simon: #1}}
\newcommand{\daniel}[1]{{\bf \color{yellow!60!black} daniel: #1}}
\newcommand{\dani}[1]{{\bf \color{orange!80!black} dani: #1}}
\newcommand{\jo}[1]{{\bf \color{brown!80!black} Joachim: #1}}
\newcommand{\zihang}[1]{{\bf \color{blue!80!black} Zihang: #1}}
\else 
\newcommand{\danupon}[1]{}
\newcommand{\martin}[1]{}
\newcommand{\simon}[1]{}
\newcommand{\daniel}[1]{}
\newcommand{\dani}[1]{}
\newcommand{\jo}[1]{}
\newcommand{\zihang}[1]{}
\fi  

\usepackage[ruled,vlined,linesnumbered,noresetcount]{algorithm2e}

\newcommand{\eps}{\varepsilon}

\newcommand{\set}[2][ ]{\{#2 \ifthenelse{\equal{#1}{ }}{ }{~|~#1}\}}

\newcommand{\GI}{\mathsf{GI}}
\newcommand{\isored}{\le_p^{\cong}}
\newcommand{\cP}{\mathsf{P}}
\newcommand{\cNP}{\mathsf{NP}}
\newcommand{\coma}{\mathsf{coMA}}

\newcommand{\attribute}{attribute\xspace}
\newcommand{\attributes}{attributes\xspace}
\newcommand{\attributed}{attributed\xspace}

\newcommand{\Attributed}{Attributed\xspace}
\newcommand{\att}{\alpha}
\newcommand{\advice}{advice\xspace}

\newcommand{\poly}{\mathrm{poly}}

\newcommand{\NP}{\cNP}
\newcommand{\coNP}{\mathsf{coNP}}

\newcommand{\UP}{\mathsf{UP}}
\newcommand{\coUP}{\mathsf{coUP}}

\Crefname{algocf}{Algorithm}{Algorithms}

\DeclareMathOperator{\trace}{\mathrm{Tr}}

\newcommand{\opt}{\mathit{OPT}}

\newcommand{\WLequiv}[1]{\equiv_{#1\text{-}\mathrm{WL}}}
\newcommand{\cfiFunc}{\mathrm{CFI}}
\newcommand{\cfiGraph}[2]{\cfiFunc(#1, #2)}
\newcommand{\Graphs}{\mathcal{G}} 
\newcommand{\homs}[2]{\#\mathrm{Hom}(#1 \to #2)}
\newcommand{\WLcost}{\mathrm{cost}}
\newcommand{\WLcheck}{\mathrm{check}}

\newcommand{\odd}{\mathrm{odd}}
\newcommand{\even}{\mathrm{even}}

\newcommand{\lovG}[2]{\mathcal{L}(#1, #2)}

\newcommand{\cspName}{\mathsf{Constrained}\text{ }\mathsf{Shortest}\text{ }\mathsf{Path}}
\newcommand{\cstName}{\mathsf{Constrained}\text{ }\mathsf{Spanning}\text{ }\mathsf{Tree}}
\newcommand{\negName}{\mathsf{Negative}\text{ }\mathsf{Triangle}}
\newcommand{\repPathName}{\mathsf{Replacement}\text{ }\mathsf{Path}}

\newcommand{\giName}{\mathsf{Graph}\text{ }\mathsf{Isomorphism}}
\newcommand{\vcName}{\mathsf{Vertex}\text{ }\mathsf{Cover}}
\newcommand{\isName}{\mathsf{Independent}\text{ }\mathsf{Set}}

\newcommand{\parityName}{\mathsf{Parity}\text{ }\mathsf{Game}}
\newcommand{\meanName}{\mathsf{Mean}\text{ }\mathsf{Payoff}\text{ }\mathsf{Game}}
\newcommand{\energyName}{\mathsf{Energy}\text{ }\mathsf{Game}}
\newcommand{\stochasticName}{\mathsf{Simple}\text{ }\mathsf{Stochastic}\text{ }\mathsf{Game}}

\newcommand{\tspName}{\mathsf{Traveling}\text{ }\mathsf{Salesman}\text{ }\mathsf{Problem}}

\newcommand{\metrictspName}{\mathsf{Metric }\text{ }\mathsf{Traveling }\text{ }\mathsf{Salesman}\text{ }\mathsf{Problem}}

\newcommand{\lspName}[1]{\ell_{#1}\text{-}\mathsf{Shortest}\text{ }\mathsf{Path}}

\newcommand{\lstName}[1]{\ell_{#1}\text{-}\mathsf{Spanning}\text{ }\mathsf{Tree}}

\newcommand{\maxFlowName}{\mathsf{Maximum}\text{ }\mathsf{Flow}}

\newcommand{\maxISName}{\mathsf{Maximum}\text{ }\mathsf{Independent}\text{ }\mathsf{Set}}

\newcommand{\maxCliqueName}{\mathsf{Maximum}\text{ }\mathsf{Clique}}

\newcommand{\APSP}{\mathsf{APSP}}
\newcommand{\csp}{\mathsf{CSP}}
\newcommand{\cst}{\mathsf{CST}}
\newcommand{\ProbCollection}{\mathcal{H}}

\newcommand{\SAT}{\mathsf{SAT}}
\newcommand{\kSAT}[1]{#1\text{-}\SAT}
\newcommand{\MaxkSAT}[1]{\mathsf{Max}\text{ }\kSAT{#1}}

\newcommand{\COL}{\mathsf{Coloring}}
\newcommand{\kCOL}[1]{#1\text{-}\COL}

\newcommand{\IS}{\mathsf{IS}}
\newcommand{\MaxCut}{\mathsf{MaxCut}}
\newcommand{\VC}{\mathsf{VC}}
\newcommand{\HIT}{\mathsf{Hitting\text{ }Set}}
\newcommand{\kMedian}{k\text{-}\mathsf{Median}}
\newcommand{\SetCover}{\mathsf{Set\text{ }Cover}}
\newcommand{\EXA}{\mathsf{Exact\text{ }Cover}}

\newcommand{\LPZEROONE}{0/1\text{-}\mathsf{Integer\text{ }Linear\text{ }Programming}}
\newcommand{\Clique}{\mathsf{Clique}}
\newcommand{\kClique}[1]{#1\text{-}\mathsf{Clique}}
\newcommand{\SetPacking}{\mathsf{Set}\text{ }\mathsf{Packing}}
\newcommand{\CliqueCover}{\mathsf{Clique}\text{ }\mathsf{Cover}}
\newcommand{\DirHam}{\mathsf{Directed}\text{ }\mathsf{Hamiltonian}\text{ }\mathsf{Cycle}}

\newcommand{\UndirHam}{\mathsf{Undirected}\text{ }\mathsf{Hamiltonian}\text{ }\mathsf{Cycle}}

\newcommand{\Ham}{\mathsf{Hamiltonian}\text{ }\mathsf{Cycle}}

\newcommand{\STEINER}{\mathsf{Steiner}\text{ }\mathsf{Tree}}
\newcommand{\FeedbackArcSet}{\mathsf{Feedback}\text{ }\mathsf{Arc}\text{ }\mathsf{Set}}
\newcommand{\FeedbackVertexSet}{\mathsf{Feedback}\text{ }\mathsf{Vertex}\text{ }\mathsf{Set}}
\newcommand{\DMatching}{\mathsf{3D}\text{-}\mathsf{Matching}}
\newcommand{\MetricTSP}{\mathsf{Metric}\text{ }\mathsf{TSP}}
\newcommand{\TSP}{\mathsf{TSP}}
\newcommand{\kCenter}{k\text{-}\mathsf{Center}}
\newcommand{\GAP}[1]{\mathsf{Gap}\text{-}#1}

\makeatletter
\def\blfootnote{\gdef\@thefnmark{}\@footnotetext}
\makeatother

\newcommand{\orcid}[1]{\href{https://orcid.org/#1}{\includegraphics[height=1.8ex]{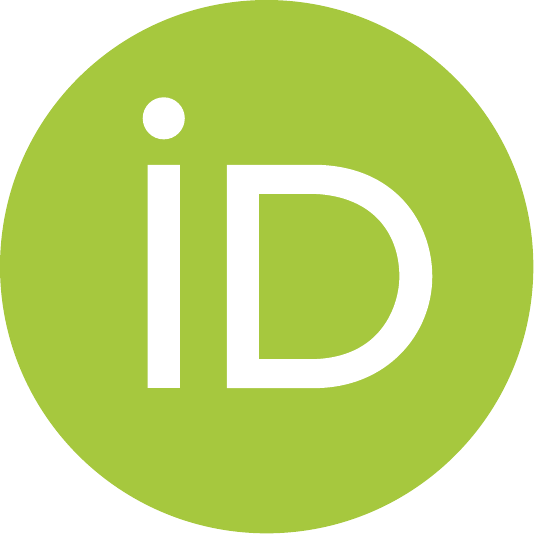}}}

\title{Can We Break Fine-Grained and NP-Hardness Barriers if We've Seen the Graph Before? The Isomorphic-Priors Model   }

\author{
Dani Dorfman \orcid{https://orcid.org/0009-0000-5134-4645} \\
Max Planck Institute for Informatics \\ 
Saarland Informatics Campus
\and
Simon D\"oring \orcid{https://orcid.org/0009-0002-6667-5257} \\
Max Planck Institute for Informatics \\ 
Saarland Informatics Campus
\and
Martin G. Herold\footnote{Martin G. Herold was initially funded by the Deutsche Forschungsgemeinschaft (DFG, German Research Foundation) – Project number 399223600.} \orcid{https://orcid.org/0009-0002-1804-2842} \\
Max Planck Institute for Informatics \\ 
Saarland Informatics Campus
\and
Daniel Neuen \orcid{https://orcid.org/0000-0002-4940-0318} \\
TU Dresden
\and
Joachim Spoerhase \orcid{https://orcid.org/0000-0002-2601-6452} \\
University of Liverpool
\and
Danupon Nanongkai \orcid{https://orcid.org/0000-0003-4468-2675} \\
Max Planck Institute for Informatics \\ 
Saarland Informatics Campus
\and
Zihang Wu \orcid{https://orcid.org/0009-0002-5396-5036} \\
Max Planck Institute for Informatics \\ 
Saarland Informatics Campus
}

\date{}

\begin{document}
	
	\begin{titlepage}
		\maketitle \pagenumbering{roman}

\begin{abstract}
If we run a heavy-duty computation on prior data, can we avoid repeated computation for similar future inputs?
Inspired by this question, we introduce a new computational model for graph problems called \emph{algorithms with isomorphic priors}.  
Solving a graph problem $\Pi$ in this model involves two phases: (i) The \emph{preprocessing} phase quickly analyzes {\em prior graphs} $G_1, \ldots, G_k$ along with the (previously computed) exact optimal values $\mathit{OPT}(G_i)$. (ii) 
Subsequently, given a new graph $H$, a fast \emph{query} phase must either (a) output the exact solution $\mathit{OPT}(H)$, or (b) correctly report that $H$ is not isomorphic to any $G_i$.\footnote{Our algorithms sometimes return $\mathit{OPT}(H)$ even though no graph $G_i$ is isomorphic to $H$. We focus on problems $\Pi$ for which isomorphic instances yield identical optimal values, a property that holds for most natural graph problems.} {\em Can we avoid computing $\mathit{OPT}(H)$ from scratch when $H$ is isomorphic to some $G_i$?} We show that this is the case for a number of problems; for many others, we establish conditional lower bounds.

\textbf{(1)} Some \emph{NP-hard} problems, including
$\cspName$ and $\lspName{p}$ and $\cstName$,
admit polynomial preprocessing and query times in our model.
In contrast, almost all of Karp's 21 NP-complete problems and
$(2-\varepsilon)$-approximate $k$-$\mathsf{Center}$, for every fixed
$\varepsilon>0$, admit no such algorithms unless Graph Isomorphism
($\mathsf{GI}$) is in $\mathrm{P}$, even with $O(1)$ priors.
Unless $\mathsf{GI}\in\mathrm{coMA}$, these lower bounds persist
with {\em unbounded preprocessing} provided the data structure size
and query time remain polynomial.

\textbf{(2)} In contrast to conditional $n^{3-o(1)}$ fine-grained lower bounds, our framework achieves an $O(n^\omega)$ query time for $\negName$ and a near-linear query time for $\repPathName$. It also achieves near-linear query time for $\maxFlowName$.

\textbf{(3)} While it remains a major open problem whether \emph{infinite-duration games} (such as $\mathsf{Parity}\text{ }\mathsf{Game}$, $\mathsf{Mean}\text{ }\mathsf{Payoff}\text{ }\mathsf{Game}$, $\mathsf{Energy}\text{ }\mathsf{Game}$, and $\mathsf{Stochastic}\text{ }\mathsf{Game}$) admit polynomial-time algorithms, they can be easily solved in near-linear time within our model.

Our proofs rely on a simple combination of existing tools and are accessible to readers without specialized background in these methods. The model leads to open questions about, e.g., approximation guarantees and
fine-grained speedups, including whether isomorphic priors enable
faster algorithms for $k\text{-}\mathsf{Clique}$. We formulate a $k\text{-}\mathsf{Clique}$
hypothesis which implies $\mathsf{GI}\notin\mathrm{P}$.
\end{abstract}

		\setcounter{tocdepth}{3}
		\newpage
		\tableofcontents
		\newpage
	\end{titlepage}
	
	\newpage
	\pagenumbering{arabic}

\section{Introduction} \label{sec:intro}

If we invest in heavy-duty computation on prior data, can we avoid repeating this effort for similar future inputs? In an era defined by data abundance, this has become a central algorithmic challenge. There is a pressing need for data-driven algorithms that exploit knowledge of existing instances, fundamentally departing from the traditional worst-case design paradigm where inputs are analyzed without any prior context.

In this paper, we explore this question theoretically by proposing a new computational model for graphs and related objects.
While this motivation is shared by existing paradigms, such as average-case analysis over input distributions or learning-augmented algorithms with oracles, our model uniquely grants direct access to raw prior data without relying on learned oracles or distributions.

\paragraph{Model: Algorithms with Isomorphic Priors.}
Solving a graph problem $\Pi$ in our model proceeds in two phases. In the \emph{preprocessing} phase, the algorithm is given a set of \emph{prior graphs} $G_1, \ldots, G_k$ along with their exact optimal values $\opt(G_i)$, which may have been computationally expensive to obtain. Subsequently, in the \emph{query} phase, the algorithm receives a new \emph{query graph}~$H$. The algorithm must either output $\opt(H)$ if $H$ is isomorphic to some $G_i$, or correctly report that $H$ is not isomorphic to any prior graph.\footnote{\label{foot:isomorphic}Recall that graphs $G$ and $H$ are isomorphic (denoted $G \cong H$) if one graph can be obtained simply by renaming the vertices of the other; i.e., there exists a bijection $\pi\colon V(G) \to V(H)$ such that $(u,v) \in E(G)$ if and only if $(\pi(u), \pi(v)) \in E(H)$. Furthermore, if the edges carry additional \attributes (such as weights, lengths, or costs), the bijection must preserve these values; i.e., every \attribute $\alpha$ satisfies $\alpha_G(u,v) = \alpha_H(\pi(u), \pi(v))$ for all $(u,v) \in E(G)$.} Throughout, we let $n$ and $m$ denote the number of vertices and edges in $H$, respectively. We assume without loss of generality that all prior graphs $G_1, \ldots, G_k$ also contain exactly $n$ vertices and $m$ edges (because graphs of different sizes cannot be isomorphic).

For concreteness, let us consider two examples of problems $\Pi$:

\begin{enumerate}
    \item {\bf $\negName$}: Given a weighted undirected graph $G$, we want to know if $G$ contains a negative triangle, i.e., three vertices $a, b, c$ such that the sum of the weights of edges $(a,b)$, $(b,c)$, and $(c,a)$ is less than zero. This problem can be easily solved in $O(n^3)$ time. Breaking this cubic barrier would be a major breakthrough in fine-grained complexity \cite{DBLP:conf/focs/WilliamsW10,Williams2018ICM}.
    
    In our model, the algorithm is given weighted undirected graphs $G_1, \ldots, G_k$ alongside binary indicators $\opt(G_i)$ denoting whether each $G_i$ contains a negative triangle. After preprocessing this information, the algorithm is given another graph $H$ and must either return whether $H$ contains a negative triangle, or report that $H$ is not isomorphic to any $G_i$.

\item \textbf{$\cspName$ ($\csp$):} The input consists of a directed graph $G$ where each edge $e$ has a non-negative length $\ell(e)$ and a non-negative cost $c(e)$, two designated vertices $s$ and $t$, and a budget $B$. We want to find a path from $s$ to $t$ whose total cost is at most $B$ and whose total length is minimized. Let $\opt(G)$ denote the length of the optimal path.

$\csp$ is $\NP$-hard \cite{GareyJ79,hansen1980bicriterion} and has been extensively studied under various names, such as $\mathsf{Restricted}$ $\mathsf{Shortest}$ $\mathsf{Path}$ or $\mathsf{Resource}$ $\cspName$, from both the theoretical \cite{AshvinkumarBK25,Bernstein12,DBLP:conf/focs/PapadimitriouY00,DBLP:journals/orl/LorenzR01,DBLP:conf/esa/MehlhornZ00,DBLP:journals/mor/Hassin92} and applied \cite{Yinetal2024,DBLP:journals/pacmmod/WangW23,ahmadi2021fast,DBLP:journals/pvldb/0001XYL16,Storandt12,DBLP:conf/alenex/MehlhornZ01,xue2000primal,LozanoToran1992NonUniformGI,beasley1989algorithm,DBLP:journals/networks/AnejaAN83,handler1980dual} perspectives.

 In our model, the algorithm is given a budget\footnote{We assume for simplicity that the budget is the same for all input instances. Algorithms for this case can be trivially extended to handle varying budgets since, naturally, two instances are isomorphic if and only if their input graphs are isomorphic and their budgets are identical.} $B$ and preprocesses prior input graphs $G_1, \ldots, G_k$ with their respective length and cost functions, along with their precomputed solutions $\opt(G_1), \ldots, \opt(G_k)$. After preprocessing this information, the algorithm receives a query graph $H$ with its own length and cost functions, and must either return the constrained shortest path distance $\opt(H)$ or report that $H$ is not isomorphic to any $G_i$.\footnote{Following the definition in \Cref{foot:isomorphic}, an isomorphism between $G_i$ and $H$ must strictly preserve the source and target nodes, mapping $s$ to $s$ and $t$ to $t$.}

\end{enumerate}

Note that an algorithm may still correctly return the optimal value $\opt(H)$ even if no $G_i$ is isomorphic to $H$; indeed, the algorithms in this paper exhibit this behavior.

\paragraph{Can we Avoid Computing $\opt(H)$ from Scratch when $H$ is Isomorphic to some $G_i$?} 
Ideally, if $H$ is isomorphic to one of the prior graphs $G_1, \ldots, G_k$, the algorithm should compute $\opt(H)$ significantly faster than traditional worst-case bounds guarantee. Conversely, if $H$ is not isomorphic to any prior graph, the algorithm should quickly report this so that a heavy-duty fallback computation can be invoked from scratch. For instance, if a weighted query graph $H$ is isomorphic to a prior graph $G_i$ (i.e., identical up to a renaming of vertices), can we determine whether $H$ contains a negative triangle in truly subcubic time $O(n^{3-\epsilon})$? Similarly, for $\csp$, can we achieve polynomial preprocessing and query times in our model?

Trivially, every problem in our model can be solved via Graph Isomorphism ($\GI$); however, the best known algorithm for $\GI$ requires quasipolynomial time~\cite{Babai16} and is efficient only on restricted graph classes~\cite{hopcroft1974linear, luks1982isomorphism, bodlaender1990polynomial}. This is prohibitively slow for a query phase even for $\NP$-hard problems like $\csp$, and is inferior to the classical worst-case time for problems in $\cP$ like $\negName$. 
\textit{Is this new model capable of bypassing worst-case barriers without invoking computationally expensive graph isomorphism algorithms?}

\subsection{Our Results}

\paragraph{Algorithmic Results.} 
We show that our model has the power to bypass
worst-case time complexity barriers for various classes of computational problems.

The first class of problems we consider encompasses \textbf{$\NP$-hard graph optimization}. Besides the aforementioned $\cspName$ ($\csp$) problem, \emph{$\lspName{p}$} ($\ell_p$-SP) stands as another classic multi-objective challenge~\cite{DBLP:conf/focs/PapadimitriouY00,yu1998robust, DBLP:journals/siamdm/BersteinLMORWW08, DBLP:conf/soda/CarlsonMM26, DBLP:conf/icalp/MakarychevOT24, DBLP:conf/icalp/00010024}. Started by Papadimitriou and Yannakakis~\cite{DBLP:conf/focs/PapadimitriouY00} and the subject of recent progress characterizing its exact approximation landscape~\cite{DBLP:conf/soda/CarlsonMM26, DBLP:conf/icalp/MakarychevOT24, DBLP:conf/icalp/00010024}, this problem requires finding an $s$-$t$ path in a graph with $d$-dimensional edge weights such that the $\ell_p$-norm of the path's total aggregated weight vector is minimized. Closely related are \emph{$\cstName$ ($\cst$) and $\lstName{p}$ ($\ell_p$-ST)}~\cite{DBLP:journals/orl/HongCP04,xue2000primal,DBLP:conf/swat/RaviG96, marathe1998bicriteria,DBLP:journals/networks/Shogan83,DBLP:journals/cor/AggarwalAN82,GareyJ79}, where the optimization goal shifts from finding an $s$-$t$ path to computing a spanning tree.

All these problems are $\NP$-hard. In contrast, in our model, we only require near-linear preprocessing and query times to solve the problems for paths and subcubic time for trees.

\begin{theorem}[Algorithms for $\NP$-hard Problems]\label{thm:main:NP-hard Upper Bound}
There exist deterministic isomorphic-priors algorithms that solve $\cspName$ and $\lspName{p}$ using $\widetilde{O}(km)$ preprocessing time and $\widetilde{O}(m)$ query time.\footnote{$\tilde O$ hides polylogarithmic factors.} For $\cstName$ and $\lstName{p}$, there exist randomized\footnote{\label{foot:whp} The algorithms are correct with high probability. We say an event occurs \emph{with high probability} (w.h.p.) if it holds with probability at least $1 - 1/N^c$ for any arbitrarily large constant $c \ge 1$, where $N$ denotes the input size. Furthermore, our algorithms are analyzed under the standard \emph{oblivious adversary} model; that is, we assume the sequence of queries is generated independently of the algorithm's internal random choices (see, e.g., \cite{DBLP:conf/crypto/NaorY15} and \cite[Definition 3.4]{DBLP:books/crc/KatzLindell2014}).} algorithms requiring $\widetilde{O}(k n^\omega)$ preprocessing time and $\widetilde{O}(n^\omega)$ query time, where $\omega<2.37134$ is the matrix multiplication exponent. 
\end{theorem}

Next, we consider {\bf problems in $\cP$}. Beyond $\negName$, $\repPathName$ is another fundamental problem in $\cP$ that requires $n^{3-o(1)}$ time in the worst case under standard fine-grained hypotheses~\cite{DBLP:conf/focs/WilliamsW10,Williams2018ICM}. In this problem, given a shortest $s$-$t$ path $P$ in a graph $G$, the goal is to compute, for every edge $e \in P$, the length of the shortest $s$-$t$ path in $G - e$. Our Isomorphic-Priors model allows us to bypass this classical cubic barrier, achieving near-linear query time for $\repPathName$ and truly subcubic query time for $\negName$:

\begin{theorem}[Algorithms for Problems in $\cP$]\label{thm:main:Triangle Upper Bound}
There exists a deterministic isomorphic-priors algorithm that solves $\repPathName$ using $\widetilde{O}(km)$ preprocessing time and $\widetilde{O}(m)$ query time.  For $\negName$, there exists a randomized\textsuperscript{\ref{foot:whp}} isomorphic-priors algorithm requiring $\widetilde{O}(k n^\omega)$ preprocessing time and $\widetilde{O}(n^\omega)$ query time. 

\end{theorem}

Interestingly, while $\negName$ and $\repPathName$ are subcubic equivalent in the classical setting~\cite{DBLP:conf/focs/WilliamsW10, DBLP:conf/soda/AbboudGW15}, their current complexities diverge in the Isomorphic-Priors model.

Additionally, we can solve \textbf{$\maxFlowName$} with near-linear preprocessing and query times, which is slightly faster than the current worst-case complexity~\cite{DBLP:conf/focs/ChenKLPGS22,DBLP:conf/focs/Brand0PKLGSS23,DBLP:conf/stoc/BrandLLSS0W21,DBLP:conf/focs/BrandLNPSS0W20}.\\

Our third class of problems concerns the \textbf{frontier between $\cP$ and $\NP$}, namely \textbf{infinite-duration games} played on graphs, including \textit{$\parityName$}, \textit{$\meanName$}, \textit{$\energyName$}, and \textit{$\stochasticName$}. In these games, two players take turns moving a token along the edges of a directed graph to form an infinite path. For example, in an \textit{$\energyName$} (played on a bipartite graph with integer edge weights), the token starts at vertex $s$ with initial energy $e_0$. Traversing an edge adds its weight to the running total. Player A wins if they can keep the energy non-negative indefinitely; otherwise, Player B wins. We defer formal definitions to \Cref{sec:1WL}.

Originating from automata theory, reactive systems synthesis, and model checking \cite{EmersonJutla91, Thomas1997}, these games lie in $\NP \cap \coNP$ (and often $\UP \cap \coUP$ \cite{Jurdzinski98, Condon92}), making $\NP$-hardness highly unlikely. While a major breakthrough (STOC 2017 best paper) showed $\parityName$ are solvable in quasipolynomial time \cite{CaludeJKL17STOC}, harder variants ($\meanName$, $\energyName$, and $\stochasticName$) still rely on randomized subexponential simplex-like methods \cite{Halman07, BjorklundSV06}. Determining whether any of these games admit polynomial-time algorithms remains a prominent open question.

In our setting, however, all these problems can be easily solved in near-linear time.

\begin{theorem}[Algorithms for Infinite-Duration Games]\label{thm:main:Infinite Games Upper Bounds}
There exist deterministic Isomorphic-Priors algorithms that solve $\parityName$, $\meanName$, $\energyName$, and $\stochasticName$ using $\widetilde{O}(km)$ preprocessing time and $\widetilde{O}(m)$ query time. 

\end{theorem}

More broadly, a similar result holds for every infinite-duration game with a well-defined value (value-determined game; see \cref{def:value-determined-game})

In all the results above, we compute only the optimal \emph{values}. If the corresponding optimal \emph{solutions} are provided with the priors, then our algorithms can also compute an optimal solution for the query graph.

\paragraph{Conditional Lower Bounds.} We observe that many known reductions can be naturally adapted to our framework to establish conditional lower bounds, which hold even when the database contains only a single prior graph ($k=1$). Specifically, we show that for many $\NP$-hard problems, achieving both polynomial preprocessing and query times would imply the major algorithmic breakthrough of $\GI \in \cP$. Moreover, barring a breakthrough proving $\GI \in \coma$, these lower bounds hold even in a stronger model allowing unbounded preprocessing time.

\begin{theorem}[Informal]\label{thm:intro:hardness}
Unless $\GI \in  \cP$, none of the following problems can be solved in the Isomorphic-Priors model using simultaneously polynomial preprocessing time and polynomial query time, even when restricted to a single prior graph ($k=1$): 
$\kSAT{3}$\footnote{For $\SAT$, we say that two CNF formulas are isomorphic if one can obtained from the other by changing the names of variables and clauses.}, $\MaxkSAT{2}$, $\Clique$, $\COL$, $\kCOL{3}$, $\CliqueCover$, $\EXA$, $\SetCover$, $\SetPacking$, $\HIT$, $\VC$, $\MaxCut$, $\STEINER$, $\FeedbackArcSet$, $\FeedbackVertexSet$,
$\DirHam$, $\UndirHam$, $\DMatching$, $\LPZEROONE$, $\kMedian$, $\MetricTSP$, $\TSP$ within any approximation factor, and $(2-\epsilon)$-approximate $\kCenter$.

Unless $\GI \in \coma$, these lower bounds hold even against algorithms with unbounded preprocessing time, provided they are restricted to a polynomial-size data structure for answering queries.
\end{theorem}

The list above includes almost all of Karp's 21 $\NP$-complete problems \cite{DBLP:conf/coco/Karp72}, except a few problems on numbers that are trivial in our model such as Knapsack and Partition problems.

\Cref{thm:intro:hardness} already includes hardness of approximation results for the $\tspName$ ($\TSP$) and $\kCenter$. If the PCP theorem can be  transferred to argue that $\MaxkSAT{3}$ cannot be approximated within an arbitrary constant factor in our model, we can additionally argue that neither a $(1-1/e+\epsilon)$-approximation for $\mathsf{Max\text{-}}k\mathsf{\text{-}Coverage}$ nor a $(1+2/e-\epsilon)$-approximation for $\kMedian$ can be achieved using simultaneously polynomial preprocessing time and polynomial query time (see \Cref{sec:Hardness}). We view the deeper exploration of the PCP theorem, inapproximability bounds, and improved approximation algorithms within our model as an exciting future research direction. 

Proving hardness in our model requires \emph{isomorphism-preserving reductions} (\Cref{sec:isoreduction}): a reduction $f$ must map an instance $G$ to a target instance $f(G)$ while guaranteeing that if $G \cong H$, then $f(G) \cong f(H)$. This is related to the existing notion of \emph{strong isomorphism reductions}~\cite{DBLP:journals/jsyml/BussCFFM11}, which demands the biconditional $G \cong H \iff f(G) \cong f(H)$.
Some standard reductions  preserve isomorphism, such as the reduction from Graph Isomorphism to finding a maximum clique in a modular product graph~\cite{DBLP:journals/sigact/Kozen78} (which immediately implies the hardness of $\maxCliqueName$/$\maxISName$ in our model).
However, many classical reductions are not isomorphism-preserving, such as the standard transformation from $\SAT$ to $\kSAT{3}$ or the PCP theorem (e.g., the assignment tester in \cite{Dinur2007PCP}). Determining whether the rich landscape of classical reductions can be systematically adapted into isomorphism-preserving variants remains a formidable research program.

\paragraph{Recent Developments.}
Since we completed the initial draft of this paper on April 1, 2026,
the following related developments have emerged.

First, in the $\mathsf{All}\text{-}\mathsf{Pairs}\text{ }\mathsf{Shortest}\text{ }\mathsf{Paths}$ problem ($\APSP$), the input is an
edge-weighted graph, and the task is to compute all pairwise
distances. Fox~\cite{Fox26} showed that $\APSP$ can be solved in the
Isomorphic-Priors model by a randomized algorithm with
$\widetilde{O}(kn^\omega)$ preprocessing time and
$\widetilde{O}(n^\omega)$ query time.

Neuen independently observed that $\APSP$, $\negName$, and $\cstName$ admit a unified randomized approach with
$\widetilde{O}(kn^\omega)$ query time and no nontrivial preprocessing
beyond storing the prior instances and their answers. Although the
query bound retains a linear dependence on $k$, the approach handles
all three problems using the same comparison procedure, based on
only $O(\log n)$ rounds of $2$-WL which can be performed in time $\widetilde{O}(n^\omega)$ using Skresanov's randomized implementation of $2$-WL~\cite{Skresanov21}.\footnote{
More precisely, Skresanov's randomized implementation of
$2$-WL~\cite{Skresanov21} allows $r$ rounds of refinement to be
computed in $\widetilde{O}(rn^\omega)$ time, with high probability.
The key observation is that $O(\log n)$ rounds suffice to transfer
the required answers between instances that are not distinguished by the refinement.
For Negative Triangle and Constrained Minimum
Spanning Tree, such instances have the same decision answer and
optimal value, respectively. For $\APSP$, ordered pairs receiving
the same color have the same distance, allowing the stored distances
to be transferred to the query graph.
}

\subsection{Organization}

We provide an overview of our arguments in \Cref{sec:overview}, which we believe already gives readers the key ideas and enough intuition to reconstruct the full proofs. We discuss open problems, related work and formally define the model and necessary notations in \Cref{sec:open,sec:related,sec:model}. In \Cref{sec:PIT}, we present algorithms based on randomized polynomial evaluation for $\negName$, $\cstName$ and $\lstName{p}$. Then, we switch to Weisfeiler-Leman techniques: \Cref{sec:1WL} presents an elementary argument based on WL coloring for $\cspName$ and $\energyName$ and \Cref{sec:WL-for-Experts} present the same results based on connections to graph homomorphism counting, as well as a deterministic algorithm for $\cstName$ and $\lstName{p}$. Our lower bounds are presented in \Cref{sec:Hardness}. We conclude in \Cref{sec:conclusion} by presenting the proofs of our main theorems, referring to the relevant results established throughout the paper, and by discussing some hypotheses within our model and their consequences. \Cref{app:1wl,app:hard} provide additional algorithmic results based on WL coloring and additional lower bound results.

\section{Overview}\label{sec:overview}

In this section, we outline the core technical ideas underlying our main results. We emphasize again that our primary contribution lies not in the invention of these individual techniques, but in their synthesis across different fields to demonstrate how worst-case computational barriers can be bypassed within the Isomorphic-Priors framework, alongside the establishment of tight conditional lower bounds.

We present three main ideas below: (1) algorithms via randomized polynomial evaluation, (2) algorithms via Weisfeiler-Leman (WL) color refinement, and (3) conditional lower bounds via isomorphism-preserving reductions. We strive to present these concepts without assuming a specialized background, with the goal of enabling a general theoretical computer science audience to reconstruct our proofs after reading this section.

\subsection{Algorithms via Randomized Polynomial Evaluation (Details in \Cref{sec:algorithms via polynomials})}

Let us consider $\negName$ as a guiding example. We define the \emph{triangle polynomial} of a weighted undirected graph $G$ as
\[
P_{\Delta}(G) \ \coloneqq \sum_{w_1 \le w_2 \le w_3} c_{w_1, w_2, w_3}(G) \, x_{w_1} x_{w_2} x_{w_3}
\]
where $c_{w_1, w_2, w_3}(G)$ is the number of triangles in $G$ whose edges possess the weight multiset $\{w_1, w_2, w_3\}$. For example, the first term of $P_{\Delta}(G)$ in \Cref{fig:intro:NTD} is $2x_{-2}^2 x_1$, because there are two triangles in $G$ with edge weights $\{-2, -2, 1\}$. 

\begin{figure}
    \centering
    \includegraphics[width=1.0\linewidth]{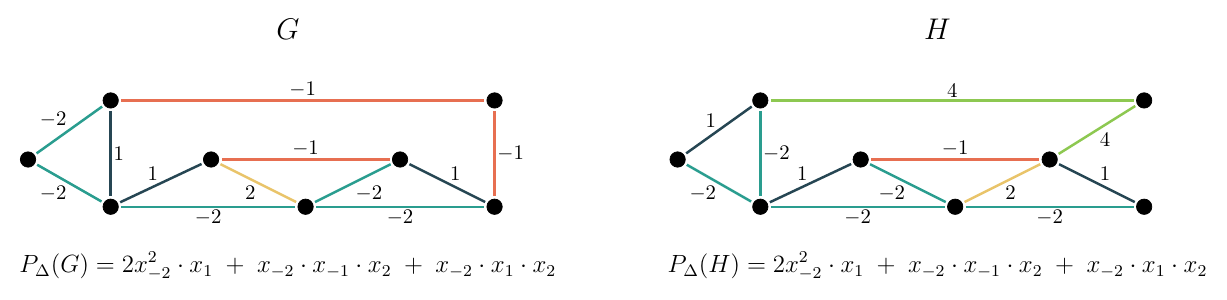}
    \caption{Examples of triangle polynomials in weighted graphs. Edges with the same color have the same weight. Note that $G$ and $H$ are not isomorphic (and in fact not even WL-equivalent).}
    \label{fig:intro:NTD}
\end{figure}

Obviously, if $P_{\Delta}(G) \equiv P_{\Delta}(H)$ (i.e., they are identically equal as polynomials), then $G$ contains a negative triangle if and only if $H$ does. (Note that this does not imply $G \cong H$, as exemplified in \Cref{fig:intro:NTD}). Our goal is therefore to efficiently identify a prior graph $G_i$ such that $P_{\Delta}(G_i) \equiv P_{\Delta}(H)$ or rule out the existence of such $G_i$.

To achieve this without explicitly constructing the polynomials, we rely on randomized Polynomial Identity Testing. By the Schwartz-Zippel lemma, since $P_{\Delta}$ is a multivariate polynomial of degree $3$, evaluating it at a random point drawn uniformly from a sufficiently large finite field $\mathbb{F}$ distinguishes non-identical polynomials with high probability. Specifically, we assign each possible variable $x_w$ a random value $r_w \in \mathbb{F}$. 

Given this random assignment, we need to evaluate $P_{\Delta}(G)$ at $\mathbf{r} = (r_{w_1}, \dots, r_{w_m})$ in $\widetilde{O}(n^\omega)$ time. To this end, we construct an $n \times n$ matrix $A$ where $A_{u,v} = r_{w(u,v)}$ if $(u,v) \in E(G)$, and $0$ otherwise. The value of the polynomial evaluated at $\mathbf{r}$ is simply $\frac{1}{6} \text{Tr}(A^3)$ which is equal to the sum of weighted triangles in the graph with adjacency matrix $A$ and can be computed via fast matrix multiplication in $\widetilde{O}(n^\omega)$ query time.

To summarize, during the preprocessing phase, we evaluate $P_{\Delta}(G_i)(\mathbf{r})$ for each prior $G_i$ and store this field element as its invariant representation. This requires $\widetilde{O}(k n^\omega)$ total preprocessing time. When we receive a query graph $H$, we evaluate $P_{\Delta}(H)(\mathbf{r})$ using the exact same random assignment $\mathbf{r}$ in $\widetilde{O}(n^\omega)$ time. We then simply find a prior graph $G_i$ that evaluates to the identical scalar value and return its precomputed solution $\opt(G_i)$. This takes  $\tilde O(n^\omega)$ query time.\footnote{There are some polylogarithmic factors needed to guarantee high success rate and for searching for the desired graph $G_i$.}
We can guarantee that, with high probability\textsuperscript{\ref{foot:whp}}, the random point $\mathbf{r}$ chosen during the preprocessing phase is such that, for every $i$, $P_{\Delta}(G_i)(\mathbf{r})=P_{\Delta}(H)(\mathbf{r})$ if and only if $P_\Delta(G_i) \equiv P_\Delta(H)$. Thus, the answer given in the query phase is correct with high probability. 

Additionally, if the actual negative triangle $T_i$ in $G_i$ is provided, we can find a negative triangle in $H$ by reducing the problem to unweighted triangle detection, which takes $O(n^\omega)$ time~\cite{DBLP:conf/stoc/Itai77}. (Roughly speaking, this is done by detecting an unweighted triangle in a tripartite graph consisting only of edges whose weights match those of $T_i$.)

To solve the Constrained Spanning Tree ($\cst$) problem, we similarly define the \emph{spanning tree polynomial} of a graph $G$ as 
\[
P_{ST}(G) \coloneqq \sum_{w_1 \le \dots \le w_{n-1}} c_{w_1, \dots, w_{n-1}}(G) \prod_{i=1}^{n-1} x_{w_i}
\]
where $c_{w_1, \dots, w_{n-1}}(G)$ is the number of spanning trees with the weight multiset $\{w_1, \dots, w_{n-1}\}$.

We leverage the fact that this multivariate polynomial can be efficiently evaluated at a random point $r$ by invoking Kirchhoff's Matrix Tree Theorem. Specifically, this evaluation reduces to computing the determinant of a randomized Laplacian matrix, which takes $\widetilde{O}(n^\omega)$ time via fast matrix multiplication. If the optimal tree $T_i$ in $G_i$ is provided, we can reconstruct a corresponding optimal tree $T$ in $H$ via matroid intersection in polynomial time~\cite{schrijver2003}, and specifically in $\widetilde{O}(m\sqrt{n})$ time via recent advancements~\cite{DBLP:conf/stoc/BlikstadMNT23,DBLP:conf/focs/ChakrabartyLS0W19,DBLP:conf/icalp/GabowS85}.
Full details and the formal algebraic constructions are deferred to \Cref{sec:CST}.

\subsection{Algorithms via Weisfeiler-Leman (WL) Coloring (Details in \Cref{sec:algorihtms via WL,sec:WL-for-Experts})}

We sketch an elementary argument here, designed to be accessible to readers unfamiliar with WL coloring, with full details deferred to \Cref{sec:algorihtms via WL}. An alternative, more concise proof based on a more general perspective
is provided in \Cref{sec:WL-for-Experts}.

We use $\cspName$ ($\csp$) as our guiding example. In the preprocessing phase, we color the vertices of each prior graph $G_i$ using the Weisfeiler-Leman (WL) color refinement process, which runs in time $O((n+m)\log n)$~\cite{CardonC82,paige1987three, berkholz2017tight}. 
An essential property of the resulting coloring is that \emph{every two vertices $u$ and $v$ that receive the same color will have identically colored neighborhoods.}

For example, in \Cref{fig:intro:WL coloring}, every pink vertex (labeled $p$) is adjacent to the target $t$ via an edge of \attributes (length and cost) $(5,2)$, to two cyan vertices (labeled $c$) via edges of \attributes $(2,4)$, and to one green vertex (labeled $g$) via an edge of \attributes $(3,3)$. Note that we initialize the process such that the source $s$ and target $t$ are assigned uniquely identifiable colors (red and blue in the example).
The edge colors shown in the figure represent the edge attributes (and are not a result of the WL coloring).

\begin{figure}
    \centering
    \includegraphics[scale=1.4]{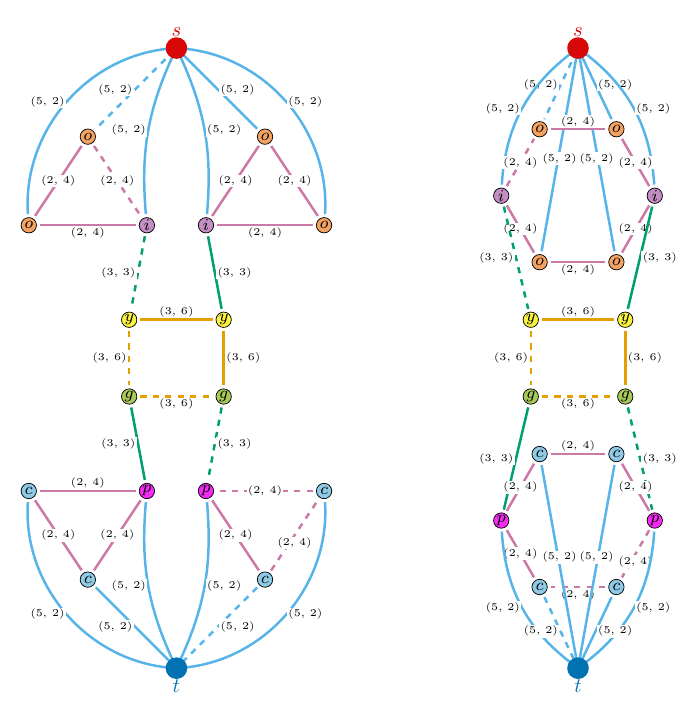}
    \caption{Two WL-equivalent but non-isomorphic graphs. Excluding $s$ and $t$, each graph contains the exact same color distribution: 4 orange ($o$), 2 indigo ($i$), 2 yellow ($y$), 2 green ($g$), 2 pink ($p$), and 4 cyan ($c$) vertices. The dashed edges in the left graph illustrate an $s$-$t$ walk with the edge-\attribute sequence $((5,2), (2,4), (3,3), (3,6), (3,6), (3,3), (2,4), (2,4), (5,2))$. As guaranteed by \Cref{claim:wl:equiv:color}, a corresponding walk with the exact same \attribute sequence must exist in the right graph, which is realized by the dashed edges shown there.}
    \label{fig:intro:WL coloring}
\end{figure}

Building on this property, we can easily establish the following claim (see the caption of \Cref{fig:intro:WL coloring} for a concrete example).

\begin{claim}[Informal]\label{claim:wl:equiv:color}
Let $G$ and $H$ be two \emph{WL-equivalent} graphs, meaning that for every color $c$, both graphs contain exactly the same number of vertices assigned color $c$. If there exists an $s$-$t$ walk (i.e., an $s$-$t$ path that is not necessarily simple) $W = (s=v_0, v_1, \dots, v_p=t)$ in $G$, then there exists a corresponding walk $W' = (s=v'_0, v'_1, \dots, v'_p=t)$ in $H$ such that, for every step $i \in \{0,\dots,p-1\}$, the \attributes (length and cost) of the edges $(v_i, v_{i+1})$ and $(v'_i, v'_{i+1})$ are identical.
\end{claim}

\begin{proof}[Proof Idea of \cref{claim:wl:equiv:color} ]
Imagine placing a token at the source vertex $s$ in both graphs $G$ and $H$.
At each step, we move the token in $G$ along the walk $W$ from a vertex of color $c$ to a neighbor of color $c'$ via an edge with \attribute $\att$.
The definition of WL-equivalence guarantees that the token in $H$ (which currently resides on a vertex of the exact same color $c$) has an identically \attributed edge to a neighbor of color $c'$.
By matching these moves step-by-step, we successfully build the corresponding walk $W'$ in $H$.
\end{proof}

As an immediate consequence of the claim above, every two WL-equivalent graphs $G$ and $H$ yield the exact optimal value, i.e., $\opt(G) = \opt(H)$.\footnote{To establish this, we rely on the standard fact that any $s$-$t$ walk can be transformed into a simple $s$-$t$ path by iteratively removing cycles, which does not increase the total cost or length of the walk since edge lengths and costs are non-negative.}
Naturally, during the query phase, our algorithm simply computes the WL coloring of the query graph $H$, searches for a prior graph $G_i$ that is WL-equivalent to $H$, and returns $\opt(G_i)$. This process achieves the desired near-linear query time. Additionally, if each prior graph $G_i$ is equipped with its optimal path $P_i$, we can efficiently reconstruct the corresponding optimal path in $H$ using the token-tracking procedure described in the proof sketch. 

Other path-centric optimization problems, such as $\lspName{p}$, $\repPathName$, and Infinite-Duration Games, can be solved using a similar approach.
Furthermore, we note that the Constrained Spanning Tree ($\cst$) problem can be solved deterministically in $O(n^3 \log n)$ query time via 2-WL, a higher-dimensional variant of the WL coloring; see \Cref{sec:WL-for-Experts}.

\subsection{Lower Bounds (Details in \Cref{sec:Hardness})}

\paragraph{Graph Isomorphism with Promise.}
Our starting point is the $\giName$ ($\GI$) problem itself, formulated within our model under a single prior instance ($k=1$). In the preprocessing phase, the algorithm receives a prior pair of graphs $(G, H)$ along with $\opt(G, H)$, which indicates whether $G \cong H$. During the query phase, the algorithm receives a new pair $(G', H')$ and must either (i) decide whether $G' \cong H'$, or (ii) correctly declare that $(G, H) \not\cong (G', H')$ (i.e., $G' \not\cong G$ or $H' \not\cong H$).

We argue that this problem cannot be solved with polynomial preprocessing and query times unless $\GI \in \cP$, even under the promise that $G' \cong G\cong H$. 
Let $(X, Y)$ be an arbitrary worst-case instance of $\GI$. We define the prior instance simply as $G=H=X$, meaning $\opt(G,H)$ trivially indicates that $G \cong H$. We then present the query instance as $G'=X$ and $H'=Y$. Observe that every algorithm capable of answering this query in polynomial preprocessing and query time would immediately decide whether $X \cong Y$ from scratch without prior, implying $\GI \in \cP$.

\paragraph{From $\GI$ with Promise to $\maxISName$.}
Now we reduce $\GI$ with the $G' \cong G\cong H$ promise from above to $\maxISName$ (equivalently, $\maxCliqueName$) with a single prior.
(This reduction already exists in \cite{DBLP:journals/sigact/Kozen78}.) 
Assume for contradiction that there exists an algorithm $A$ that takes a prior graph $Z$ alongside the size of its maximum independent set\footnote{Recall that an independent set of a graph is a set of mutually non-adjacent vertices, and $\alpha(Z)$ denotes the size of the maximum independent set in $Z$.}, denoted $\alpha(Z)$, preprocesses this information in polynomial time, and subsequently, for every query graph $Z'$, either computes $\alpha(Z')$ or declares that $Z' \not\cong Z$ in polynomial time. We use $A$ to solve $\GI$ under our single-prior setting as follows.

{\em Constructing Prior Graph $Z$:}
Upon receiving the prior pair $(G,H)$ for $\GI$, our reduction constructs a prior graph $Z$ as follows.
For every $u \in V(G)$ and $v \in V(H)$, we create a vertex $z_{u,v}$ in $Z$. We add an edge between distinct vertices $z_{u,v}$ and $z_{u',v'}$ if the assignments $u \mapsto v$ and $u' \mapsto v'$ violate the isomorphism condition; i.e., we add an edge if either: (i) $(u=u')$ $\oplus$\footnote{We use the logical xor to avoid self loops (when both $u=u'$ and $v=v'$).} $(v=v')$,
or (ii) the adjacency is not preserved, meaning $(u,u') \in E(G)$ but $(v,v') \notin E(H)$, or vice versa.
Observe that $\alpha(Z) = |V(G)|$ if and only if $G \cong H$.\footnote{To see this, note that condition (i) ensures every independent set selects at most one vertex for each $u \in V(G)$ and $v \in V(H)$, bounding its size by $|V(G)|$. If $G \cong H$ via an isomorphism $f$, the set $\{z_{u, f(u)} \mid u \in V(G)\}$ forms an independent set of size exactly $|V(G)|$. Conversely, every independent set of size $|V(G)|$ defines a bijection, and condition (ii) guarantees this bijection preserves all edges and non-edges, forming a valid isomorphism.} 
By the  $ G\cong H$ promise, we know that $\alpha(Z) = |V(G)|$. 
Thus, we pass $Z$ and its solution $\alpha(Z) = |V(G)|$ to algorithm $A$ for preprocessing.

{\em Constructing Query Graph $Z'$:}
Upon receiving a query pair $(G', H')$ for $\GI$, we construct a graph $Z'$ using the exact same procedure; hence, $\alpha(Z') = |V(G')|$ if and only if $G' \cong H'$. We then query algorithm $A$ with $Z'$. If $A$ returns $\alpha(Z') = |V(G')|$, we output $G' \cong H'$; otherwise, we declare that $G' \not\cong H'$.

We claim that this is always correct. If $(G,H) \cong (G', H')$, our construction guarantees that $Z \cong Z'$ (meaning our reduction is \emph{isomorphism-preserving}). Consequently, $A$ must output $\alpha(Z') = |V(G')|$. Conversely, if $(G,H) \not\cong (G', H')$, the promise that $G' \cong G \cong H$ forces the conclusion that $G' \not\cong H'$.\footnote{To see this, assume for contradiction that $G' \cong H'$. Due to the $G' \cong G \cong H$ promise, we have $H' \cong G' \cong G\cong H$. We would therefore have $(G',H') \cong (G,H)$, contradicting our initial assumption.}
Thus, $A$ must either return $\alpha(Z') < |V(G')|$ or declare $Z \not\cong Z'$, both of which trigger our reduction to correctly declare $G' \not\cong H'$.

Finally, observe the crucial role of the isomorphism-preserving property here. If our reduction was not isomorphism-preserving, the condition $Z \not\cong Z'$ could arise regardless of whether $(G,H) \cong (G',H')$ or $(G,H) \not\cong (G',H')$. In such a scenario, if algorithm $A$ outputs $Z \not\cong Z'$, we would be unable to distinguish whether $(G,H) \cong (G',H')$ (which dictates we must return $G' \cong H'$ due to the $G'\cong G\cong H$ promise) or if $(G,H) \not\cong (G',H')$.

Many NP-hardness reductions are naturally isomorphism-preserving, meaning many standard lower bounds carry over to our framework. A prime exception is Karp's standard reduction from $\SAT$ to $\kSAT{3}$ \cite{DBLP:conf/coco/Karp72}, which fails to preserve isomorphism. We circumvent this by routing our reduction through graph coloring. We begin with Karp's $\SAT$ to $\COL$ reduction, which is isomorphism-preserving. In the second step, we employ Lovász's reduction from $\COL$ to $\kCOL{3}$ \cite{lovasz_graph}.\footnote{We refer the interested reader to \url{https://users.encs.concordia.ca/~chvatal/notes/color.html}} Finally, a standard mapping from $\kCOL{3}$ back to $\kSAT{3}$ completes the chain, yielding the desired isomorphism-preserving reduction from $\SAT$ to $\kSAT{3}$. \Cref{fig:reductions} provides an overview of our reduction chains.

\section{Open Problems}\label{sec:open}

Our results raise several questions about the power and limitations of the Isomorphic-Priors model.
We begin with four central directions: which fine-grained complexity barriers persist in our model,
an Isomorphic-Priors $\kClique{k}$ hypothesis with consequences for
$\giName$, the approximability of hard optimization
problems, and extensions beyond exact isomorphism. We then discuss
structural questions about the sources of tractability and the
broader scope of the framework.

\paragraph{Fine-Grained Complexity.}
Many questions remain for problems in
$\cP$. For example, our algorithm for $\negName$ uses randomness
to achieve $\widetilde{O}(kn^\omega)$ preprocessing time and
$\widetilde{O}(n^\omega)$ query time. Can the same bounds be achieved
deterministically? Can we improve them to near-linear time? 

More broadly, the standard weighted versions of $\negName$,
$\repPathName$, $\mathsf{All}$-$\mathsf{Pairs}$ $\mathsf{Shortest}$ $\mathsf{Paths}$ ($\APSP$), $\mathsf{Minimum}$-$\mathsf{Weight}$
$\mathsf{Cycle}$, $\mathsf{Second}$ $\mathsf{Shortest}$ $\mathsf{Simple}$ $\mathsf{Path}$, and $\mathsf{Radius}$ are equivalent under
subcubic reductions in the classical setting without
priors~\cite{DBLP:conf/focs/WilliamsW10,DBLP:conf/soda/AbboudGW15}.
Yet our algorithms achieve different query bounds for $\negName$ and $\repPathName$, and it remains unclear whether the
classical equivalences extend to our model. Doing so requires
isomorphism-preserving reductions. Can we solve all the problems in this
equivalence class in near-linear query time?

Beyond the $\APSP$ equivalence class, it is natural to ask whether
fine-grained barriers based on OV, Hitting Set, and combinatorial
Triangle Detection persist in the Isomorphic-Priors model, and
whether such lower bounds can be transferred via
isomorphism-preserving fine-grained reductions.

\paragraph{Isomorphic-Priors $\kClique{k}$ Hypothesis.}
For every fixed integer $k\geq 3$ divisible by three, fast matrix
multiplication solves $\kClique{k}$ in $n^{\omega k/3+o(1)}$ time.
The standard $\kClique{k}$ Hypothesis rules out any constant
improvement in the exponent~\cite{Williams2018ICM}. We hypothesize
that this barrier persists with isomorphic priors.

\begin{restatable}{hypothesis}{kclique}
(Isomorphic-Priors $\kClique{k}$ Hypothesis)\label{conj:iso-k-clique}
For every fixed integer $k\geq 3$ and constant $\epsilon>0$, no
isomorphic-priors algorithm can detect a $k$-clique in an
$n$-vertex graph with both preprocessing and query times
$O(n^{\omega k/3-\epsilon})$, even with $O(1)$ prior graphs.
\end{restatable}

\textbf{This hypothesis implies $\GI\notin\cP$:}
A polynomial-time algorithm for $\giName$ would
allow us to compare a query graph against $O(1)$ priors and retrieve
a matching prior's answer, or report that none exists, in
$n^{O(1)}$ time. 
This would contradict the hypothesis for sufficiently large fixed $k$.
We provide further details and discuss weaker variants in
\Cref{sec:hypothesis}. An important open direction is to refute
these hypotheses or establish them under more standard complexity
assumptions.

\paragraph{Approximability.} Although we establish
conditional lower bounds for exact computation and for some
approximation problems, much of the approximation landscape remains
open. For example, can we achieve a $0.876$-approximation for $\MaxkSAT{3}$,
a $0.879$-approximation for $\MaxCut$, a $1.99$-approximation for $\vcName$, or an $n^{0.99}$-approximation for $\mathsf{Max}$ $\mathsf{Clique}$, with polynomial
preprocessing and query times? Can isomorphic priors also enable
a $1.385$-approximation for $\STEINER$ or a $4/3$-approximation for
metric $\TSP$? These questions ask whether prior computations can help
us surpass classical hardness thresholds or achieve approximation
guarantees that remain elusive without priors.
Conversely, which
classical inapproximability bounds persist in our setting? In
particular, can PCP constructions and gap-preserving reductions be
adapted to establish such bounds while preserving the isomorphisms
required by our framework?

\paragraph{Beyond Exact Isomorphism.}
A natural direction is to relax the strict isomorphism requirement.
Our invariant-representation approach already allows exact answers
to be transferred between some non-isomorphic graphs. Can this
principle be extended to graphs whose representations are similar,
rather than identical? For example, what happen if $H$ is isomorphic to $G'_i$ obtained $G_i$ by small perturbations of its edge weights or adding/deleting a few edges?
Is there a natural, efficiently
computable distance between Weisfeiler--Leman (WL)
histograms~\cite{DBLP:journals/jmlr/ShervashidzeSLMB11} that yields
useful bounds on the difference between $\mathit{OPT}(H)$ and
$\mathit{OPT}(G_i)$?  
Under what assumptions could such bounds lead
to approximation guarantees for $\cspName$ with
polynomial query time? Since small perturbations can change
feasibility or cause large changes in the optimal value, meaningful
guarantees may require additional assumptions on the instances or
controlled relaxations of the constraints. Establishing such results
would provide a principled way to reuse prior answers on new
instances, with approximation guarantees tied to their similarity
to the priors.

A complementary approach is to replace isomorphism with other
structural relations, such as attribute-preserving homomorphisms
or subgraph embeddings, and ask what information these relations
allow us to transfer. For $\negName$, a weight-preserving
homomorphism of a prior graph $G_i$ into a query graph $H$ transfers
a positive answer, whereas a weight-preserving homomorphism from
$H$ to $G_i$ transfers a negative answer. Can we efficiently
identify and exploit such relations without being given the
corresponding maps? What additional restrictions would allow
these one-sided implications to yield fast query algorithms?
More generally, which problem-specific relations support exact
answers, approximation guarantees, or useful one-sided conclusions
beyond the isomorphism classes of the priors?

\paragraph{Isomorphism-Preserving Algorithms.}
Our hardness results rely on isomorphism-preserving reductions.
This raises a broader question: how much computational overhead is
required to make basic algorithmic constructions consistent across
isomorphic inputs? For example, a spanning tree of a connected graph
can be computed in $O(n+m)$ time. Can we efficiently select such a
tree while preserving isomorphism? Specifically, we seek an algorithm
$A$ that returns a spanning tree $T_G=A(G)$ of each connected graph
$G$, such that whenever $G\cong H$, there exists an isomorphism
$\varphi\colon G\to H$ satisfying $\varphi(T_G)=T_H$.
Equivalently, we require
\[
    G\cong H
    \quad\Longrightarrow\quad
    (G,T_G)\cong(H,T_H),
\]
where an isomorphism of pairs must preserve both the graph edges
and the distinguished spanning-tree edges. Can such a selection
be computed in polynomial, or even near-linear, time?
More generally, which efficiently solvable search problems admit
efficient solution-selection procedures that are consistent in
this sense, and how does their complexity relate to graph
canonization? We can ask similar questions for other problems such as shortest paths, minimum cut,  maximal and maximum matching, and $\Delta+1$ coloring. 

\newpage

\paragraph{Beyond $\NP$.}
Which natural problems complete for classes such as
$\mathrm{PSPACE}$, $\mathrm{EXPTIME}$, or $\mathrm{NEXPTIME}$
admit polynomial preprocessing and query times in the
Isomorphic-Priors model?\footnote{NFA Universality and Existential Length Universality
for NFAs are $\mathrm{PSPACE}$-complete and
$\mathrm{NEXPTIME}$-complete, respectively
\cite{DBLP:conf/stacs/GawrychowskiLRS20}.
We believe that both admit polynomial preprocessing and query
times in our setting by adapting our 1-WL techniques, under
isomorphisms that preserve transition labels and initial and
accepting states. However, a formal treatment of these extensions
is beyond the scope of this paper.}
How does this question extend to exact counting problems complete
for $\#\cP$, when the exact counts are supplied with the
priors? More generally, which structural properties allow answers
to classically hard problems to be efficiently reused, and which
problems admit conditional lower bounds through
isomorphism-preserving reductions? Understanding this landscape
would clarify how the complexity of reusing prior answers differs
from that of computing them from scratch.

\paragraph{The Power of Invariant Representations.}
Our algorithmic results rely on efficiently computable invariant representations (e.g., WL coloring or spanning tree polynomial) whose equality guarantees equality of
optimal values. Does this approach characterize polynomial-time
solvability in the Isomorphic-Priors model? Specifically, if a
problem admits polynomial preprocessing and query
times, must there exist a polynomial-time computable function $f$
such that
\[
    G\cong H
    \ \Longrightarrow\
    f(G)=f(H)
    \ \Longrightarrow\
    \mathit{OPT}(G)=\mathit{OPT}(H)?
\]
Here, $f$ must be computable from the input graph alone, without
access to its optimal value or to a database of priors.
Alternatively, can jointly processing a query and the supplied
prior answers yield algorithms that cannot be captured by such
representations?

\paragraph{Characterizing Tractability.}
Our algorithmic results for $\NP$-hard problems (consider for example $\cspName$ or $\cstName$) exploit
numerical edge attributes, whereas problems on unweighted
graphs such as $\vcName$
and $\Ham$ remain conditionally hard. To what extent do these attributes account for the
difference? Can we identify natural $\NP$-hard problems on
unweighted undirected graphs, with no additional vertex
or edge attributes, that admit polynomial preprocessing and query
times? Here, we seek positive
results beyond those that follow immediately from efficient
isomorphism testing on restricted graph classes.

$\cspName$ and $\cstName$ provide
examples of weakly $\NP$-hard problems that become tractable
in our model. 
Do all problems admitting
pseudopolynomial algorithms admit polynomial preprocessing and
query times with isomorphic priors? 
Conversely, are there natural weakly $\NP$-hard problems that
remain conditionally intractable in our model?
Understanding these questions would help distinguish the roles
of numerical information and graph structure in enabling the
reuse of prior computations.

\paragraph{From Values to Solutions.}\label{par:value-vs-solution}
Our algorithms can reconstruct optimal solutions when suitable
optimal witnesses are supplied with the priors, but these
reconstruction arguments are problem-specific. Is there a general
principle that converts efficient value-retrieval algorithms into
efficient solution-reconstruction algorithms in our model?
Standard decision-to-search reductions do not immediately suffice:
the modified instances they generate need not be isomorphic to
any prior, so the algorithm may abstain on them. Which properties
of the witnesses or invariant representations permit efficient
reconstruction? Can we identify a natural separation between
retrieving optimal values and reconstructing optimal solutions,
even when optimal prior solutions are supplied?

\paragraph{More Concrete Problems.}
Can makespan scheduling of unit-length jobs with precedence
constraints on a fixed number of identical machines~\cite{DBLP:conf/soda/NederlofSW25}, or Rabin games~\cite{DBLP:conf/tacas/MajumdarST24}, be solved with polynomial preprocessing
and query times in the Isomorphic-Priors model?
Can we obtain near-linear preprocessing time per prior instance
and near-linear query time for global minimum vertex
cuts in weighted directed graphs~\cite{chuzhoy2026faster,
DBLP:journals/jacm/LiNPSY25}, or for generalized B\"uchi games~\cite{DBLP:conf/mfcs/ChatterjeeDHL16}?\footnote{For strict definitions and discussions of Rabin games and generalized B\"uchi games in the Isomorphic-Priors model, see \cref{app:hardness-graphs}.}

\section{Related Works}\label{sec:related}

\paragraph{A Possible Connection to Graph Machine Learning.}
While machine learning (ML) is not a focus of this work, we briefly
discuss a conceptual connection between our model and graph learning.
This discussion is exploratory: our results have formal implications
for abstract isomorphism-invariant predictors, but whether these
implications can inform the design or training of practical ML models
remains to be established.

Consider an ML model $M$ that takes a graph $G$ as input and outputs some value $M(G)$.
We call such a model \emph{isomorphism-invariant} if, for every pair of isomorphic
input graphs $G\cong H$, we have $M(G)=M(H)$. Isomorphism-invariance is a natural
requirement for graph-level prediction tasks in which vertex identifiers carry no semantic
meaning, and is a standard design principle in graph representation learning
\cite{DBLP:conf/iclr/MaronBSL19,DBLP:conf/nips/KerivenP19}.
Examples of isomorphism-invariant prediction targets include structural graph quantities
such as triangle counts, diameter, and chromatic number, as well as the optimal values of
combinatorial optimization problems such as Maximum Cut or $\tspName$.
Many commonly used graph-learning architectures satisfy this symmetry by design:
message-passing Graph Neural Networks (GNNs), when equipped with a permutation-invariant
graph-level readout, yield isomorphism-invariant predictors
\cite{xu2018how,morris2019weisfeiler};
higher-order invariant GNN architectures have been developed as well
\cite{DBLP:conf/iclr/MaronBSL19,DBLP:conf/nips/KerivenP19}.
Several Graph Transformer architectures similarly incorporate isomorphism-respecting
structural or positional encodings
\cite{DBLP:conf/nips/KreuzerBHLT21}.
The following observation gives a concrete connection between our model
and the construction of isomorphism-invariant graph predictors.

\begin{observation}
Let $P$ be an isomorphism-invariant function on graph instances that
does not admit an algorithm with preprocessing and query times
polynomial in $n$ and $k$ in the Isomorphic-Priors model.
Then there is no training procedure that,
for every correctly labeled training set
\[
    \mathcal{S}=\{(G_1,P(G_1)),\ldots,(G_k,P(G_k))\},
\]
constructs an isomorphism-invariant predictor $M_{\mathcal{S}}$ satisfying all
of the following:
\begin{itemize}[noitemsep]
    \item \emph{Efficiency:} The training time, inference time, and description length of $M_{\mathcal{S}}$ are each polynomial in $n$ and $k$.

    \item \emph{Training completeness:} The predictor answers
    correctly on every training example:
    $M_{\mathcal{S}}(G_i)=P(G_i)$ for every $i\in[k]$.

    \item \emph{Soundness:} The predictor never returns an incorrect answer: For every query graph $H$,
    $M_{\mathcal{S}}(H)\in\{P(H),\bot\}$, where $\bot$ denotes abstention.
\end{itemize}
\end{observation}

\begin{proof}
Suppose such a training procedure exists. We use it to compute $P$
in the Isomorphic-Priors model.

During preprocessing, given the prior graphs $G_1,\ldots,G_k$
together with $P(G_1),\ldots,P(G_k)$, we run the training
procedure on the corresponding training set $\mathcal{S}$ and store the
resulting predictor $M_{\mathcal{S}}$.

Given a query graph $H$, we evaluate $M_{\mathcal{S}}(H)$. If the output is not
$\bot$, we return it; otherwise, we report that $H$ is not isomorphic
to any prior graph. Soundness guarantees that every non-abstaining
answer is correct. To justify the report upon abstention, suppose
that $H\cong G_i$ for some $i\in[k]$. Then
\[
    M_{\mathcal{S}}(H)=M_{\mathcal{S}}(G_i)=P(G_i)=P(H),
\]
where the equalities follow, respectively, from the isomorphism-invariance of $M_{\mathcal{S}}$, the training completeness, and the isomorphism-invariance of $P$.

The resulting  algorithm has polynomial preprocessing and query times, contradicting the assumed hardness of $P$ in the Isomorphic-Priors model.
\end{proof}

For example, assuming $\GI \notin \cP$, our lower bound
for $\TSP$ rules out an isomorphism-invariant learner that
simultaneously guarantees efficiency, training completeness, and soundness for every correctly labeled training set.
Such a learner must therefore relax at least one of these
requirements. For example, one may allow incorrect predictions on
some query graphs, or require the guarantees to hold only with high
probability when the training set is sampled from a specified
distribution, rather than for every possible training set.
The observation alone does not rule out useful learning-based
methods under such relaxations.

In the other direction, our  invariant-representation
algorithms for, e.g., $\csp$ and $\energyName$ construct predictors satisfying
all three conditions for every correctly labeled training set.
These predictors answer correctly on every graph isomorphic to a
training example and can also transfer answers to non-isomorphic
graphs with matching representations. Thus, these constructions
establish the feasibility of the abstract predictor requirements,
but do not by themselves show that a particular neural architecture
can represent or learn the resulting predictors. Natural next
questions are whether concrete graph-learning architectures can
learn or exploit the identified representations, and under what
structural or distributional assumptions they can provide
non-abstaining answers on graphs not isomorphic to any training
example.

\paragraph{Beyond Worst-Case Analysis.}
Our work contributes to the long-standing effort in theoretical computer science to define models that bypass overly pessimistic worst-case lower bounds. Classical approaches in this vein include average-case analysis and smoothed analysis, which assume inputs are drawn from specific distributions or subjected to random perturbations to model typical real-world data rather than adversarial instances. Other paradigms, such as instance optimality and universal optimality, abandon uniform worst-case bounds altogether; instead, they require an algorithm's performance to be competitive with the best possible algorithm tailored specifically to the given input instance or network topology. In the emerging field of learning-augmented algorithms (or algorithms with predictions), an algorithm receives a prediction (often learned from prior data) to guide its execution, with its guarantees degrading gracefully depending on the noise or error in the prediction.

Our model differs from these frameworks in that we explicitly allow the algorithm to preprocess prior data without assuming that this data comes from a certain distribution or that there is an oracle that provides predictions. We believe our model offers an orthogonal angle to study beyond-worst-case analysis, directly answering whether we can bypass existing lower bounds for graphs that are structurally equivalent to those we have seen before.

\paragraph{Graph Isomorphism ($\GI$).}
Clearly, if we can test whether two $n$-vertex, $m$-edge graphs are isomorphic in $T(n,m)$ time, then every problem in our model can be solved with $O(k \cdot T(n,m))$ query time and zero preprocessing time. Furthermore, if a more demanding problem called \emph{graph canonization} can be solved in $T(n,m)$ time, we can eliminate the dependence on $k$ in the query phase entirely, reducing the query time to $O(T(n,m))$ at the expense of an $O(k \cdot T(n,m))$ preprocessing time.\footnote{In graph canonization, an algorithm maps every graph $G$ to a canonical string $f(G)$ such that $f(G)=f(H)$ if and only if $G$ and $H$ are isomorphic. Since the length of $f(G)$ is bounded by $T(n,m)$, we can spend $O(k \cdot T(n,m))$ preprocessing time to insert the $k$ canonical forms into a deterministic trie. For a query graph $H$, computing $f(H)$ and traversing the trie takes exactly $O(T(n,m))$ time, matching the overall query bound.}

This implies that all problems in our setting can be solved in quasipolynomial time using Babai's landmark $\GI$ algorithm \cite{Babai16,Babai19}, and even in polynomial time for restricted graph classes, such as planar \cite{hopcroft1974linear}, bounded-degree \cite{luks1982isomorphism,DBLP:journals/siamcomp/GroheNS23}, or bounded-treewidth graphs \cite{bodlaender1990polynomial,LokshtanovPPS17,GroheNSW20}, and more generally, all classes that exclude a fixed graph as a topological minor \cite{GroheM15, Neuen24}. However, it remains a major open question whether $\giName$ is solvable in polynomial time for general graphs. Moreover, to beat the $O(n^3)$ bound for $\negName$, we would need a canonization algorithm running in $O(n^{3-\epsilon})$ time, which appears even more elusive.\footnote{A simple isomorphism test running in $O(n^{3-\epsilon})$ time would not suffice, as it yields an $O(k\cdot n^{3-\epsilon})$ query time in our setting. This trivially exceeds the naive $O(n^3)$ baseline whenever $k=\omega(n^\epsilon)$.} 

The core technical challenge of our model is therefore: \emph{which specific problem structures allow us to bypass the general $\giName$ barrier?} This relationship goes both ways: Any failure to achieve fast algorithms in our model can be viewed as a barrier to fast graph isomorphism testing. Ultimately, we hope that future studies within the Isomorphic-Priors framework will reveal the fundamental combinatorial structures that make $\giName$ inherently difficult.

\paragraph{Weisfeiler-Leman (WL) Color Refinement.}

The Weisfeiler-Leman (WL) algorithm is a fundamental heuristic for isomorphism testing (see, .e.g, \cite{GroheN21,Kiefer20,Kiefer20b}). On the positive side, Atserias and Maneva~\cite{DBLP:journals/siamcomp/AtseriasM13} established relationship between Sherali-Adams (SA) hierarchy of linear programming relaxations and the levels of the WL algorithm. Consequently, higher-dimensional WL colorings can serve as invariant representations for polynomial-time solvable problems formulated via basic LP relaxations, such as {\em maximum bipartite matching and $\maxFlowName$}. However, since these problems already admit efficient algorithms, this connection does not yield new speedups in our model; in fact, we show in this paper that the use of WL coloring can be refined to achieve near-linear time in our setting.

Another line of work, which is largely orthogonal to our focus, aims to demonstrate that WL conceptually subsumes other polynomial-time graph isomorphism heuristics (a motivation explicitly discussed by Arvind et al.~\cite{ArvindFKV22}).

Strikingly, beyond these LP-based equivalences, positive algorithmic results using WL invariants to accelerate exactly or approximately hard problems appear virtually nonexistent in the literature. In this light, a core conceptual contribution of our work is demonstrating that this well-known technique can be easily leveraged to bypass traditional worst-case computational barriers.

In contrast to the scarcity of positive results, there is a rich lineage of negative and inapproximability bounds associated with the WL hierarchy. As opposed to relying on standard complexity assumptions (e.g., $\cP \neq \cNP$ or the Unique Games Conjecture), this line of work leverages WL, often formulated via its equivalence to infinitary first-order logic with counting, as a mathematical proxy to unconditionally rule out algorithms expressible in fixed-point logic with counting (FPC). Through this lens of unconditional inexpressibility, Atserias and Dawar~\cite{AtseriasD19} proved that $k$-WL cannot approximate $\vcName$ to a factor better than $7/6$ (alongside other inapproximability results), while Tucker-Foltz~\cite{TuckerFoltz24} recently established inapproximability thresholds for Unique Games under WL refinement. For exact computation, Göbel et al.~\cite{GobelGR24} showed that finding a dominating set of size $k$ requires a WL dimension of at least $k$, fitting into a broader understanding of WL's limitations on conjunctive queries.

Another related line of work explores the relationship between WL coloring and homomorphism counts~\cite{Dvorak10,Dvorak10, DBLP:conf/icalp/DellGR18,GroheRS22}, as well as (induced) subgraph counts~\cite{Neuen24a,LanzingerB24,CurticapeanN25}. This homomorphism framework has recently become a central tool in the machine learning community for characterizing the expressiveness of Graph Neural Networks (GNNs)~\cite{bouritsas2022improving, DBLP:conf/iclr/ZhangGDY0024}. We purposefully present both elementary combinatorial arguments (\Cref{sec:1WL}) and the more succinct homomorphism formulation (\Cref{sec:WL-for-Experts}) to ensure accessibility for a broad algorithmic audience. In a similar spirit, it is known that the $2$-dimensional Weisfeiler-Leman (2-WL) algorithm subsumes the spectral information of the adjacency matrix (as well as other associated graph matrices)~\cite{Furer10,RattanS23,ArvindFKV25}, which we exploit for our deterministic algorithms for $\cstName$ and $\lstName{p}$.

\section{Our Model of Computation}\label{sec:model}

\subsection{Preliminary: \Attributed Graphs and Isomorphism}

Throughout this work, we consider graphs whose vertices and edges may carry \emph{\attributes}. For instance, a vertex \attribute may be used to designate sources or sinks, while an edge may carry a vector \attribute representing multiple edge weights. The specific \attributes depend on the problem under consideration. 

For example, in $\negName$, the \attribute of each edge is simply its weight, and there are no vertex attributes (i.e., every vertex receives the same default attribute, denoted by $\bot$). For $\cspName$, the \attribute of each edge $e$ with length $\ell(e)$ and cost $c(e)$ is the vector $\att(e) = (\ell(e), c(e))$. The vertex attributes designate the terminals: $\att(v) = \mathsf{s}$ if $v$ is the source, $\att(v) = \mathsf{t}$ if $v$ is the sink, and $\att(v) = \bot$ otherwise.

We adopt the natural notion of \emph{isomorphism} for \attributed graphs: an isomorphism between two graphs $G$ and $H$ is a bijection between their vertex sets that preserves adjacency as well as all vertex and edge \attributes.

\begin{definition}[Isomorphism of \attributed graphs]\label{def:att:iso}
Let 
\[
G = (V_G, E_G, \att_G^V, \att_G^E)
\quad\text{and}\quad
H = (V_H, E_H, \att_H^V, \att_H^E)
\]
be two graphs, where $V_G, V_H$ are vertex sets, $E_G, E_H$ are edge sets, and $\att^V, \att^E$ assign \attributes from alphabets $\Sigma_V, \Sigma_E$ to the vertices and edges, respectively.

We say that $G$ and $H$ are \emph{isomorphic} (denoted $G \cong H$) if there exists a bijection $\varphi : V_G \to V_H$ such that the following conditions hold.
\begin{enumerate}[nolistsep]
    \item \textbf{Adjacency is preserved:} $(u,v) \in E_G \iff (\varphi(u), \varphi(v)) \in E_H$ for all $u,v \in V_G$.
    \item \textbf{Vertex \attributes are preserved:} $\att_G^V(u) = \att_H^V(\varphi(u))$ for all $u \in V_G$.
    \item \textbf{Edge \attributes are preserved:} $\att_G^E(u,v) = \att_H^E(\varphi(u), \varphi(v))$ for all $(u,v) \in E_G$.
\end{enumerate}

\end{definition}

\begin{remark}[Isomorphism-closed problems]\label{remark:iso:closed}
This paper focuses on \emph{isomorphism-closed problems}, a class encompassing virtually all natural computational tasks on graphs and related objects. Let $\Pi$ be a decision or optimization problem (e.g., $\cspName$, $\negName$, $\SAT$, or $\energyName$). Each problem $\Pi$ we consider is equipped with an equivalence relation $\cong$ defining when two problem instances are isomorphic.\footnote{For graph problems (e.g., $\GI$, $\csp$, $\cst$, $\IS$, $\MaxCut$), this corresponds to standard or attributed graph isomorphism (see \Cref{def:att:iso}). For problems like $\SAT$, two CNF formulas $\Phi$ and $\Phi'$ are isomorphic if $\Phi'$ can be obtained from $\Phi$ via a bijection of variables (preserving negations) alongside arbitrary permutations of clauses and literals.} We assume that $\Pi$ is inherently \emph{closed under isomorphism}. For a decision problem, this means $I \cong J$ implies that $I$ is a YES-instance if and only if $J$ is a YES-instance. For an optimization problem, $I \cong J$ implies that both instances share the identical optimal value, meaning $\opt_\Pi(I) = \opt_\Pi(J)$.
\end{remark}

\subsection{The Isomorphic-Priors Model}

We study every problem $\Pi$ within a model consisting of two phases:

\paragraph{Preprocessing Phase.}
We are given a database $\mathcal{S} = \{G_1, G_2, \ldots, G_k\}$ of $k$ \attributed graphs, referred to as \emph{prior graphs}.\footnote{For simplicity, we focus on defining the model for graphs, though this framework naturally extends to every combinatorial object equipped with a well-defined notion of isomorphism.} We assume every graph in $\mathcal{S}$ has exactly $n$ vertices and $m$ edges, and that all global problem parameters are fixed across all instances.
Each $G_i \in \mathcal{S}$ is provided alongside its exact solution, denoted by $\opt(G_i)$. For decision problems (e.g., $\negName$), $\opt(G_i) \in \{\text{YES}, \text{NO}\}$. For optimization problems (e.g., $\cspName$), $\opt(G_i)$ defaults to the optimal objective \emph{value}. Note, however, that if the database equips $G_i$ with an explicit \emph{optimal solution} (e.g., the actual optimal path or vertex set), all our algorithms are capable of reconstructing and outputting the exact optimal solution for the query graph as well.

The goal of the preprocessing phase is to compute a data structure over $\mathcal{S}$ to accelerate future queries. Generally, we consider polynomial preprocessing time as acceptable.

\paragraph{Query Phase.}
Upon receiving a query graph $H$, the algorithm must satisfy the following correctness guarantees. If $H$ is isomorphic to some prior graph $G_i \in \mathcal{S}$, the algorithm must output the correct solution $\opt(H)=\opt(G_i)$. 
Otherwise, if $H$ is not isomorphic to any graph in $\mathcal{S}$, the algorithm must either (i) output the correct solution $\opt(H)$, or (ii) correctly declare that $H \not\cong G_i$ for all $G_i \in \mathcal{S}$.

We say an algorithm is correct \emph{with high probability} (w.h.p.) if it answers the query correctly with probability at least $1 - 1/N^c$ for any arbitrarily large constant $c \ge 1$, where $N$ denotes the input size (of a single query). This probability is taken over the internal randomness of the algorithm. Throughout, our algorithms are analyzed under the standard \emph{oblivious adversary} model; that is, we assume the sequence of queries is generated independently of the algorithm's internal random choices (see, e.g., \cite{DBLP:conf/crypto/NaorY15} and \cite[Definition 3.4]{DBLP:books/crc/KatzLindell2014}).

\begin{definition}\label{def:iso:model:poly}
We say that a problem $\Pi$ is \emph{polynomial-time solvable inside the Isomorphic-Priors model} if it can be solved in polynomial preprocessing time and polynomial query time. 
\end{definition}

\subsection{Our Focused Approach: Invariant Representations}\label{sec:Inv:Rep}

All our algorithms employ essentially the same unified design framework called \emph{invariant representation}.
In the following, we give an abstract description in its most basic form.
Let $\mathcal{G}$ be the space of all valid input graphs for a computational problem $\Pi$.
Algorithms in this framework compute a function $f\colon \mathcal{G} \to \mathcal{X}$, for some $\mathcal{X}$ called the \emph{feature space}, that satisfies the following.
\begin{enumerate}
    \item \textbf{Permutation Invariance:} For every $G, H \in \mathcal{G}$, if $G \cong H$, then $f(G) = f(H)$.
    \item \textbf{Solution-Preserving:} For every $G, H \in \mathcal{G}$, if $f(G) =f(H)$, then $\opt_\Pi(G) = \opt_\Pi(H)$.
\end{enumerate}

We say that $G$ and $H$ are \emph{$f$-equivalent} if $f(G) = f(H)$.
We call such an $f$ an \emph{Invariant and Solution-Preserving Signature for $\Pi$}.
It is easy to see that Invariant and Solution-Preserving Signatures that are efficiently computable suffice to solve problems in our setting.
For the next theorem, we assume that $\mathcal{X} = \{0,1\}^*$ (which can always be achieved by encoding objects from $\mathcal{X}$ using binary strings).

\begin{theorem}\label{theo:IPM:via:sig}
    Let $\Pi$ be a graph optimization or decision problem and let $f\colon \mathcal{G}\to \{0,1\}^*$ be an Invariant and Solution-Preserving Signature for $\Pi$.
    Also let $T(n)$ denote the time required to evaluate $f$ on inputs of size at most $n$.
    
    Then there is an algorithm that solves $\Pi$ with $k$ prior instances each of size at most $n$ in $O(k \cdot T(n))$ preprocessing time and $O(T(n))$ query time.
\end{theorem}

\begin{proof}
    Let $G_1,\dots,G_k$ denote the prior instances, which all have size at most $n$.
    In the preprocessing phase, we compute all signatures $f(G_i) \in \{0,1\}^*$ and store them together with the corresponding $\opt_i$ in a trie.
    Note that storing a single signature string in a trie takes time proportional to the length of the string, which is bounded by $O(T(n))$.
    Hence, the preprocessing phase runs in time $O(k \cdot T(n))$.

    In the query phase, we are getting an instance $H$.
    We first compute the signature $f(H)$ and search this string (together with the corresponding optimal value) in the trie.
    If there is some $i \in [k]$ such that $f(G) =f(H)$, then $\opt_\Pi(G_i) = \opt_\Pi(H)$, and we return the stored optimal value $\opt_\Pi(G_i)$.    
    If no matching signature is found, then we return that $H \not \cong G_i$ for ever $i \in [k]$.
\end{proof}

We stress that, although all our algorithms follow the basic approach outlined above, we usually do not rely on \Cref{theo:IPM:via:sig} as a black-box.
First, in several cases, our algorithms are randomized, i.e., the output of the signature function $f$ depends on random bits.
In this case, the preservation of solutions is usually achieved only with high probability.

Also, in all cases discussed in this paper, the Invariant and Solution-Preserving Signatures we build can even be used to reconstruct optimal solutions (and not only the value).
However, as already discussed above, the reconstruction is problem-specific, and does not follow from the general argument above (see also \nameref{par:value-vs-solution}).

\section{Algorithms via Randomized Polynomial Evaluation}\label{sec:algorithms via polynomials}\label{sec:PIT}

In this section, we show how to utilize basic ideas from the field of randomized algorithm to obtain fast algorithms in the Isomorphic-Priors model.
Our approach maps the multiset of edge weights of specific subgraphs (e.g. triangles, spanning trees) to the monomials of a formal multivariate polynomial. To efficiently compare these polynomials without explicitly expanding them, we rely on Polynomial Identity Testing (PIT) over finite fields.
To this end, we use the Schwartz-Zippel Lemma.

\begin{lemma}[Schwartz-Zippel Lemma]\label{lem:schartz:zippel}
Let $P(x_1, \dots, x_m)$ be a non-zero multivariate polynomial of total degree $d$ over a field $\mathbb{F}$. If $r_1, \dots, r_m$ are chosen independently and uniformly at random from a finite subset $S \subseteq \mathbb{F}$, then $\Pr[P(r_1, \dots, r_m) = 0] \le \frac{d}{|S|}$.
\end{lemma}

We start with $\negName$ where we are presented with a weighted graph $G$ and have to decide if $G$ contains a negative triangles (see \cref{def:neg:tri}). The assumption that no algorithm solves this problem in time $O(n^{3 - \varepsilon})$ for any $\varepsilon > 0$ is a widely used hardness assumption in fine-grained complexity theory (see \cite{DBLP:conf/focs/WilliamsW10}). However, we show that inside the Isomorphic-Priors model, we can solve the problem in time $O(n^\omega)$, where $\omega < 2.373$ denote the matrix multiplication exponent.

\begin{theorem}\label{cor:neg_triangle_preprocessing}
There is an algorithm that solves $\negName$ with $k$ prior instances each having at most $n$ vertices  in  $\widetilde{O}(k \cdot n^\omega)$ preprocesses time and $\widetilde{O}(n^\omega)$ query time with high probability.

Additionally, if each prior $G_i$ that contains a negative triangle is equipped with a weight multiset $\{w^{(i)}_1, w^{(i)}_2, w^{(i)}_3\}$, then we can output the specific vertices of a negative triangle in $H$ in $\widetilde{O}(n^\omega)$ time with high probability. 
\end{theorem}
    Next, we study the \emph{$\cstName$} ($\cst$) problem. Here, the input is graph with cost and length \attributes attached to each edge and a budget $B$. The problem is to find a spanning tree that minimizes the cost and respect a budget constraint (see \cref{def:cst}). It is well-known that this problem is $\cNP$-complete.  However, we show that inside the Isomorphic-Priors model, we can solve the problem in time $O(n^\omega)$.
\begin{theorem}
\label{cor:constrained_spanning_tree_preprocessing}
There is an algorithm that solves the $\cst$ problem with $k$ prior instances each having at most $n$ vertices in $\widetilde{O}(k \cdot n^\omega)$ preprocesses time and $\widetilde{O}(n^\omega)$ query time with high probability.

Additionally, if each prior instance $G_i$ is equipped with a multiset of cost-distance pairs $W_i \coloneqq \{w^{(i)}_1, \dots w^{(i)}_{n-1}\}$, we can explicitly output the corresponding spanning tree in $H$ in $\widetilde{O}(n^\omega)$ time  with high probability. 
\end{theorem}

\subsection{$\negName$}\label{sec:NTD}

\begin{definition}[$\negName$]\label{def:neg:tri}
An instance is defined by an undirected, edge-weighted graph $G=(V, E, w)$, where $w: E \to \mathbb{R}$ assigns a real weight to each edge. The objective is to determine whether there exist three vertices $u, v, x \in V$ forming a triangle such that the sum of their edge weights is strictly negative: $w(u,v) + w(v,x) + w(x,u) < 0$.

Two problem instances $G$ and $H$ of $\negName$ are isomorphic if their underlying graphs are isomorphic. Note that an isomorphism has to preserve  edge weights. Therefore, $\negName$ is closed under isomorphism.
\end{definition}

In order to use \cref{lem:schartz:zippel}, we have to find a way how to detect (or count) negative triangles via polynomial. To this end,
we associate a formal variable $x_{w}$ with every distinct edge weight $w$. The \emph{triangle polynomial} of $G$ is defined as:
\[
P_{\Delta}(G) \ \coloneqq \sum_{w_1 \le w_2 \le w_3} c_{w_1, w_2, w_3}(G) \, x_{w_1} x_{w_2} x_{w_3}
\]
where $c_{w_1, w_2, w_3}(G)$ is the exact number of triangles in $G$ whose edges possess the weight multiset $\{w_1, w_2, w_3\}$. See \cref{fig:intro:NTD}. We now, obtain that $P_{\Delta}(G)$ encodes all different kinds of triangles inside of $G$ and thus especially the negative triangles.

\begin{lemma}
\label{lem:neg_triangle_equiv}
If $G$ and $H$ are two graphs such that $P_\Delta(G) \equiv P_\Delta(H)$, then $G$ contains a negative triangle if and only if $H$ contains a negative triangle.
\end{lemma}
\begin{proof}
If $P_{\Delta}(G) \equiv P_{\Delta}(H)$ then the coefficients of every corresponding monomial match exactly, i.e. $c_{w_1, w_2, w_3}(G) = c_{w_1, w_2, w_3}(H)$ for all combinations of weigths. Consequently, $G$ contains a weight combination with $c_{w_1, w_2, w_3}(G) > 0$ and $w_1 + w_2 + w_3 < 0$ if and only if $H$ contains the exact same number of such negative triangles.
\end{proof}

\begin{proof}[Proof of \Cref{cor:neg_triangle_preprocessing}]
Let $\mathcal{S} \coloneq \{G_i = (V_i, E_i, w_i)\}_{i=1\dots k}$ be a database of prior instances and assume that for each $i \in [k]$ we are given $\opt(G_i)$ that indicates whether $G_i$ contains a negative triangle.

\textbf{Preprocessing:} Given $\mathcal{S}$, we first compute a prime number $p \geq 3 k^{2} n^{\log(n)}$ and construct a finite field $\mathbb{F}_p$. Note that all basic operations in $\mathbb{F}_p$ can be performed in time $O(\log p) = \widetilde{O}(1)$. Next, we assign arbitrarily a unique non-zero element in $\mathbb{F}_p$ to each edge-weight of the graphs $G_i$ and store the edge-weights together with the corresponding field elements in a binary search tree. Formally, we construct a function $\mathrm{enc}_\mathcal{S} \colon \mathbb{R} \to \mathbb{F}_p \cup \{\bot\}$ that, on input $x$, either outputs an unique element in $ \mathbb{F}_p$ if $x$ is the weight of some graph in $\mathcal{S}$, or returns $\bot$ otherwise. Note that $\mathrm{enc}_\mathcal{S}$ can be computed in time $O( k \cdot n^2 \cdot \log p)$, and evaluating $\mathrm{enc}_\mathcal{S}$ on some $x$ takes time $O(\log(k n^2 ) \cdot \log p )=\widetilde{O}(1)$.

Given a graph $G_i$, we write $\mathrm{enc}_\mathcal{S}(G_i)$ for graph where each edge weight $w$ was replaced with $\mathrm{enc}_\mathcal{S}(w)$.  This can be done in time $\widetilde{O}(k  n^2 )$.

We construct a $(k \cdot n^2)$-wise independent family of hash functions
$h \colon \mathbb{F}_p \to \mathbb{F}_p$ using \cite[Construction~3.32]{DBLP:journals/fttcs/Vadhan12}. Concretely, we sample a random polynomial over $\mathbb{F}_p$ of degree at most $ k \cdot n^2  - 1$. 
Note that such a hash function can be constructed in time $O(k \cdot n^2)$. Give an edge-weight $w$, we define $r(w) \coloneq h(\mathrm{enc}_{\mathcal{S}}(w))$.
For a graph $G_i$, we write $r(G_i)$ for the graph obtained by replacing each edge weight $w$ with $r(w)$.

Let $x_1, \dots, x_m \in \mathbb{F}_p$, where $m \leq k \cdot n^2$, be a sequence of values. Using the Fast Fourier Transformation, we can evaluate $h(x_i)$ for all $i$ in time $\widetilde{O}(k \cdot n^2)$. Thus, we can compute $\tilde{G}_1 \coloneqq r(G_1), \dots, \tilde{G}_k \coloneqq r(G_k)$ in time $\widetilde{O}(k n^2)$. To be more specific, let $w_1, \dots w_m \in \mathbb{R}$ be an enumeration of all edge-weights of graphs in $\mathcal{S}$.  Note that $m \leq k \cdot n^2$. We compute the vector $\mathbf{r} = ( r_1 \coloneqq r(w_1), \dots, r_m \coloneqq r(w_m))$ in time $\widetilde{O}(k n^2)$. Further, we store each $r_i$ with the corresponding $w_i$ into a binary data structure $R_\mathcal{S} \colon \mathbb{R} \to \mathbb{F}_p$. Note that we can evaluate $R_\mathcal{S}(w)$ in time $O(\log(k n^2)) = \widetilde{O}(1)$. We write $A_{\tilde{G}_i}$ for the adjacency matrix of $\tilde{G}_i$.  
Note that this is a matrix over $\mathbb{F}_p$. 

Next, we define $f_\Delta(G) \coloneqq \frac{1}{6} \trace(A_{\tilde{G}}^3)$.\footnote{Without loss of generality, we assume $p > 6$.} Computing $A_{\tilde{G}}^3$ takes time $\widetilde{O}(n^\omega)$. Therefore, calculating all $f_\Delta(G_i)$ is in time $\widetilde{O}(k n^\omega)$. For each $i \in [k]$, we compute $f_\Delta(G_i)$ and store the value together with $\opt(G_i)$ in a linear data structure. Formally, we define a data structure $D_\mathcal{S}$ that on input $x$ either returns $D_\mathcal{S}(x) = \opt(G_i)$ for a graph $G_i$ with $f_\Delta(G_i) = x$ or returns that no such graph exist. $D_\mathcal{S}$ can be constructed in time $\widetilde{O}(k)$ and can be queried in time $\widetilde{O}(1)$. Hence, the preporcessing phase takes time $\widetilde{O}(k n^\omega)$.
 
\textbf{Query:}  Given an input graph $H$, we use $\mathrm{enc}_\mathcal{S}$ to check that all weights of $H$ correspond to weights of matrices in $\mathcal{S}$. If this is not the case then our algorithm claims that $H$ is not isomorphic to any graph in $\mathcal{S}$. Next, we use the data structure $R_\mathcal{S}$ to compute $\tilde{H} \coloneqq r(H)$ in time $\widetilde{O}(n^2)$. We then again apply fast matrix multiplication to compute $f_\Delta(H)$ in time $\widetilde{O}(n^\omega)$. Lastly, we check if $x \coloneqq f_\Delta(H)$ is in $D_\mathcal{S}$. If this is the case then we return $D_\mathcal{S}(x)$, otherwise we report that $H$ is not isomorphic to any graph in $\mathcal{S}$. This can be done in time $\widetilde{O}(1)$. Hence, the whole query time is in $\widetilde{O}(n^\omega)$.

\textbf{Correctness:}  As a first step, we show 
\begin{align}\label{eq:triangel:poly:matrix}
    f_\Delta(G) \coloneqq \frac{1}{6} \trace(A_{\tilde{G}}^3) = \sum_{w_1 \leq w_2 \leq w_3} c_{w_1, w_2, w_3}   r(w_1) r(w_2) r(w_3) = P_\Delta(G)(\mathbf{r}),
\end{align}
where $P_\Delta(G)(\mathbf{r})$ denotes the evaluation of $P_\Delta(G)$ where we assign the value $r(w_i)$ to each variable $x_{w_i}$.
A standard result (e.g. \cite[Fact 1]{DBLP:journals/jintseq/FelsnerH15}) states that
\[\trace(A_{\tilde{G}}^3) = \sum_{x_1 \in V(G)} \sum_{x_2 \in V(G)}\sum_{x_3 \in V(G)} r(w(x_1, x_2)) \cdot r(w(x_2, x_3))  \cdot r(w(x_3, x_1)),\]
where $w(x_i, x_j)$ is equal to the weight of the edge from $x_i$ to $x_j$ and $r(w(x_i, x_j)) = 0$ if there is no edge from $x_i$ to $x_j$.
Next, observe that the vertices $x_1, x_2, x_3 \in V(G)$ form a triangle in $G$ if and only if $r(w(x_1, x_2)) \cdot r(w(x_2, x_3)) \cdot r(w(x_3, x_1))$ appears exactly six times in the sum above.\footnote{The factor six is a result of the six possible symmetries that each orientated triangle has.} This yields \cref{eq:triangel:poly:matrix}.

Let $H$ be a graph on which we are queried. The graph $H$ can be isomorphic to some $G_i \in \mathcal{S}$ only if all edge-weights of $H$ appear as edge-weights in $\mathcal{S}$. If this is not the case the our query algorithm rejects. Thus, we assume that all edge-weights in $H$ appear in the list $w_1, \dots, w_m$.

We now consider two cases. First, suppose that $H$ is isomorphic to some $G_i$. Then, $f_\Delta(G_i) = f_\Delta(H)$ due to \cref{eq:triangel:poly:matrix} and the fact that $P_\Delta(G_i) \equiv P_\Delta(H)$ since $H \cong G_i$. Thus, $D_\mathcal{S}(f_\Delta(H))$ returns $\opt(G_i)$ as long as there aren't any different graphs $G_j$ with $f_\Delta(G_i) = f_\Delta(G_j)$ and $\opt(G_i) \neq \opt(G_j)$. By \cref{lem:neg_triangle_equiv}, this only occurs if $P_\Delta(G_i) \not \equiv P_\Delta(G_j)$ but $P_\Delta(G_i)(\mathbf{r}) = P_\Delta(G_j)(\mathbf{r})$. Since all $r_i$ are chosen independently and uniformly at random over $\mathbb{F}_p$, \cref{lem:schartz:zippel} yields that the event $P_\Delta(G_i) \not \equiv P_\Delta(G_j)$ and $P_\Delta(G_i)(\mathbf{r}) = P_\Delta(G_j)(\mathbf{r})$ has probability at most $\frac{3}{p}$ for every fixed index $j \neq i$. Using union bound, we obtain that the failure probability is at most $\frac{3k}{p} \leq k^{-1} n^{-\log(n)}$. Thus, our success probability is at least $1 - k^{-1} n^{-\log(n)}$ and there w.h.p with respect to the query size $n$.

Second, suppose $H$ is non-isomorphic to every $G_i \in \mathcal{S}$ and let $\opt(H)$ denote the value indicating whether $H$ contains a negative triangle. The only way that the algorithm fails is if there exist a $G_i$ with 
$f_\Delta(G_i) = f_\Delta(H)$ and $\opt(G_i) \neq \opt(H)$. By using the same argument as above, we obtain that the probability of this happening is at most $\frac{3k}{p} \leq k^{-1} n^{-\log(n)}$. Thus, our success probability is at least $1 - k^{-1} n^{-\log(n)}$ and there w.h.p with respect to the query size $n$.

\textbf{Reconstruction:}
Assume that for each graph $G_i \in \mathcal{S}$ containing a negative triangle, we are provided with a multiset of edge-weights $W_i \coloneqq \{w^{(i)}_1, w^{(i)}_2, w^{(i)}_3\}$ that form such a triangle. 
Suppose the algorithm is queried on a graph $H$ such that $P_{\Delta}(H) \equiv P_{\Delta}(G_i)$ for some $G_i \in \mathcal{S}$. By \cref{lem:neg_triangle_equiv}, $H$ must also contain a triangle with the identical multiset of edge-weights $W_i$; i.e., such a triangle has total weight $w^{(i)}_1 + w^{(i)}_2 + w^{(i)}_3 < 0$.

To find this specific negative triangle in $H$, we construct an auxiliary tripartite graph $H'$ with vertex sets $V_1, V_2, V_3$, where each $V_j$ is a copy of $V(H)$. For every vertex $v \in V(H)$, let $v_j$ denote its copy in $V_j$.
For every edge $\{u,v\}$ in $H$ with weight $w^{(i)}_1$, we add an undirected unweighted edge $\{u_1,v_2\}$ to $H'$. Similarly, for every edge $\{u,v\}$ in $H$ with weight $w^{(i)}_2$, we add an undirected unweighted edge $\{u_2,v_3\}$ to $H'$, and for every edge with weight $w^{(i)}_3$, we add an undirected unweighted edge $\{u_3,v_1\}$ to $H'$.
Observe that any triangle in $H'$ must consist of one vertex from each partition. Thus, an unweighted triangle in $H'$ corresponds to a triangle in $H$ with the desired edge-weights. Since triangle detection in unweighted graphs can be performed in $O(n^\omega)$ time~\cite{DBLP:conf/stoc/Itai77}, we can find the negative triangle in $H$ in $O(n^\omega)$ total time.
\end{proof}

\subsection{$\cstName$}\label{sec:CST}

\begin{definition}[$\cstName$ ($\cst$) Problem 
{\cite{DBLP:journals/cor/AggarwalAN82,DBLP:conf/swat/RaviG96}}]\label{def:cst}
Let $G=(V,E, c, d, B)$ be an undirected graph with two non-negative edge-weight
functions
\[
c : E \to \mathbb{R}_{\ge 0}
\qquad\text{and}\qquad
d : E \to \mathbb{R}_{\ge 0},
\]
and let $B \in \mathbb{R}_{\ge 0}$ be a budget.
The \emph{$\cstName$ Problem} asks to find a spanning tree
$T \subseteq E$ that
\[
\begin{aligned}
&\text{minimizes} && c(T) \coloneqq \sum_{e \in T} c(e) \\
&\text{subject to} && \sum_{e \in T} d(e) \le B .
\end{aligned}
\]
A spanning tree satisfying the constraint is called \emph{budget-feasible}. The optimal value is $\opt(G)$ for the optimal value of a problem instance $G$. If no spanning-tree satisfies the budget constraint then $\opt(G) = \infty$.

Two problem instances $G = (V, E, c, d, B)$ and $G' = (V', E', c', d', B')$ of $\cst$ are isomorphic if $B = B'$ their underlying graphs are isomorphic. Note that an isomorphism has to preserve vertex attributes and edge attributes. Therefore, $\cst$ is closed under isomorphism.
\end{definition}

Note that each edge $e \in E$ is associated with two different \attributes $c(e)$ and $d(e)$. We combine these two \attributes and say that $w(e) \coloneqq \att^E_G(e) = (c(e), d(e))$ is the weight of the edge $e$. Note that isomorphism respect weights.

Like in the $\csp$ case, we need to develop a method to find a spanning with certain weight constraints via polynomials. To this end,
we associate a formal variable $x_{w}$ with every distinct edge weight $w$. The \emph{spanning tree polynomial} of $G$ is defined as
\[
P_{ST}(G) \coloneqq \sum_{w_1 \le \dots \le w_{n-1}} c_{w_1, \dots, w_{n-1}}(G) \prod_{i=1}^{n-1} x_{w_i}
\]
where $c_{w_1, \dots, w_{n-1}}(G)$ is the number of spanning trees with the weight multiset $\{w_1, \dots, w_{n-1}\}$. Note that this polynomial is closely related to the \emph{Kirchhoff polynomial} (See \cite{Chung2000}).

\begin{lemma}
\label{lem:spanning_tree_equiv}
Let $G$ and $H$ be two graphs. If $P_{ST}(G) \equiv P_{ST}(H)$ then for any target $B$, $G$ contains a spanning tree of total weight $B$ if and only if $H$ contains a spanning tree of total weight $B$. 
\end{lemma}
\begin{proof}
If $P_{ST}(G) \equiv P_{ST}(H)$ then the coefficients of every corresponding monomial match exactly, i.e. $c_{w_1, \dots, w_{n-1}}(G) = c_{w_1, \dots, w_{n-1}}(H)$ for all combinations of weights. Consequently, $G$ contains a spanning tree that uses weights  $w_1, \dots, w_{n-1}$ if and only if $H$ contains the same number spanning tree.   
\end{proof}

\begin{proof}[Proof of \cref{cor:constrained_spanning_tree_preprocessing}]

Let $\mathcal{S} \coloneq \{G_i = (V_i, E_i, c_i, d_i, B_i)\}_{i=1\dots k}$ be a database of prior instances and assume that for each $i \in [k]$ we are given $\opt_i$ that indicates the optimal value of $\cst$.

\textbf{Preprocessing:} Given $\mathcal{S}$, we first compute a prime number $p \geq k^{2} n^{\log(n) + 1}$ and construct a finite field $\mathbb{F}_p$. Note that all basic operations in $\mathbb{F}_p$ can be performed in time $O(\log p) = \widetilde{O}(1)$. Next, we assign arbitrarily a unique non-zero element in $\mathbb{F}_p$ to each edge-weight of the graphs $G_i$ and store the edge-weights together with the corresponding field elements in a binary search tree. Remember that the edge weight $w(e)$ of an edge $e \in E_i$ is defined as the tuple $(c_i(e), d_i(e))$. Formally, we construct a function $\mathrm{enc}_\mathcal{S} \colon \mathbb{R}^2 \to \mathbb{F}_p \cup \{\bot\}$ that, on input $x$, either outputs an unique element in $ \mathbb{F}_p$ if $x$ is the weight of some graph in $\mathcal{S}$, or returns $\bot$ otherwise. Note that $\mathrm{enc}_\mathcal{S}$ can be computed in time $O( k \cdot n^2 \cdot \log p)$, and evaluating $\mathrm{enc}_\mathcal{S}$ on some $x$ takes time $O(\log(k n^2 ) \cdot \log p )=\widetilde{O}(1)$.

Given a graph $G_i$, we write $\mathrm{enc}_\mathcal{S}(G_i)$ for graph where each edge weight $w$ was replaced with $\mathrm{enc}_\mathcal{S}(w)$.  This can be done in time $\widetilde{O}(k  n^2 )$.

We construct a $(k \cdot n^2)$-wise independent family of hash functions $h \colon \mathbb{F}_p \to \mathbb{F}_p$ using \cite[Construction~3.32]{DBLP:journals/fttcs/Vadhan12}. Concretely, we sample a random polynomial over $\mathbb{F}_p$ of degree at most $ k \cdot n^2  - 1$. 
Note that such a hash function can be constructed in time $O(k \cdot n^2)$. Give an edge-weight $w$, we define $r(w) \coloneq h(\mathrm{enc}_{\mathcal{S}}(w))$.
For a graph $G_i$, we write $r(G_i)$ for the graph obtained by replacing each edge weight $w$ with $r(w)$.

Let $x_1, \dots, x_m \in \mathbb{F}_p$, where $m \leq k \cdot n^2$, be a sequence of values. Using the Fast Fourier Transformation, we can evaluate $h(x_i)$ for all $i$ in time $\widetilde{O}(k \cdot n^2)$. Thus, we can compute $\tilde{G}_1 \coloneqq r(G_1), \dots, \tilde{G}_k \coloneqq r(G_k)$ in time $\widetilde{O}(k n^2)$. To be more specific, let $w_1, \dots w_m \in \mathbb{R}^2$ be an enumeration of all edge-weights of graphs in $\mathcal{S}$.  Note that $m \leq k \cdot n^2$. We compute the vector $\mathbf{r} = ( r_1 \coloneqq r(w_1), \dots, r_m \coloneqq r(w_m))$ in time $\widetilde{O}(k n^2)$. Further, we store each $r_i$ with the corresponding $w_i$ into a binary data structure $R_\mathcal{S} \colon \mathbb{R}^2 \to \mathbb{F}_p$. Note that we can evaluate $R_\mathcal{S}(w)$ in time $O(\log(k n^2)) = \widetilde{O}(1)$. We write $L^{\tilde{G}_i}$ for the Laplacian matrix of $\tilde{G}_i$. Further, for each $j \in [n]$,\footnote{For simplicity, we assume that $G_i$ has $n$ vertices which is why $L^{G_i}$ is a $n \times n$-matrix.} we write $\tilde{L}_j^{\tilde{G}_i}$ for the matrix that is obtained by removing the $j$-th row and the $j$-th column from $L^{\tilde{G}_i}$. Note that $\tilde{L}_j^{\tilde{G}_i}$ is a matrix over $\mathbb{F}_p$. 

Next, we define $f_{ST}(G) \coloneqq  \det(\tilde{L}_1^{\tilde{G}_i})$ which can be computed in time $\widetilde{O}(n^\omega)$. For each $i \in [k]$, we compute $f_\Delta(G_i)$ and store the value together with $\opt_i$ in a linear data structure. Formally, we define a data structure $D_\mathcal{S}$ that on input $x$ either returns $D_\mathcal{S}(x) = \opt_i$ for a graph $G_i$ with $f_\Delta(G_i) = x$ or returns that no such graph exist. $D_\mathcal{S}$ can be constructed in time $\widetilde{O}(k)$ and can be queried in time $\widetilde{O}(1)$. Hence, the preporcessing phase takes time $\widetilde{O}(k n^\omega)$.
 
\textbf{Query:}  Given an input graph $H$, we use $\mathrm{enc}_\mathcal{S}$ to check that all weights of $H$ correspond to weights of matrices in $\mathcal{S}$. If this is not the case then our algorithm claims that $H$ is not isomorphic to any graph in $\mathcal{S}$. Next, we use the data structure $R_\mathcal{S}$ to compute $\tilde{H} \coloneqq r(H)$ in time $\widetilde{O}(n^2)$. We then compute $f_{ST}(H)$ in time $\widetilde{O}(n^\omega)$. Lastly, we check if $x \coloneqq f_\Delta(H)$ is in $D_\mathcal{S}$. If this is the case then we return $D_\mathcal{S}(x)$, otherwise we report that $H$ is not isomorphic to any graph in $\mathcal{S}$. This can be done in time $\widetilde{O}(1)$. Hence, the whole query time is in $\widetilde{O}(n^\omega)$.

\textbf{Correctness:}  As a first step, we show that
\begin{align}\label{eq:CST:poly}
    f_{ST}(G) =  \det(\tilde{L}_1^{\tilde{G}_i}) = \sum_{w_1 \le \dots \le w_{n-1}} c_{w_1, \dots, w_{n-1}}(G) \prod_{i=1}^{n-1} r(w_i) = P_{ST}(G)(\mathbf{r}).
\end{align}
Here $P_{ST}(G)(\mathbf{r})$ denotes the evaluation of $P_{ST}(G)$ obtained by  assigning the value $r(w_i)$ to each variable $x_{w_i}$.
A standard result (e.g. \cite[Page 3]{Chung2000}) states that for any $j \in [n]$
\[\det(\tilde{L}_j^{\tilde{G}_i}) = \sum_{T} \prod_{w \in T} r(w),\]
where the sum is over all spanning trees $T$ of $G$ and $\prod_{w \in T}$ is the product over all edge-weights in $T$. By combing spanning trees with the same edge-weights, we obtain \cref{eq:CST:poly}.

Let $H$ be a graph on which we are queried. The graph $H$ is isomorphic to some $G_i \in \mathcal{S}$ only if all edge-weights of $H$ appear as edge-weights in $\mathcal{S}$. If this is not the case, then our query algorithm rejects. Thus, we assume that all edge-weights in $H$ appear in the list $w_1, \dots, w_m$.

We consider two cases. First, suppose that $H$ is isomorphic to some $G_i$. Then, $f_{ST}(G_i) = f_{ST}(H)$ due to \cref{eq:CST:poly} and the fact that $P_{ST}(G_i) \equiv P_{ST}(H)$ since $H \cong G_i$. Thus, $D_\mathcal{S}(f_{ST}(H))$ returns $\opt_i$  provided there is no other graph $G_j$ such that $f_{ST}(G_i) = f_{ST}(G_j)$ and $\opt_i \neq \opt_j$. By \cref{lem:neg_triangle_equiv}, this only occurs if $P_{ST}(G_i) \not \equiv P_{ST}(G_j)$ but $P_{ST}(G_i)(\mathbf{r}) = P_{ST}(G_j)(\mathbf{r})$. Since all $r_i$ are chosen independently and uniformly at random over $\mathbb{F}_p$, \cref{lem:schartz:zippel} yields that the event $P_{ST}(G_i) \not \equiv P_{ST}(G_j)$ and $P_{ST}(G_i)(\mathbf{r}) = P_{ST}(G_j)(\mathbf{r})$ has probability at most $\frac{n-1}{p}$ for every fixed index $j \neq i$. Using union bound, we obtain that the failure probability is at most $\frac{(n-1)k}{p} \leq k^{-1} n^{-\log(n)}$.  Thus, our success probability is at least $1 - k^{-1} n^{-\log(n)}$ and there w.h.p with respect to the query size $n$.

Second, suppose $H$ is non-isomorphic to every $G_i \in \mathcal{S}$ and let $\opt_H$ denote the value of the optimal spanning tree that satisfies the budget constraint. The only way that the algorithm fails is if there exist a $G_i$ with 
$f_{ST}(G_i) = f_{ST}(H)$ and $\opt_i \neq \opt_H$. By using the same argument as above, we obtain that the probability of this happening is at most $\frac{(n-1)k}{p} \leq k^{-1} n^{-\log(n)}$.  Thus, our success probability is at least $1 - k^{-1} n^{-\log(n)}$ and there w.h.p with respect to the query size $n$.

\textbf{Reconstruction:} Assume that each graph $G_i$ is equipped with a multiset of cost-distance pairs $W_i \coloneqq \{w^{(i)}_1, \dots w^{(i)}_{n-1}\}$ corresponding to an optimal spanning tree in $G_i$. 
Suppose the algorithm is queried on a graph $H$ such that $P_{ST}(H)\equiv P_{ST}(G_i)$ for some $G_i \in \mathcal{S}$. By \cref{lem:spanning_tree_equiv}, $H$ must contain a spanning tree with the same multiset of cost-distance pairs $W_i$. Upon retrieving $W_i$, let $k_{c,d}$ denote the multiplicity of the pair $(c,d)$ in $W_i$. We formulate the reconstruction of an optimal tree in $H$ as a matroid intersection problem over the ground set $E(H)$ between: (1) a graphic matroid $\mathcal{M}_G$ (where independent sets are forests), and (2) a partition matroid $\mathcal{M}_P$ (where independent sets contain at most $k_{c,d}$ edges of type $(c,d)$).
Consequently, the desired spanning tree in $H$ can be computed in polynomial time using classical matroid intersection algorithms~\cite{schrijver2003}, and specifically in $\widetilde{O}(m\sqrt{n})$ time via recent advancements~\cite{DBLP:conf/stoc/BlikstadMNT23,DBLP:conf/focs/ChakrabartyLS0W19,DBLP:conf/icalp/GabowS85}.

\end{proof}

\section{Algorithms via 1-WL (Elementary Arguments)}\label{sec:algorihtms via WL}\label{sec:1WL}

In this section, we show that $\cspName$ and $\energyName$ can be solved in near-linear time in the Isomorphic-Priors model using the 1-dimensional Weisfeiler-Leman (1-WL) algorithm.
We start by formally defining both problems and stating the main algorithmic results for both problems.
After that, we give an introduction to the 1-WL algorithm and prove our algorithmic results.

\subsection{$\cspName$ and $\energyName$}

\begin{definition}[$\cspName$ ($\csp$)]
An instance of $\cspName$ is defined by a tuple $G=(V, E, c, \ell, s, t, B)$, where:
\begin{itemize}
    \item $(V, E)$ is a finite directed graph.
    \item $c: E \to \mathbb{R}_{\ge 0}$ is a cost function assigning a non-negative cost to each edge.
    \item $\ell: E \to \mathbb{R}_{\ge 0}$ is a length (or delay) function assigning a non-negative length to each edge.
    \item $s, t \in V$ are the distinguished source and target vertices.
    \item $B \in \mathbb{R}_{\ge 0}$ is the maximum allowable path cost constraint.
\end{itemize}
For a directed path $P = (e_1, \dots, e_r)$ from $s$ to $t$, the total cost is $c(P) = \sum_{i=1}^r c(e_i)$ and the total length is $\ell(P) = \sum_{i=1}^r \ell(e_i)$. The objective is to find a feasible path $P^*$ minimizing $c(P^*)$ subject to $\ell(P^*) \le B$. The optimal value is $\opt(G) = c(P^*)$. If no path satisfies the length constraint, $\opt(G) = \infty$.

Two problem instances $G = (V, E, c, \ell, s, t, B)$ and $G' = (V', E', c', \ell', s', t', B')$ of $\csp$ are isomorphic if $B = B'$ and their underlying graphs are isomorphic. Note that an isomorphism has to preserve vertex attributes and edge attributes (e.g., $s$ maps $s'$, $t$ maps $t'$, etc.). 
\end{definition}

Naturally, $\csp$ with isomorphic priors is formalized as follows. In the preprocessing phase, we are given a database $\mathcal{S} = \{G_1, \dots, G_k\}$ of $k$ prior instances. Each instance is given as a tuple $G_i = (V_i, E_i, c_i, \ell_i, s_i, t_i, B_i)$, provided together with its optimal value $\opt(G_i)$. In the query phase, we receive a query graph $G = (V, E, c, \ell, s, t, B)$.

\begin{theorem}[$\csp$ in the Isomorphic-Priors Model]
\label{cor:csp_with_preprocessing}
There is an algorithm that solves the $\csp$ problem with $k$ prior instances each having at most $n$ vertices and $m$ edges in  $O(k(n+m)\log n)$ preprocesses time and $O((n+m)\log n)$ query time. 
Additionally, if each prior $G_i$ is equipped with an optimal path, then the optimal path for the query graph can be found in $O((n+m)\log n)$ query time.
\end{theorem}

\begin{definition}[$\energyName$]
An $\energyName$ instance is defined by a tuple $G=(V, E, V_1, V_2, w, v_0)$, where:
\begin{itemize}
    \item $(V, E)$ is a finite directed graph with no terminal states.
    \item $V_1, V_2$ is a partition of $V$ ($V_1 \cup V_2 = V$, $V_1 \cap V_2 = \emptyset$), where Player~1 controls $V_1$ and Player~2 controls $V_2$.
    \item $w: E \to \mathbb{Z}$ is a weight function assigning an integer energy change to each edge.
    \item $v_0 \in V$ is the designated initial vertex.
\end{itemize}
A play $\pi = v_0 v_1 v_2 \dots$ is formed by moving a token: at step $k \ge 0$, if $v_k \in V_i$, Player~$i$ chooses an outgoing edge $(v_k, v_{k+1}) \in E$. Given an initial energy $E_0 \in \mathbb{N}$, the energy evolves via $E_{k+1} = E_k + w(v_k, v_{k+1})$. A strategy for Player~1 is winning from $v_0$ with initial energy $E_0$ if it ensures that $E_k \ge 0$ for all $k \ge 0$,
regardless of the choices of Player~2. The objective is to find the minimum initial energy $E_0^* \in \mathbb{N}$ for which Player~1 has a winning strategy from $v_0$. The optimal value is $\opt(G) = E_0^*$. If no finite initial energy allows Player~1 to win, we set $\opt(G) = \infty$.

Two $\energyName$ instances $G$ and $G'$ are isomorphic if their underlying graphs are isomorphic. Note that an isomorphism has to preserve vertex attributes and edge attributes. Therefore, $\energyName$ is closed under isomorphism.
\end{definition}

\begin{theorem}[$\energyName$ in the Isomorphic-Priors Model]
\label{cor:energy_preprocessing}
There is an algorithm that solves $\energyName$ with $k$ prior instances each having at most $n$ vertices and $m$ edges in  $O(k(n+m)\log n)$ preprocesses time and $O((n+m)\log n)$ query time. 
Additionally, if each prior $G_i$, that is a winning instance, is equipped with a memoryless winning strategy $\sigma_i$ for Player~1 from every starting vertex, then the algorithm can also compute a memoryless optimal winning strategy $\sigma_H$ with the same guarantees for $H$ in $O((n+m)\log n)$ time.
\end{theorem}

\subsection{Preliminaries: The Weisfeiler--Leman Algorithm}

To solve constrained reachability and game dynamics, we model our problem instances as \attributed directed graphs. Let $G = (V, E, \att_V, \att_E)$ be a directed graph where $\att_E\colon E \to \Sigma_E$ assigns an \emph{\attribute} to each edge (e.g., costs, lengths, or weights), and $\att_V\colon V \to \Sigma_V$ assigns an \emph{\attribute} to each vertex (e.g. source, sink, or what player a vertex belongs to).

When initializing the \emph{1-dimensional Weisfeiler--Leman (1-WL) algorithm}, the vertex \attributes are directly translated into an initial vertex \emph{coloring}, denoted by $\chi_0^G\colon V \to \mathcal{C}$, such that $\chi_0^G(u) = \chi_0^G(v)$ if and only if $\att_V(u) = \att_V(v)$. The 1-WL algorithm then iteratively \emph{refines} these colors (e.g, by recolor some vertices of color $c$ with a new color $d$), ultimately returning a \emph{stable color} for every vertex $u$, denoted by $\chi_\infty^G(u)$, which possesses the structural properties detailed in \Cref{thm:1wl_properties}.

\paragraph{The 1-WL Algorithm: A Primer.}
The 1-WL algorithm can be intuitively understood as an iterative message-passing process. Initially, every vertex only knows its own starting color. In each round, every vertex ``broadcasts'' its current color to its neighbors. Simultaneously, it listens to the incoming messages, tallying up exactly how many neighbors of each color it has, grouped by the \attributes and directions of the connecting edges. It then bundles its own current color with this exact tally to form a unique ``signature''. If two vertices that previously shared the same color generate different signatures (meaning their local neighborhoods differ), they are split into two different new colors. This process continues until no pair of vertices split into two different new colors. The final color is called the {\em stable color}. 

Formally, let $G = (V, E, \alpha_V, \alpha_E)$ be an \attributed directed graph.
For two colorings $\chi,\chi'\colon V \to \mathcal{C}$, we write $\chi \preceq \chi'$ if $\chi'(v) = \chi'(w)$ for all $v,w \in V$ such that $\chi(v) = \chi(w)$ (in other words, the partition into color classes of $\chi$ refines the corresponding partition for $\chi'$). We write $\chi \equiv \chi'$ if $\chi \preceq \chi'$ and $\chi' \preceq \chi$.
The initial coloring computed by $1$-WL is defined via $\chi_0^G(v) = \att_V(v)$ for all $v \in V(G)$.
For $i \geq 0$ the coloring is iteratively refined by setting
\[
\chi_{i+1}^G(v) \coloneqq \Bigg( \chi_i^G(v), \ \Big\{\!\! \Big\{ \big(\chi_i^G(u), \alpha_E(u,v)\big) \mid u \in N_{in}(v) \Big\}\!\! \Big\}, \ \Big\{\!\! \Big\{ \big(\chi_i^G(w), \alpha_E(v,w)\big) \mid w \in N_{out}(v) \Big\}\!\! \Big\} \Bigg),
\]
where $\{\!\{\dots \}\!\}$ denotes a multiset.
This refinement stabilizes when the partition of $V$ induced by the colors stops changing.
More concretely, since $\chi_{i+1}^G \preceq \chi_i^G$ for every $i \geq 0$, there is some minimal $i_\infty \geq 0$ such that $\chi_{i_\infty}^{G} \equiv \chi_{i_\infty+1}^{G}$, and we set $\chi_\infty^G \coloneqq \chi_{i_\infty+1}^{G}$.
Note that we perform one extra iteration after stabilization, which is necessary for some technical reasons.

The 1-WL algorithm can be used as a simple heuristic to test isomorphisms of graphs. If some color appears a different number of times in the stable coloring of the two graphs, then they cannot be isomorphic (the converse is not true in general).

\begin{definition}[1-WL Equivalence]
Two (\attributed) graphs $G$ and $H$ are \emph{1-WL equivalent} ($G \WLequiv{1} H$) if $\big\{\!\! \big\{ \chi_\infty^G(v) \mid v \in V(G) \big\}\!\! \big\} = \big\{\!\! \big\{ \chi_\infty^{H}(w) \mid w \in V(H) \big\}\!\! \big\}$.
\end{definition}

\paragraph{Properties of 1-WL.} Below are some basic facts about 1-WL. 
 
\begin{lemma}[1-WL Properties; Folklore]
\label{thm:1wl_properties}

For every two (\attributed) graphs $G = (V_G,E_G,\att_G^V,\att_G^E)$ and $H = (V_H,E_H,\att_H^V,\att_H^E)$ such that $G \WLequiv{1} H$, the following holds:
\begin{enumerate}
    \item \textbf{Refinement:} For every $v \in V_G$ and $v' \in V_H$, if $\chi_\infty^G(v) = \chi_\infty^H(v')$, then their initial colors match: $\alpha_G^V(v) = \alpha_H^V(v')$.
    \item \textbf{Neighborhood Preservation:} For every $v \in V_G$ and $v' \in V_H$, if $\chi_\infty^G(v) = \chi_\infty^{H}(v')$, then for every stable color $C$ and edge \attribute $\ell$, there exists an incoming edge $(u,v) \in E_G$ in $G$ such that $\alpha_G^E(u,v) = \ell$ and $\chi_\infty^G(u) = C$ if and only if there exists an incoming edge $(u', v') \in E_H$ in $H$ such that $\alpha_H^E(u',v') = \ell$ and $\chi_\infty^{H}(u') = C$. The analogous bi-conditional holds for outgoing edges.
\end{enumerate}
\end{lemma}

\begin{remark} The 1-WL algorithm actually guarantees  stronger properties than the above (see \Cref{thm:1wl_properties_extended}). However, the properties above suffice to get results in this section. 

\end{remark}

\begin{proof}[Proof Sketch of \Cref{thm:1wl_properties}]

Let $v \in V_G$ and $v' \in V_H$ such that $\chi_\infty^G(v) = \chi_\infty^H(v')$.
Also let $i_\infty \geq 0$ be minimal such that $\chi_{i_\infty}^{G} \equiv \chi_{i_\infty+1}^{G}$.
So $\chi_\infty^G \equiv \chi_{i_\infty+1}^{G}$ and, since $G \WLequiv{1} H$ and $\chi_\infty^G(v) = \chi_\infty^H(v')$, we also get that $\chi_\infty^H \equiv \chi_{i_\infty+1}^{H}$.
By definition, it holds for all $i \geq 0$ that, if $\chi_{i+1}^G(v) = \chi_{i+1}^H(v')$, then $\chi_{i}^G(v) = \chi_{i}^H(v')$.
So the first part follows by induction.

To prove Neighborhood Preservation, consider the coloring $\chi^* \colon V(G) \uplus V(H) \to \mathcal{C}$ defined on the disjoint union of $G$ and $H$ via $\chi^*(w) \coloneqq \chi_\infty^G(w)$ for all $w \in V(G)$, and $\chi^*(w') \coloneqq \chi_\infty^H(w')$ for all $w' \in V(H)$.
Since $G \WLequiv{1} H$, it can be verified that the coloring $\chi^*$ is stable for the disjoint union of $G$ and $H$, and in fact $\chi_{\infty}^{G \uplus H} \equiv \chi^*$ (since all refinements being made on $G$ or $H$ also happen on $G \uplus H$).
In particular, we get that $\chi_{\infty}^{G \uplus H}(v) = \chi_{\infty}^{G \uplus H}(v')$.
By definition, we conclude that
\[\Big\{\!\! \Big\{ \big(\chi_{\infty}^{G \uplus H}(u), \alpha_G^E(u,v)\big) \mid u \in N_{in}(v) \Big\}\!\! \Big\} = \Big\{\!\! \Big\{ \big(\chi_{\infty}^{G \uplus H}(u'), \alpha_H^E(u',v')\big) \mid u' \in N_{in}(v') \Big\}\!\! \Big\}.\]
In particular, a pair $(C,\ell)$ appears in the first set if and only if it appears in the second set.
\end{proof}

By repeatedly applying \Cref{thm:1wl_properties} (Neighborhood Preservation), every sequence of valid steps in $G$ can be identically mirrored step-by-step in $H$.
This sequential matching gives rise to the Walk Lifting property.

\begin{lemma}[Walk Lifting]
\label{lem:walk_lifting}
Let $G = (V_G,E_G,\att_G^V,\att_G^E)$ and $H = (V_H,E_H,\att_H^V,\att_H^E)$ be (\attributed) graphs such that $G \WLequiv{1} H$.
Also let $v_0 \in V_G$ and $v_0' \in V_H$ such that $\chi_\infty^G(v_0) = \chi_\infty^{G'}(v'_0)$.
For every directed walk $W = (v_0, e_1, v_1, \dots, e_r, v_r)$ in $G$, there exists a corresponding directed walk $W' = (v'_0, e'_1, v'_1, \dots, e'_r, v'_r)$ in $H$ such that $\chi_\infty^G(v_i) = \chi_\infty^{H}(v'_i)$ for all $0 \le i \le r$, and $\att_G^E(e_j) = \att_H^E(e'_j)$ for all $1 \le j \le r$.
\end{lemma}

\begin{proof}
We proceed by induction on the length of the walk $r$. 

In the case case $r = 0$, the walk $W$ consists of a single vertex $v_0$. We set $W'$ to be the walk consisting only of $v'_0$. By the lemma's hypothesis, $\chi_\infty^G(v_0) = \chi_\infty^{H}(v'_0)$, so the condition holds trivially.

So assume the claim holds for walks of length $k < r$. Let $W_k = (v_0, \dots, e_k, v_k)$ be the prefix of length $k$ in $G$, which by the inductive hypothesis lifts to a valid walk $W'_k = (v'_0, \dots, e'_k, v'_k)$ in $H$ where $\chi_\infty^G(v_k) = \chi_\infty^{H}(v'_k)$. 

The next step in $W$ traverses the edge $e_{k+1} = (v_k, v_{k+1})$ with some \attribute $\ell = \att_G^E(e_{k+1})$. Let $C = \chi_\infty^G(v_{k+1})$ be the stable color of the destination vertex. By \Cref{thm:1wl_properties} (Neighborhood Preservation), since $\chi_\infty^G(v_k) = \chi_\infty^{G'}(v'_k)$, the existence of an outgoing edge from $v_k$ to a vertex of color $C$ with \attribute $\ell$ guarantees the existence of an outgoing edge $e'_{k+1} = (v'_k, v'_{k+1})$ in $H$ with \attribute $\att_H^E(e'_{k+1}) = \ell$ such that $\chi_\infty^{H}(v'_{k+1}) = C$. 
Appending this edge $e'_{k+1}$ and vertex $v'_{k+1}$ to $W'_k$ yields a valid walk of length $k+1$ satisfying all conditions.
\end{proof}

\paragraph{Runtime.} The straightforward mathematical definition of 1-WL constructs nested multisets as colors, leading to exponentially growing color representations (as strings) that are highly inefficient to compute and store. However, by replacing these unwieldy nested tuples with canonical integer identifiers at each iterative step, one obtains the near-linear runtimes required for our analysis \cite{CardonC82,paige1987three, berkholz2017tight}. 
More precisely, we rely on the following result.

\begin{theorem}[\cite{berkholz2017tight}]
    \label{thm:1-wl-implementation}
    There is an algorithm, that given an edge-labeled graph $G$, outputs a vertex-coloring $\chi^G\colon V(G) \to [|V(G)|]$ in time $$O((|V(G)| + |E(G)|) \log |V(G)|)$$ such that the following properties hold:
    \begin{enumerate}
        \item For every graph $G$ it holds that $\chi^G \equiv \chi^G_\infty$.
        \item For every two graphs $G,H$ such that $G \WLequiv{1} H$ it holds that $(G,\chi^G) \WLequiv{1} (H,\chi^{H})$.
        \item For every two isomorphic graphs $G,H$ and every isomorphism $\varphi\colon G \cong H$ it holds that $\chi^G(v) = \chi^{H}(\varphi(v))$ for all $v \in V(G)$.
    \end{enumerate}
\end{theorem}

Note that the output coloring can effectively be treated as the 1-WL coloring. One key difference is that, although $\chi_\infty^G(v) \neq \chi_\infty^H(w)$, it may be that $\chi^{G}(v) = \chi^{H}(w)$ since the exponentially growing color representations are compressed into a set of size $n$.
However, this usually does not cause any problems.

The coloring from \Cref{thm:1-wl-implementation} can be used to obtain a \emph{complete invariant} for 1-WL equivalence (see, e.g., \cite{Otto17}).
For the $\cspName$ and $\energyName$ (and other problems; see Appendix~\ref{app:1wl}), this gives a Invariant and Solution-Preserving Signature.

\begin{theorem}
    \label{thm:1-wl-canon}
    There is a deterministic algorithm that, given an \attributed graph $G$, outputs a string $\operatorname{WL}_1(G) \in \{0,1\}^*$ in time $O((|V(G)| + |E(G)|) \log |V(G)|)$, such that
    for every two input graphs $G,H$ it holds that $G \WLequiv{1} H$ if and only if $\operatorname{WL}_1(G) = \operatorname{WL}_1(H)$.
\end{theorem}

\begin{proof}[Proof Sketch]
    On the input $G = (V,E,\alpha_V,\alpha_E)$, we can compute the coloring $\chi^G$ via \cref{thm:1-wl-implementation}.
    Now, we build the string $\operatorname{WL}_1(G) \in \{0,1\}^*$ to encode the following information:
    \begin{itemize}
        \item For each color $a$ in the image $\chi^G$, the size of $(\chi^G)^{-1}(a)$ and $\alpha_V(v)$ for some (and thus for all) $v \in (\chi^G)^{-1}(a)$ (in increasing order of colors $a \in [|V(G)|]$).
        \item For each pair of colors $a,a'$ in the image $\chi^G$, and each edge \attribute $\ell$ in the image of $\alpha_E$, the number of edges from $(\chi^G)^{-1}(a)$ to $(\chi^G)^{-1}(a')$ with edge \attribute $\ell$.
        If this number is zero, then we do not explicitly encode this information in the string.
    \end{itemize}
    Note that the resulting string has length at most $O((|V(G)| + |E(G)|) \log |V(G)|)$, and can be computed in the same time.
    It is easily seen that $G \WLequiv{1} H$ if and only if $\operatorname{WL}_1(G) = \operatorname{WL}_1(H)$.
\end{proof}

\subsection{1-WL Equivalence Implies Equal Optimality}

Now, let us return to $\cspName$ and $\energyName$.
In order to prove \Cref{cor:csp_with_preprocessing,cor:energy_preprocessing} we argue that 1-WL equivalent instances achieve the same optimal value.
As a result, it suffices to compute complete invariants for 1-WL for all priors, and store optimal solutions for them.

\begin{theorem}
\label{thm:csp-1-wl}
Let 
\begin{align*}
G &= (V,E,c:E\to \mathbb{R}_{\ge 0},\, \ell:E\to \mathbb{R}_{\ge 0},\, s,t, B)
\qquad\text{and}\qquad \\
G' &= (V',E',c':E'\to \mathbb{R}_{\ge 0},\, \ell':E'\to \mathbb{R}_{\ge 0},\, s',t', B)
\end{align*}
be two instances of the $\cspName$ ($\csp$) problem.
If $G \WLequiv{1} G'$, then they have the same optimal value for $\csp$.
\end{theorem}

Before turning to the proof, let us recall how instances $G = (V,E,c:E\to \mathbb{R}_{\ge 0},\, \ell:E\to \mathbb{R}_{\ge 0},\,s,t, B)$ of $\csp$ are interpreted as \attributed graphs.
The \attribute of each edge $e \in E$ with length $\ell(e)$ and cost $c(e)$ is the vector $\att_E(e) = (\ell(e), c(e))$.
The vertex attributes designate the terminals: $\att_V(v) = \mathsf{s}$ if $v$ is the source, $\att_V(v) = \mathsf{t}$ if $v$ is the sink, and $\att_V(v) = \bot$ otherwise.

\begin{proof}
Let $\opt(G)$ and $\opt(G')$ denote the optimal values. By symmetry, it suffices to show $\opt(G') \le \opt(G)$. If no feasible path exists in $G$, then $\opt(G) = \infty$ and the inequality trivially holds.

Assume a feasible optimal path exists, and let $P = (v_0, e_1, v_1, \dots, e_r, v_r)$ be an optimal directed path in $G$ from $v_0 = s$ to $v_r = t$ with total cost $C \ge 0$ and length $L \le B$.
Recall that each edge $e_j$ carries the \attribute $\att_E(e_j) = (c(e_j), \ell(e_j))$.

Because $G \WLequiv{1} G'$ and $\att_V$ isolates the distinguished vertices into singleton color classes, we obtain that $\chi_\infty^G(s)$ and $\chi_\infty^{G'}(s')$.
By \Cref{lem:walk_lifting} (Walk Lifting), there exists a directed walk $W' = (v_0', e_1', v_1', \dots, e_r', v_r')$ in $G'$ starting at $s'$ with the exact same sequence of edge \attributes $(c, \ell)$ and stable vertex colors.
In particular, $\chi_\infty^G(t) = \chi_\infty^G(v_r) = \chi_\infty^{G'}(v_r')$.
Since $t$ and $t'$ are distinguished via $\att_V$ and $\att_V'$, we obtain that $v_r' = t'$ by \Cref{thm:1wl_properties}.

The total cost and length of walk $W'$ are identical to $P$, meaning $W'$ is a feasible solution in $G'$ with cost $C$. If $W'$ contains cycles, they can be removed to form a simple directed path $P'$. Given that the edge lengths and costs are non-negative, removing cycles cannot increase the total cost or violate the length constraint. Thus, $\opt(G') \le C = \opt(G)$.
\end{proof}

\begin{theorem}
    \label{thm:energy-game-1-wl}
    Let $G = (V, E, V_1, V_2, w, v_0)$ and $G' = (V', E', V_1', V_2', w', v_0')$ be two $\energyName$ instances. If $G \WLequiv{1} G'$, then Player~1 has a winning strategy in $G$ from $v_0$ with initial energy $E_0$ if and only if Player~1 has a winning strategy in $G'$ from $v'_0$ with initial energy $E_0$.
\end{theorem}

Again, let us first recall how instances $G = (V, E, V_1, V_2, w, v_0)$ are interpreted as \attributed graphs.
The \attribute of each edge $e \in E$ is $\att_E(e) = w(e)$.
The vertex attributes indicate the partition according to players and the initial vertex: $\att_V(v) = (i,\mathsf{s})$ if $v$ is the initial vertex, and $\att_V(v) = (i,\bot)$ otherwise, where in both cases $i \in \{1,2\}$ is the unique index such that $v \in V_i$.

\begin{proof}
By symmetry, it suffices to show that if Player~1 has a winning strategy $\sigma$ in $G$ starting at~$v_0$ with energy $E_0$, they also have a winning strategy $\sigma'$ in $G'$ starting at $v_0'$ with energy $E_0$.

We define $\sigma'$ dynamically by maintaining a simulated ``shadow play'' in $G$. Let the current finite prefix of the play in $G'$ be $\pi' = (v'_0, \dots, v'_k)$ and the maintained shadow play in $G$ be $\pi = (v_0, \dots, v_k)$. We enforce the invariant that $\chi_\infty^G(v_i) = \chi_\infty^{G'}(v'_i)$ and $w(v_{i-1}, v_i) = w'(v'_{i-1}, v'_i)$ for all $i \le k$.

\begin{enumerate}
    \item \textbf{Base Case:} For $k=0$, the play begins at $v'_0$. We set the shadow play to $v_0$. Since $G \WLequiv{1} G'$, we conclude that $\chi_\infty(v'_0) = \chi_\infty(v_0)$ using \Cref{thm:1wl_properties} (Refinement).
    
    \item \textbf{Player 1's Turn:} Suppose the current play is at $v'_k \in V'_1$. Because colors strictly preserve initial partitions (by \Cref{thm:1wl_properties}, Refinement), the shadow vertex $v_k$ satisfies $v_k \in V_1$. 
    In $G$, strategy $\sigma$ dictates an outgoing edge $(v_k, v_{k+1})$ with weight $W$. By \Cref{thm:1wl_properties} (Neighborhood Preservation), since $\chi_\infty^{G'}(v'_k) = \chi_\infty^G(v_k)$, there exists an outgoing edge in $G'$ to some vertex $v'_{k+1}$ such that $\chi_\infty^{G'}(v'_{k+1}) = \chi_\infty^G(v_{k+1})$ and $w'(v'_k, v'_{k+1}) = W$. Player~1's strategy $\sigma'$ chooses this edge, and we append $v_{k+1}$ to the shadow play.
    
    \item \textbf{Player 2's Turn:} Suppose the current play is at $v'_k \in V'_2$, meaning the shadow vertex $v_k$ satisfies $v_k \in V_2$. 
    Player~2 chooses an arbitrary edge $(v'_k, v'_{k+1})$ with weight $W$. By \Cref{thm:1wl_properties} (Neighborhood Preservation), since $\chi_\infty^G(v_k) = \chi_\infty^{G'}(v'_k)$, there exists an outgoing edge from $v_k$ to some vertex $v_{k+1}$ in $G$ such that $\chi_\infty^G(v_{k+1}) = \chi_\infty^{G'}(v'_{k+1})$ and $w(v_k, v_{k+1}) = W$. We append this $v_{k+1}$ to the shadow play.
\end{enumerate}

Through this induction, every infinite play $\pi'$ in $G'$ conforming to $\sigma'$ maps to a valid shadow play $\pi$ in $G$ conforming to $\sigma$. Crucially, the sequence of edge weights traversed in $\pi'$ is identical to that in $\pi$. Because $\sigma$ is winning in $G$, the energy level along $\pi$ satisfies $E_k \ge 0$ for all $k \ge 0$. Since the step-by-step weights match exactly, the energy level in $G'$ also satisfies $E'_k = E_k \ge 0$ for all $k \ge 0$. Thus, $\sigma'$ is a winning strategy for Player~1 in $G'$. 
\end{proof}

\subsection{Proof of \Cref{cor:csp_with_preprocessing,cor:energy_preprocessing}}

\begin{proof}[Proof of \Cref{cor:csp_with_preprocessing}]
We describe the different phases of the algorithm.

\textbf{Preprocessing:} For each graph $G_i \in \mathcal{S}$, we run the 1-WL algorithm (Theorem~\ref{thm:1-wl-implementation}) and compute its complete invariant for 1-WL equivalence (Theorem~\ref{thm:1-wl-canon}).
Both subroutines run in time $O((n+m)\log n)$ time.
We insert each string $\operatorname{WL}_1(G_i) \in \{0,1\}^*$ into a binary search tree, where the $i$-th letter determines whether we take the left/right subtree on the $i$-th level. 
This takes time proportional to the length of the string $\operatorname{WL}_1(G_i)$, i.e., it can be done in time $O((n+m)\log n)$.
For the corresponding node where the key $\operatorname{WL}_1(G_i)$ is saved, we save the $\opt_i$ as well as an optimal solution path (if provided) and the stable coloring $\chi^{G_i}$.

Note that, if two graphs $G_i, G_j \in \mathcal{S}$ yield the same invariant, they are 1-WL equivalent (Theorem~\ref{thm:1-wl-canon}).
In particular, by \cref{thm:csp-1-wl}, this implies $\opt_i = \opt_j$, and no conflicting values can be assigned to a single node of the search tree.
Overall, processing all $k$ graphs takes $O(k(n+m)\log n)$ time.

\textbf{Query:} Given the query graph $H$, we run the 1-WL algorithm (Theorem~\ref{thm:1-wl-implementation}) and compute its complete invariant for 1-WL equivalence (Theorem~\ref{thm:1-wl-canon}) in $O((n+m)\log n)$ time. 
Looking up the computed invariant in the binary search tree takes $O((n+m)\log n)$ time, allowing us to return $\opt(H) = \opt_i$ (which is correct by \cref{thm:csp-1-wl}).
If the string $\operatorname{WL}_1(H)$ is not contained in the search tree, then we can correctly claim that $H$ is not isomorphic to any graph in $\mathcal{S}$ because graphs that are not 1-WL equivalent are not isomorphic.
The total query time is $O((n+m)\log n)$.

\textbf{Reconstruction:} If optimal paths are stored, the lookup yields the optimal path $P_i = (v_0, e_1, v_1, \dots, e_r, v_r)$ from the equivalent graph $G_i$.
Because $H \WLequiv{1} G_i$, analogous to the proof of \Cref{thm:csp-1-wl}, we can inductively construct a correspondind walk in $H$: if we are currently at $u' \in V(H)$ corresponding to $v_{j} \in P_i$, we scan the outgoing edges of $u'$ to find an edge $e'$ with the same \attributes as $e_{j+1}$ that leads to a vertex $v'$ with the same 1-WL stable color as $v_{j+1}$. 
The existence of such an edge is guaranteed by \Cref{thm:1wl_properties} (Neighborhood Preservation). Scanning for the next edge takes time proportional to the out-degree of the current vertex. Over the entire path, this takes at most $O(n+m)$ time. This procedure yields a valid walk $W'$ in $H$ with identical cost and length to $P_i$. Finally, we remove cycles from $W'$ in $O(n)$ time to produce a simple optimal path $P_H$. The total query time for the path remains $O((n+m)\log n)$.
\end{proof}

We now turn from the proof of $\csp$ to the proof of \Cref{cor:energy_preprocessing}. The shadow play construction in the proof of Theorem~\ref{thm:energy-game-1-wl} may use memory, even when the original strategy is memoryless. To reconstruct memoryless strategies, we therefore use color-consistent strategies, defined as follows.

\begin{definition}[Color-Consistent Strategy]
Let $G=(V,E,V_1,V_2,w,v_0)$ be an $\energyName$ instance with stable 1-WL coloring $\chi_\infty^G$. A memoryless strategy $\sigma$ of Player~1 is \emph{color-consistent} if, for all $u,v\in V_1$ with $\chi_\infty^G(u)=\chi_\infty^G(v)$,
\[
\chi_\infty^G(\sigma(u))=\chi_\infty^G(\sigma(v)) \quad\text{and}\quad w(u,\sigma(u))=w(v,\sigma(v)).
\]
\end{definition}

The following lemma, whose proof is deferred to Appendix~\ref{app:color-consistent-strategies}, shows that a memoryless strategy optimal from every vertex can be made color-consistent in linear time, given the stable 1-WL coloring.

\begin{lemma}\label{lem:color-consistent-strategy}
Let $G=(V,E,V_1,V_2,w,v_0)$ be an $\energyName$ instance with $n$ vertices and $m$ edges. Suppose we are given its stable 1-WL coloring $\chi_\infty^G$ and a memoryless strategy $\sigma$ of Player~1 that wins from every vertex with the minimum possible initial energy. Then a color-consistent memoryless strategy $\widehat{\sigma}$ with the same optimality property can be computed in $O(n+m)$ time.
\end{lemma}

\begin{proof}[Proof of \Cref{cor:energy_preprocessing}]
We describe the different phases of the algorithm.

\textbf{Preprocessing:} For each graph $G_i \in \mathcal{S}$, we run the 1-WL algorithm (Theorem~\ref{thm:1-wl-implementation}) and compute its complete invariant for 1-WL equivalence (Theorem~\ref{thm:1-wl-canon}).
Both subroutines run in time $O((n+m)\log n)$ time.
We insert each string $\operatorname{WL}_1(G_i) \in \{0,1\}^*$ into a binary search tree, where the $i$-th letter determines whether we take the left/right subtree on the $i$-th level. 
This takes time proportional to the length of the string $\operatorname{WL}_1(G_i)$, i.e., it can be done in time $O((n+m)\log n)$.
For the corresponding node where the key $\operatorname{WL}_1(G_i)$ is saved, we save $WIN_i = \opt(G_i)$ and the stable coloring $\chi^{G_i}$. If an optimal memoryless strategy $\sigma_i$ is provided, we apply Lemma~\ref{lem:color-consistent-strategy} to compute and store a color-consistent optimal memoryless strategy $\widehat{\sigma}_i$. This additional step takes $O(m+n)$ time.

Note that, if two graphs $G_i, G_j \in \mathcal{S}$ yield the same invariant, they are 1-WL equivalent (Theorem~\ref{thm:1-wl-canon}).
In particular, by \cref{thm:energy-game-1-wl}, this guarantees $WIN_i = WIN_j$, so no conflicting values can be assigned to a single node of the search tree.
Overall, processing all $k$ graphs takes $O(k(n+m)\log n)$ time.

\textbf{Query:} Given the query graph $H$, we run the 1-WL algorithm (Theorem~\ref{thm:1-wl-implementation}) and compute its complete invariant for 1-WL equivalence (Theorem~\ref{thm:1-wl-canon}) in $O((n+m)\log n)$ time. 
Looking up the computed invariant in the binary search tree takes $O((n+m)\log n)$ time, allowing us to return $\opt(H) = WIN_i$ (which is correct by \cref{thm:energy-game-1-wl}).
If the string $\operatorname{WL}_1(H)$ is not contained in the search tree, then we can correctly claim that $H$ is not isomorphic to any graph in $\mathcal{S}$ because graphs that are not 1-WL equivalent are not isomorphic.
The total query time is $O((n+m)\log n)$.

\textbf{Reconstruction:} Suppose $\opt(G_i)<\infty$ and a color-consistent optimal memoryless strategy $\widehat{\sigma}_i$ is stored. For each Player~1 color class $C$ in $G_i$, choose a representative $u_C$ and set $v_C=\widehat{\sigma}_i(u_C)$. For every Player~1 vertex $u'\in V_{H}$ with $\chi_\infty^H(u')=C$, choose a successor $v'$ satisfying
\[
w_H(u',v')=w_i(u_C,v_C)
\quad\text{and}\quad
\chi_\infty^H(v')=\chi_\infty^{G_i}(v_C),
\]
and define $\sigma_H(u')=v'$. Such a successor exists by \Cref{thm:1wl_properties} (Neighborhood Preservation). Selecting representatives and scanning the outgoing edges of Player~1 vertices takes $O(n+m)$ time, using the color identifiers computed by Theorem~\ref{thm:1-wl-implementation}.

To prove correctness, consider any walk $u_0'u_1'\dots$ consistent with $\sigma_H$, and choose $u_0\in V(G_i)$ with $\chi_\infty^{G_i}(u_0)=\chi_\infty^H(u_0')$. As in the proof of \cref{thm:energy-game-1-wl}, we construct a shadow walk $u_0u_1\dots$ in $G_i$ with matching vertex colors and edge weights. At Player~1 vertices, color consistency ensures that $\widehat{\sigma}_i(u_k)$ has the same successor color and edge weight as the move chosen by $\sigma_H$ at $u_k'$. At Player~2 vertices, Neighborhood Preservation supplies a matching edge. Hence the shadow walk is consistent with~$\widehat{\sigma}_i$.

The shadow-play argument in both directions shows that equally colored vertices have the same minimum initial energy. Since $\widehat{\sigma}_i$ is optimal from every starting vertex and the two walks have identical edge weights, $\sigma_H$ is optimal from every starting vertex as well. In particular, it wins from the designated initial vertex with energy $\opt(H)=\opt(G_i)$. The total query time remains $O((n+m)\log n)$.
\end{proof}

\section{More on Weisfeiler-Leman (Advanced Arguments)}
\label{sec:WL-for-Experts}

\newcommand{\WL}[2]{\chi_{\infty}^{#2,#1}}
\newcommand{\WLit}[3]{\chi_{#2}^{#3,#1}}

The two algorithms from \Cref{sec:1WL} highlight a general approach for obtaining algorithms in the Isomorphic-Priors model: if optimal solution values are preserved under WL equivalence (see, e.g., \Cref{thm:csp-1-wl,thm:energy-game-1-wl}), then we obtain efficient algorithms in our model.
This strategy is not limited to 1-WL, but also holds for higher-dimensional Weisfeiler-Leman.
More precisely, for each $\ell \geq 1$, the $\ell$-dimensional Weisfeiler-Leman ($\ell$-WL) algorithm is a heuristic for graph isomorphism testing that accumulates information by coloring $\ell$-tuples of vertices. For increasing $\ell \geq 1$, the algorithm becomes more expressive, but it also becomes computationally more expensive.

In this section, we present a systematic way to use the $\ell$-WL algorithm within the Isomorphic-Priors model.
After describing the algorithm in \cref{sec:wl-intro}, we begin by showing in \cref{sec:ell:WL:iso} that every problem $\Pi$ that is invariant under $\ell$-WL is polynomial time solvable inside the Isomorphic-Priors model (see \Cref{theo:WL:iso,theo:1WL:iso}). We then develop a technique that helps us prove when a problem is in fact $\ell$-WL invariant.

In \cref{sec:WL:CSP} and \cref{sec:WL:CST}, we apply this technique to the $\csp$ and $\cst$ problems and show that they are 1-WL invariant and 2-WL invariant, respectively. As a result, we obtain fast algorithms for both problems in the Isomorphic-Priors model.
Furthermore, additional examples are discussed in Appendix~\ref{app:1wl}.
However, the Weisfeiler-Leman algorithm has its limits. In \cref{sec:WL:fails}, we present an example of a problem that lies in $\cP$ (and can therefore be efficiently solved in the Isomorphic-Priors model) but that is not $\ell$-WL invariant for any constant $\ell$. This means that Weisfeiler-Leman cannot be used to obtain a fast algorithm for this problem in the Isomorphic-Priors model.

\subsection{The Weisfeiler-Leman Algorithm}
\label{sec:wl-intro}

We start with a description of the $\ell$-WL algorithm (see, e.g., \cite{GroheN21,Kiefer20,Kiefer20b}).
Fix some $\ell \geq 2$ and let $G = (V,E,\att_V,\att_E)$ be an \attributed graph.
For two colorings $\chi,\chi'\colon V^\ell \rightarrow \mathcal{C}$, we say $\chi$ \emph{refines} $\chi'$, denoted $\chi \preceq \chi'$, if $\chi(\bar v) = \chi(\bar w)$ implies $\chi'(\bar v) = \chi'(\bar w)$ for all $\bar v,\bar w \in V^\ell$.
The colorings $\chi$ and $\chi'$ are \emph{equivalent}, denoted $\chi \equiv \chi'$,  if $\chi \preceq \chi'$ and $\chi' \preceq \chi$.

We define the initial coloring $\WLit{\ell}{0}{G}\colon V^\ell \rightarrow \mathcal{C}$ cumputed by $\ell$-WL as the coloring where each tuple is colored with the isomorphism type of its underlying ordered subgraph.
More precisely, for a second \attributed graph $G' = (V',E',\att_V',\att_E')$, and vertices $v_1,\dots,v_\ell \in V$ and $v_1',\dots,v_\ell' \in V'$ we have $\WLit{\ell}{0}{G}(v_1,\dots,v_\ell) = \WLit{\ell}{0}{G'}(v_1',\dots,v_\ell')$
if and only if, for all $i,j \in [\ell]$, we have
\begin{itemize}
    \item $v_i = v_j$ if and only if $v_i' = v_j'$,
    \item $\att_V(v_i) = \att_V'(v_i')$,
    \item $(v_i,v_j) \in E$ if and only if $(v_i',v_j') \in E'$, and
    \item if $(v_i,v_j) \in E$, then $\att_E(v_i,v_j) = \att_E'(v_i',v_j')$.
\end{itemize}

We then recursively define the coloring $\WLit{\ell}{i}{G}$ obtained after $i$ rounds of the algorithm.
For $\bar v = (v_1,\dots,v_k) \in V^\ell$, set
\[\WLit{\ell}{i+1}{G}(\bar v) \coloneqq \big(\WLit{\ell}{i}{G}(\bar v), \mathcal{M}_i(\bar v)\big),\]
where
\[\mathcal{M}_i(\bar v) \coloneqq \Big\{\!\!\Big\{ \big(\WLit{\ell}{i}{G}(\bar v[w/1]),\dots,\WLit{\ell}{i}{G}(\bar v[w/\ell])\big) \;\Big\vert\; w \in V \Big\}\!\!\Big\}\]
and $\bar v[w/i] \coloneqq (v_1,\dots,v_{i-1},w,v_{i+1},\dots,v_\ell)$ is the tuple obtained from substituting the $i$-th entry of $\bar v$ with $w$.
By definition, $\WLit{\ell}{i+1}{G} \preceq \WLit{\ell}{i}{G}$ holds for all $i \geq 0$.
So there is a minimal~$i_\infty \geq 0$ such that $\WLit{\ell}{i_{\infty}}{G} \equiv \WLit{\ell}{i_{\infty}+1}{G}$, and we set $\WL{\ell}{G} \coloneqq \WLit{\ell}{i_\infty+1}{G}$.

\begin{definition}[$\ell$-WL Equivalence]
Two (\attributed) graphs $G$ and $H$ are \emph{$\ell$-WL equivalent} ($G \WLequiv{\ell} H$) if $\big\{\!\! \big\{ \WL{\ell}{G}(\bar v) \mid \bar v \in (V(G))^\ell \big\}\!\! \big\} = \big\{\!\! \big\{ \WL{\ell}{H}(\bar w) \mid \bar w \in (V(H))^\ell \big\}\!\! \big\}$.
\end{definition}

Observe that isomorphic (\attributed) graphs $G,H$ are always $\ell$-WL equivalent.

Given an (\attributed) graph $G$, the coloring $\WL{\ell}{G}$ can be computed in time $O(n^{\ell+1}\log n)$ \cite{immerman1990describing} (more precisely, similar to \Cref{thm:1-wl-implementation}, we compute a coloring equivalent to $\WL{\ell}{G}$).
Moreover, similar to \Cref{thm:1-wl-canon}, we can compute a complete invariant for $\ell$-WL equivalence (see, e.g., \cite{Otto17}).
In turn, this can be used as an Invariant and Solution-Preserving Signature for problems such as $\cstName$ (see \cref{sec:WL:CST}).

\begin{theorem}
    \label{thm:k-wl-canon}
    Let $\ell \geq 2$.
    There is a deterministic algorithm that, given an \attributed graph $G$, outputs a string $\operatorname{WL}_\ell(G) \in \{0,1\}^*$ in time $O(|V(G)|^{\ell+1} \log |V(G)|)$, such that
    for every two input graphs $G,H$ it holds that $G \WLequiv{\ell} H$ if and only if $\operatorname{WL}_\ell(G) = \operatorname{WL}_\ell(H)$.
\end{theorem}

\subsection{$\ell$-WL Invariant Functions are in the Isomorphic-Priors Model}\label{sec:ell:WL:iso}

A function $f$ on (\attributed) graphs is called \emph{$\ell$-WL invariant} if $f(G) = f(H)$ for all (\attributed) graphs $G,H$ with $G \WLequiv{\ell} H$.\footnote{In the literature (see, e.g., \cite{GobelGR24}), the minimal $\ell \geq 1$ (if it exists) such that $f$ is $\ell$-WL invariant is also called the \emph{WL dimension} of $f$.}
We say that an optimization problem $\Pi$ is \emph{$\ell$-WL invariant} if $\opt_\Pi$ is $\ell$-WL invariant (where $\opt_\Pi$ denotes the function that maps an (\attributed) input graph to the value of an optimum solution).
For such problems, we obtain polynomial-time algorithm in the Isomorphic-Priors model using similar arguments as in Section \ref{sec:1WL}.

\begin{theorem}\label{theo:WL:iso}
    Let $\ell \geq 2$ and let $\Pi$ be an optimization problem that is $\ell$-WL invariant. 
    
    Then there is an algorithm that solves $\Pi$ with $k$ prior instances each of size at most $n$ in  $O(k n^{\ell+1} \log n)$ preprocessing time and $O( n^{\ell+1} \log n)$ query time. 
\end{theorem}

\begin{proof}
    Write $\mathcal{S} \coloneqq \{G_i\}_{i = 1\dots k}$ for the prior instances.
    We describe the different phases of the algorithm.

    \textbf{Preprocessing:} For each graph $G_i \in \mathcal{S}$, we compute a complete invariant $\operatorname{WL}_\ell(G) \in \{0,1\}^*$ for $G_i$ using Theorem~\ref{thm:k-wl-canon} in time $O( n^{\ell+1} \log n)$ time.
    We insert each string $\operatorname{WL}_\ell(G_i) \in \{0,1\}^*$ into a binary search tree, where the $i$-th letter determines whether we take the left/right subtree on the $i$-th level. 
    This takes time proportional to the length of the string $\operatorname{WL}_\ell(G_i)$, i.e., it can be done in time $O( n^{\ell+1} \log n)$.
    For the corresponding node where the key $\operatorname{WL}_\ell(G_i)$ is saved, we save $\opt_\Pi(G_i)$.

    Note that, if two graphs $G_i, G_j \in \mathcal{S}$ have the same complete invariant, they are $\ell$-WL equivalent (Theorem~\ref{thm:k-wl-canon}).
    In particular, since $\Pi$ is $\ell$-WL invariant, we get that $\opt_\Pi(G_i) = \opt_\Pi(G_j)$, and no conflicting values can be assigned to a single node of the search tree.
    Overall, processing all $k$ graphs takes $O(k n^{\ell+1} \log n)$ time.

    \textbf{Query:} Given the query graph $H$, we compute a complete invariant $\operatorname{WL}_\ell(H) \in \{0,1\}^*$ for $G_i$ using Theorem~\ref{thm:k-wl-canon} in time $O( n^{\ell+1} \log n)$ time. 
    Looking up the computed invariant in the binary search tree takes $O( n^{\ell+1} \log n)$ time, allowing us to return $\opt_\Pi(H) = \opt_\Pi(G_i)$ (which is correct since $\Pi$ is $\ell$-WL invariant).
    If the string $\operatorname{WL}_\ell(H)$ is not contained in the search tree, then we can correctly claim that $H$ is not isomorphic to any graph in $\mathcal{S}$ because graphs that are not $\ell$-WL equivalent are not isomorphic.
    The total query time is $O( n^{\ell+1} \log n)$.
\end{proof}

For 1-WL invariant optimization problems, we obtain the following variant which is proved analogously using Theorem \ref{thm:1-wl-canon}.

\begin{theorem}
    \label{theo:1WL:iso}
    Let $\Pi$ be an optimization problem that is 1-WL invariant. Then there is an algorithm that solves $\Pi$ with $k$ prior instances each of size at most $n$ in  $O(k(n+m) \log n)$ preprocessing time and $O((n+m) \log n)$ query time. 
\end{theorem}

We note that, in \Cref{cor:csp_with_preprocessing,cor:energy_preprocessing}, we can also reconstruct an optimal solution of a query graphs if we are given optimal solutions for all priors.
While it does not seem possible to obtain such a statement in the general setting of \Cref{theo:WL:iso,theo:1WL:iso}, we note that all our arguments for specific problems do imply that optimal solutions can indeed be reconstructed (see Appendix~\ref{app:1wl} for further examples).

In the following, we use an alternative way to characterize when an optimization problem is $\ell$-WL invariant. We limit ourself to minimization problems. However, our approach can easily be extended to maximization problems. Remember that $\mathcal{G}$ is the set of valid attributed input graphs.

\begin{definition}
    Let $\Pi$
    be an minimization problem. We call a triple $(W, \WLcheck \colon \Graphs \times W  \to \mathbb{R}_{\geq 0}, \WLcost \colon W \to \mathbb{R})$ a \emph{witness function of $\opt_\Pi$} if:
    \begin{itemize}
        \item  For all $G$ and $X \in W$: if $\WLcheck(G, X) > 0$ then $\opt_\Pi(G) \leq \WLcost(X)$.
        \item If $\opt_\Pi(G) < \infty$ then there is a $X \in W$ with $\opt_\Pi(G) = \WLcost(X)$
        and $\WLcheck(G, X) > 0$.
    \end{itemize}
\end{definition}

The intuition behind $(W, \WLcheck \colon \Graphs \times W \to \mathbb{R}_{\geq 0}, \WLcost \colon W \to \mathbb{R})$ is that it characterizes the witnesses of the minimization problem $\Pi$. 
The function ``$\WLcheck$'' checks for an element $X$ if there exist a feasible solution (which is the case if $\WLcheck(G, X) > 0$) with cost at most $\WLcost(X)$.
Further, if $G$ has a feasible solution (i.e. $\opt_\Pi(G) < \infty$) then there is a \emph{witness} $X \in G$ that attains the optimal value (i.e. $\opt_\Pi(G) = \WLcost(X)$ and $\WLcheck(G, X) > 0$). 

In the following, we show that we can a witness function  to show that $\opt_\Pi$ is $\ell$-WL. To this end, we show that the witness function itself is $\ell$-WL invariant. 

\begin{definition}
    Let $\WLcheck \colon \Graphs \times W \to \mathbb{R}_{\geq 0} $ be a function. We say that $\WLcheck$ is \emph{$\ell$-WL invariant on $\Graphs$} if: for all graphs $G, H \in \Graphs$ with $G \WLequiv{\ell} H$ and all $X \in W$ we obtain that $\WLcheck(G, X) > 0$ if and only if $\WLcheck(H, X) > 0$.
\end{definition}

\begin{lemma}\label{lem:witness:inv:implies:inv}
    Let $\Pi$
    be an minimization problem and  $(W, \WLcheck \colon \Graphs \times W  \to \mathbb{R}_{\geq 0}, \WLcost \colon W \to \mathbb{R})$ be a witness function of $\opt_\Pi$. If $\WLcheck$ is $\ell$-WL invariant on $\Graphs$ then $\Pi$ is $\ell$-WL invariant.
\end{lemma}
\begin{proof}
    Let $G, H \in \Graphs$ be two graphs with $G \WLequiv{\ell} H$. Without loss of generality, we assume $\opt_\Pi(G) \leq \opt_\Pi(H)$ (otherwise we change to role of $G$ and $H$). We have to show that $\opt_\Pi(H) \leq \opt_\Pi(G)$. If $\opt_\Pi(G) = \infty$ then we are already done. Otherwise, $\opt_\Pi(G) < \infty$. Thus, there exist a $X \in W$ with $\opt_\Pi(G) = \WLcost(X)$ and $\WLcheck(G, X) > 0$. Now, since $\WLcheck$ is $\ell$-WL invariant, we obtain that $\WLcheck(H, X) > 0$. Hence, $\opt_\Pi(H) \leq \WLcost(X) = \opt_\Pi(G)$.
\end{proof}

\subsection{$\csp$ is $1$-WL Invariant}\label{sec:WL:CSP}

We first demonstrate the above approach on the $\cspName$ ($\csp$) problem and present an alternative proof for \Cref{cor:csp_with_preprocessing} using a characterization of 1-WL via homomorphism counts \cite{Dvorak10, DBLP:conf/icalp/DellGR18}.
Write $\opt$ for the optimal value function of $\csp$ and $\Graphs$ for the set of valid input graphs of $\csp$. We show that $\opt$ is 1-WL invariant for all values.\footnote{There is the small technical detail that an input graph $G$ has a length constraint $B$ that is part of the input. Consistent with previous sections, we say that two input graphs $G$ and $H$ are 1-WL equivalent if $G$ and $H$ have the same length constraint $B$ and  $G \WLequiv{1} H$ for the underlying attributed graphs.}
This together with \cref{theo:1WL:iso} immediately yields that $\csp$ is in the Isomorphic-Priors model.

\begin{theorem}\label{theo:csp:wl:inv}
    The Problem $\csp$ is 1-WL invariant.
\end{theorem}

To prove \cref{theo:csp:wl:inv} we define a witness for $\opt$ and show that this witness is 1-WL invariant. To this end, we start by defining graph homomorphisms.

\begin{definition}[Homomorphism of \attributed Graphs]\label{def:hom}
    Let 
\[
G = (V_G, E_G, \att_G^V, \att_G^E)
\quad\text{and}\quad
H = (V_H, E_H, \att_H^V, \att_H^E)
\]
be two graphs, where $V_G, V_H$ are vertex sets, $E_G, E_H$ are edge sets, and $\att^V, \att^E$ assign \attributes from alphabets $\Sigma_V, \Sigma_E$ to the vertices and edges, respectively.

We say that a function $f \colon V_H \to V_G$ is a \emph{homomorphism} from $H$ to $G$ if:
\begin{enumerate}
    \item \textbf{Adjacency is preserved:} For all $(u,v) \in E_H$ it is true that  $(\varphi(u), \varphi(v)) \in E_G$.
    \item \textbf{Vertex \attributes are preserved:} For all $u \in H$ it is true that $\att_H^V(u) = \att_G^V(f(u))$.
    \item \textbf{Edge \attributes are preserved:} For all $(u,v) \in E_H$ it is true that $\att_H^E(u, v) = \att_G^E(f(u), f(v))$.
\end{enumerate}
We write $\homs{H}{G}$ for the number of homomorphism from $H$ to $G$.
\end{definition}

Next, we use homomorphism to define a witness function. Note that each input graph $G$ has a length constraint $B$.

\paragraph*{Definition of a Witness:}
\begin{itemize}
    \item $W = \{X \, | \, X \subseteq \mathbb{N}_{\geq 1} \times \mathbb{R}_{\geq 0} \times \mathbb{R}_{\geq 0} \}$, 
    \item $\WLcost(X) = \sum_{(n, c, \ell) \in X } \, n \cdot c$,
    \item Let $P$ be a graph. We say that $P$ is a \emph{$X$-colored $s,t$-path} if 
    \begin{itemize}
        \item $P$ is a path with colored start-vertex $s$ and colored end-vertex $t$,
        \item each edge in $P$ is colored with a color $(c, \ell)$, and 
        \item for each $(n, c, \ell) \in X$ the path $P$ contains exactly $n$ edges of color $(c, \ell)$.
    \end{itemize}
    Let
    $P_X \coloneqq \{P \, | \, P  \text{ is a $X$-colored $s,t$-path} \}$ be a set of paths then we define
    \[\WLcheck(G, X) = \begin{cases}
        \sum_{P \in P_X} \homs{P}{G} &\text{ if $\sum_{(n, c, \ell) \in X} \, n \cdot \ell \leq B$ } \\
        0 &\text{ otherwise }
    \end{cases},\]
    where $B$ is length constraint of $G$.
\end{itemize}

\begin{lemma}\label{lem:csp:witness}
    The triple $(W, \WLcheck \colon \Graphs \times W \to \mathbb{R}, \WLcost \colon W \to \mathbb{R})$ is a witness function of $\opt$.
\end{lemma}
\begin{proof}
    We start by showing that for all $X \in W$ with $\sum_{(n, c, \ell) \in X} \, n \cdot \ell \leq B$ we obtain 
    \begin{align}\label{eq:walk}
    \WLcheck(G, X) = \; \#\{\text{walks in $G$ from $s$ to $t$ s.t. for all $(n, c, \ell) \in X$ the color $(c, \ell)$ is used $n$ times}\}. 
    \end{align}
    Let $P \in P_X$ be a path. Since $P$ is a $X$-colored $s,t$-path, we obtain that $V(P) = \{v_1, \dots, v_t\}$ and $E(P) = \{e_1, \dots e_{t-1}\}$ with
    $\att_P^V(v_1) = s$, $\att_P^V(v_t) = t$ and $\att^E_P(e_i) = (c_i, \ell_i)$. Now,
    $\homs{P}{G}$ is equal to the number of walks in $G$ that start at $s$, end at $t$, and where we used an edge of color $(c_i, \ell_i)$ during the $i$th step of the walk. Thus, equation (\ref{eq:walk}) follows since $\WLcheck(G, X)$ computes the sum over all possible $X$-colored $s,t$-paths.

    Next, we check that $(W, \WLcheck \colon \Graphs \times W \to \mathbb{R}, \WLcost \colon W \to \mathbb{R})$ is a witness function of $\opt$. For the first property, let $G \in \Graphs$, $X \in W$ and assume that $\WLcheck(G, X) > 0$. By (\ref{eq:walk}), we obtain that there is walk from $s$ to $t$ that uses the colors $X$. This walk induces a path $P$ from $s$ to $t$ that is obtained by deleting repeating vertices and edges of the walk. Note that for each $(n, c, \ell) \in X$ the path $P$ uses edges of color $(c, \ell)$ at most $n$ times. Further since $\WLcheck(G, X) > 0$, we obtain 
    \[\ell(P) \leq  \sum_{(n, c, \ell) \in X} n \cdot \ell \leq B \quad \text{ and } \quad c(P) \leq \sum_{(n, c, \ell) \in X} n \cdot c  = \WLcost(X).\]
    Thus, $P$ is a feasible path of length at most $B$. Further, $\opt(G) \leq \WLcost(X)$ since $P$ has cost at most $\WLcost(X)$.

    Next, assume that $\opt(G) < \infty$ then there is a $s, t$-path $P^\ast$ with $\ell(P^\ast) \leq B$ and $\opt(G) = c(P^\ast)$. Now, let $X$ be the set that contains $(n, c, \ell)$ if and only if the edge color $(c, \ell)$ is used $n$ times in $P^\ast$. Observe that $\sum_{(n, c, \ell) \in X} n \cdot \ell = \ell(P^\ast) \leq B$. Further note that $\opt(G) = c(P^\ast) = \sum_{(n, c, \ell) \in X} n \cdot \ell  = \WLcost(X)$. Since $P^\ast$ is path $s,t$-path and therefore a walk, (\ref{eq:walk}) yields that $\WLcheck(G, X) \geq 1 > 0$. Thus, $\opt(G) = \WLcost(X)$ and $\WLcheck(G, X) > 0$.
\end{proof}

The proof of the following lemma was discussed in \cite[Section 1.1]{DBLP:conf/wg/Boker19}. The lemma is a generalization of Theorem 1 in \cite{DBLP:conf/icalp/DellGR18}.

\begin{lemma}\label{lem:hom:indist}
    Let $F$ be a tree and $H, G$ be two graphs. If $H \WLequiv{1}G$ then $\homs{F}{G} = \homs{F}{H}$.
\end{lemma}

\begin{lemma}\label{lem:csp:witness:inv}
    The function $\WLcheck \colon \Graphs \times W \to \mathbb{R}$ is 1-WL invariant.
\end{lemma}
\begin{proof}
    Let $G, H \in \Graphs$ be two graphs with $G \WLequiv{1} H$.
    According to \cref{lem:hom:indist}, $\homs{F}{G} = \homs{F}{H}$ for all trees $F$. Since every path is a tree, we obtain that $\WLcheck(G, X) = \WLcheck(H, X)$ for all $X \in W$ that satisfy the length constraint, that is, $\sum_{(n, c, \ell) \in X} n \cdot \ell \leq B$. Note that $G$ and $H$ have the same lenght constraint since $G \WLequiv{1} H$.
\end{proof}

\begin{proof}[Proof of \cref{theo:csp:wl:inv}]
    According to \cref{lem:csp:witness}, the triple $(W, \WLcheck \colon \Graphs \times W \to \mathbb{R}, \WLcost \colon W \to \mathbb{R})$ is a witness function of $\opt$. Further, \cref{lem:csp:witness:inv} shows that $\WLcheck$ is  1-WL invariant. Thus, \cref{lem:witness:inv:implies:inv} implies that $\csp$ is 1-WL invariant. 
\end{proof}

\begin{corollary}\label{cor:CSP}
    There is an algorithm that solves the $\csp$ problem with $k$ prior instances each of size at most $n$ in  $O(k (n+m) \log n)$ preprocessing time and $O((n+m) \log n)$ query time. 
\end{corollary}
\begin{proof}
    This is a direct consequence of \cref{theo:1WL:iso} combined with \cref{theo:csp:wl:inv}.
\end{proof}

\subsection{$\cst$ is $2$-WL Invariant}\label{sec:WL:CST}

Next, we again turn to the $\cstName$ ($\cst$) problem.
Recall that, in \Cref{cor:constrained_spanning_tree_preprocessing}, we obtained a randomized algorithm in the Isomorphic-Priors model with $\widetilde{O}(kn^\omega)$ preprocessing time and $O(n^\omega)$ query time.
In the following, we obtain a deterministic algorithm via 2-WL.
Write $\opt$ for the optimal value function of $\cst$ and $\Graphs$ for the set of valid input graphs of $\cst$. We show that $\opt$ is 2-WL invariant for all values.\footnote{As before, we say that two input graphs $G$ and $H$ are 2-WL equivalent if $G$ and $H$ have the same budget $B$ value and $G \WLequiv{2} H$ for the underlying attributed graphs.} This together with \cref{theo:WL:iso} this immediately yields a deterministic algorithm for $\cst$  in the Isomorphic-Priors model.

\begin{theorem}\label{theo:cst:wl:inv}
    The problem $\cst$ is 2-WL invariant.
\end{theorem}

\paragraph*{Definition of a Witness:}
To prove \cref{theo:cst:wl:inv} we define a witness for $\opt$ and show that this witness is 2-WL invariant.

\begin{itemize}
    \item $W = \{X \, | \, X \subseteq \mathbb{N}_{\geq 0} \times \mathbb{R}_{\geq 0} \times \mathbb{R}_{\geq 0} \}$, 
    \item $\WLcost(X) = \sum_{(n, c, d) \in X } \, n \cdot c$,
    \item Let $T$ be a graph. We say that $T$ is a \emph{$X$-colored tree} if 
    \begin{itemize}
        \item $T$ is a tree where each vertex has the same color,
        \item each edge in $T$ is colored with a color $(c, d)$, and 
        \item for each $(n, c, d) \in X$ the tree $T$ contains exactly $n$ edges of color $(c, d)$.
    \end{itemize}
    We define 
    \[\WLcheck(G, X) = \left|\left\{T : T \text{ is a spanning tree of $G$ and $T$ is $X$-colored and $\!\!\sum_{(n, c, d) \in X } \, n \cdot d \leq B$}\right\} \right|,\]
    where $B$ is the budget associated with the input graph $G$. This means that $\WLcheck(G, X) = 0$, whenever $X$ is outside the budget constraint of $B$ 
\end{itemize}

\begin{lemma}\label{lem:cst:witness}
    The triple $(W, \WLcheck \colon \Graphs \times W \to \mathbb{R}, \WLcost \colon W \to \mathbb{R})$ is a witness function of $\opt$.
\end{lemma}
\begin{proof}
    Let $G \in \Graphs$ and $X \in W$ and assume that $\WLcheck(G, X) > 0$. By the definition of $\WLcheck(G, X)$, we obtain that there is $X$-colored spanning tree $T$ of $G$. For each $(n, c, d) \in X$, the tree $T$ uses edges of color $(c, d)$ exactly $n$ times. Further, since $\WLcheck(G, X) > 0$ we obtain 
    \[\ell(T) =  \sum_{(n, c, d) \in X} n \cdot d \leq B \quad \text{ and } \quad c(T) = \sum_{(n, c, d) \in X} n \cdot c  = \WLcost(X).\]
    Thus, $T$ is a feasible tree of length at most $B$. Additionally, $\opt(G) \leq \WLcost(X)$ since $T$ has cost $\WLcost(X)$.

    Next, assume that $\opt(G) < \infty$. Now, there is a spanning tree $T^\ast$ with $\sum_{(n, c, d) \in X} n \cdot d  \leq B$ and $\opt(G) = c(T^\ast)$. Let $X$ be the set that contains $(n, c, d)$ if and only if the edge color $(c, d)$ is used $n$ times in $T^\ast$. Note that $\opt(G) = c(T^\ast) = \sum_{(n, c, d) \in X} n \cdot d  = \WLcost(X)$. Since $T^\ast$ is by definition a $X$-colored spanning tree, we obtain that $\WLcheck(G, X) > 0$. Thus, $\opt_\cst(G) = \WLcost(X)$ and $\WLcheck(G, X) > 0$.
\end{proof}

We prove the following lemma which together with \cref{lem:cst:witness} immediately yields \cref{theo:cst:wl:inv}.

\begin{lemma}\label{lem:cst:witness:inv}
    The function $\WLcheck \colon \Graphs \times W \to \mathbb{R}$ is 2-WL invariant on $\Graphs$.
\end{lemma}
\begin{proof}[Proof (sketch).]
    
To prove \cref{lem:cst:witness:inv} we show that $\WLcheck(G, X) = \WLcheck(H, X)$ for all $G, H \in \Graphs$ with $G \WLequiv{2} H$ and $X \in W$. Without loss of generality we assume that $X$ satisfies $\sum_{(n, c, d) \in X } \, n \cdot d \leq B$ since otherwise $0 = \WLcheck(G, X) = \WLcheck(H, X) = 0$. We start with the Kirchhoff polynomial (see \cite{Chung2000})
    \[\mathrm{Kir}(G) \coloneqq \sum_{T} \prod_{e\in T} X_e,\]
    where each edge $e$ in $G$ has its own variable $X_e$ and the sum is over all spanning trees $T$ of $G$. If we substituted each $X_e$ with a new variable $c(e)$, where $c(e)$ is the color of the edge $e$, then we get the \emph{Spanning Tree Polynomial} $P_{ST}(G)$ from \cref{sec:CST}.  Formally, we have $\mathrm{Kir}(G)|_{X_e = c(e)} \equiv P_{ST}(G)$, where $\mathrm{Kir}(G)|_{X_e = c(e)}$ is obtained by substituting $X_e$ with $c(e)$ in $\mathrm{Kir}(G)$. 
    
    Further, note that the function $\WLcheck(G, X) $ computes the coefficients of the \emph{Spanning Tree Polynomial}. Specially, the coefficient $c_{w_1, \dots w_{n-1}}$ of $P_{ST}(G)$ is equal to $\WLcheck(G, X)$ where $X$ is the set that contains a triple $(n, c, d)$ if and only if the color $(c, d)$ appears exactly $n$ times in the multiset $\{w_1, \dots w_{n-1}\}$ (note that each $w_i$ is actually a $(c, d)$-tuple). Thus, if $P_{ST}(G) \equiv P_{ST}(H)$ then $\WLcheck(G, X) = \WLcheck(H, X)$ for all $X \in W$, proving the theorem.
    
    To show $P_{ST}(G) \equiv P_{ST}(H)$, we use a result from \cite{Chung2000}
    which states that for each $i \in [n]$
    \[\mathrm{Kir}(G) \equiv \det(\tilde{L}_i(G)),\]
    where $\tilde{L}_i(G)$ is the Laplacian $L(G)$ of $G$ with its $i$-th row and $i$-th column removed. Note that $\det(\tilde{L}_i(G))$ is a polynomial over the variables $X_e$.
    Now, $\mathrm{Kir}(G)|_{X_e = c(e)} \equiv P_{ST}(G)$ and $\mathrm{Kir}(G) \equiv \det(\tilde{L}_i(G))$ implies that $P_{ST}(G) \equiv \det(\tilde{L}^C_i(G))$ where $\tilde{L}^C_i(G)$ is obtained by substituted each variable $X_e$ in $\tilde{L}_i(G)$ with $c(e)$. 
    
    Thus, it is enough to prove $P_{ST}(G) \equiv \det(\tilde{L}^C_i(G)) \equiv \det(\tilde{L}^C_i(H)) \equiv P_{ST}(H)$ whenever $H \WLequiv{2} G$. By \cite{RattanS23}\footnote{See Theorem 26 combined with Example 22 that shows that Laplacian mapping $G \to L(G)$ is in $\mathfrak{E}$.}, $L^C(G)$ is cospectral to $L^C(H)$. Let $d^G_i$ (respectively $d^H_i$) be the $i$th coefficient of the characteristic polynomial of $L^C(G)$ (respectively $L^C(H)$). Note that $d_i^H, d_i^G \in \mathbb{C}[c(e)]$.\footnote{That is, the ring of polynomials with variables $c(e)$.} Since $L^C(G)$ and $L^C(H)$ are cospectral, it follows that $d^H_i = d^G_i$. Note that $d^H_1$ is the linear coefficient of the characteristic polynomial.\footnote{That is, the coefficient corresponding to $x^1$ in the characteristic polynomial of $L^C(H)$.} A well-known theorem from linear algebra (see e.g. \cite[Theorem 1.8]{matrixBook}) implies
    \[d^H_1 \equiv -\sum_{i = 1}^n \det(\tilde{L}^C_i(H)), \qquad d^G_1 \equiv -\sum_{i = 1}^n \det(\tilde{L}^C_i(G)) . \]
    Now, since $\det(\tilde{L}^C_i(H)) \equiv \det(\tilde{L}^C_j(H))$ and $\det(\tilde{L}^C_i(G))\equiv= \det(\tilde{L}^C_j(G))$ for all $i, j \in [n]$, we obtain
    \[ P_{ST}(G) \equiv \det(\tilde{L}^C_i(G)) \equiv \frac{-d^G_1}{n} \equiv \frac{-d^H_1}{n} \equiv \det(\tilde{L}^C_i(H)) \equiv P_{ST}(H). \]
\end{proof}

\paragraph{}

\begin{proof}[Proof of \cref{theo:cst:wl:inv}]
    Let $B \in \mathbb{R}_{> 0}$. 
    According to \cref{lem:cst:witness}, the triple $(W, \WLcheck \colon \Graphs \times W \to \mathbb{R}, \WLcost \colon W \to \mathbb{R})$ is a witness function of $\opt$. Further, \cref{lem:cst:witness:inv} shows that $\WLcheck$ is  2-WL invariant on $\Graphs$. Thus, \cref{lem:witness:inv:implies:inv} implies that $\csp$ is 2-WL invariant. 
\end{proof}

\begin{corollary}\label{cor:CST:WL}
    There is an algorithm that solves the $\cst$ problem with $k$ prior instances each of size at most $n$ in  $O(k n^{3} \log n)$ preprocessing time and $O( n^{3} \log n)$ query time. 
\end{corollary}
\begin{proof}
    This is a direct consequence of \cref{theo:WL:iso} combined with \cref{theo:cst:wl:inv}.
\end{proof}

\subsection{Examples when Weisfeiler-Leman Fails}\label{sec:WL:fails}

Given a graph $G$ of order $n$ and maximum degree $k \geq 3$.
In \cite{DBLP:journals/combinatorica/CaiFI92} Cai, Fürer, and Immerman introduced the idea of CFI graphs. A CFI graph is obtained by taking $G$, a parity $p \in \{\odd, \even\}$ and computing a labeled graph $\cfiGraph{G}{p}$. This construction can be performed in time $O(2^k \cdot n)$.\footnote{See \cite[Section 6]{DBLP:journals/combinatorica/CaiFI92} for a detailed construction. Note that in this paper $X(G) = \cfiGraph{G}{\even}$ and $\tilde{X}(G) = \cfiGraph{G}{\odd}$.} Further, $\cfiGraph{G}{p}$ has a maximum degree at most $2^{k-2} + 1$. 

\begin{lemma}\label{lem:cfi:iso}
    Let $k \geq 3$.
    There is an algorithm that given a $k$-regular graph $G$, a parity $p \in \{\odd, \even\}$, and a labeled graph $H$ checks whether $\cfiGraph{G}{p}$ and $H$ are isomorphic. Further the algorithm runs in time $O(|V(G)|^{f(k)})$ for some computable function $f$.
\end{lemma}
\begin{proof}
    By construction the graph $\cfiGraph{G}{p}$ has a maximum degree of $2^{k-2} + 1$. By \cite{DBLP:journals/siamcomp/GroheNS23}, isomorphism testing on unlabeled graphs with maximum degree $2^{k-2} + 1$ can be done in time $O(n^{\log(2^{k-2} + 1)^c}) \subseteq O(n^{k^c})$ for a fixed constant $c$ that is independent of $k$. By \cite{booth1979problems}, the problem of finding isomorphism in unlabeled graphs is polynomially time equivalent to the problem of finding isomorphism in labeled graphs. Thus, there is an algorithm that runs in time $O(|V(G)|^{f(k)})$, for some computable function $f$, and checks if a labeled graph $H$ is isomorphic to $\cfiGraph{G}{p}$.
\end{proof}

Write $\Graphs$ for the set if of all graphs.

\begin{lemma}\label{lem:cfiFunc}
    There is a function $\cfiFunc \colon \Graphs \to \{0, 1\}$ with the following properties:
    \begin{itemize}
        \item The function $\cfiFunc$ is computable in polynomial time: Given a colored graph $G$ or order $n$, we can compute $\cfiFunc(G)$ in time $O(n^c)$ for some fixed constant $c$.
        \item For all isomorphic graphs $G$ and $H$ we obtain $\cfiFunc(G) = \cfiFunc(H)$.
        \item For all $\ell$ there are graphs $G, H$ of order $\Theta(\ell)$ with $G \WLequiv{\ell} H$ and $\cfiFunc(G) \neq \cfiFunc(H)$. 
    \end{itemize}
\end{lemma}
\begin{proof}
    By \cite{DBLP:journals/combinatorica/Alon21}, for each $n$ we can construct a $4$-regular expander graph $E_n$ with $n$ vertices. To be more specific, there is an algorithm that given $n$ constructs $E_n$ in time $O(n^d)$ for same constant $d$. Further, is well-known that expander graphs have linear separator.\footnote{A separator of a graph $G$  is a subset $S \subseteq V(G)$ such that the induced subgraph with vertex set $V(G) \setminus S$ has no connected component with more than $V(G)/2$ vertices.} Formally, there is a global constant $\alpha > 0$ such that each separator of $E_n$ has at least $\alpha n$ many vertices. We use this to define the following property on graphs: 
    \[\cfiFunc \colon \Graphs \to \{0, 1\}, G \mapsto \begin{cases}
        1 & \text{if there is an $n'$ with $G \cong \cfiGraph{E_{n'}}{\odd}$} \\
        0 & \text{otherwise}
    \end{cases}.\] 
    Next, we check that $\cfiFunc$ has the promised properties
    \begin{itemize}
        \item Let $G$ be an input graph of order $n$. We show how to compute $\cfiFunc(G)$. For each $n'$ the graph $\cfiGraph{E_{n'}}{\odd}$ has $(4 + 4+ 2^3)n' = 16n'$ many vertices. Thus, we compute $n' = n/16$. If $n'$ is not an integer then we return $0$ since $G$ cannot be isomorphic to any $\cfiFunc(E_{m}, \odd)$. Otherwise, we use  \cite{DBLP:journals/combinatorica/Alon21} to compute $E_{n'}$ in time $O(n^d)$ and then construct $\cfiGraph{E_{n'}}{\odd}$ in time $O(n)$. Note that $G$ is either isomorphic to $\cfiGraph{E_{n'}}{\odd}$ or $G$ is non-isomorphic to every $\cfiFunc(E_{m}, \odd)$. Next, since $E_{n'}$ is by construction 4-regular \cref{lem:cfi:iso} yields that we can check $G \cong \cfiGraph{E_{n'}}{\odd}$ in time $O(n^{d'})$, for some fixed integer $d'$. If this is then case then return 1. Otherwise, we return 0.
        By setting $c = \max(d, d')$, we obtain that $\cfiFunc(G)$ can be computed in time $O(n^c)$.
        
        \item If $G$ and $H$ are isomorphic then $G \cong \cfiGraph{E_{n'}}{\odd}$ if and only if $H \cong \cfiGraph{E_{n'}}{\odd}$. Thus, $\cfiFunc(G) = \cfiFunc(H)$.
        \item Let $\ell \in \mathbb{N}$ and $\alpha > 0$ be a constant such that  each separator of $E_n$ has at least $\alpha n$ many vertices. 
        Let $n \coloneqq \lceil (\ell + 2) /\alpha \rceil$. 
        We show that the graphs $\cfiGraph{E_{n}}{\odd}$ and $\cfiGraph{E_{n}}{\even}$ are sufficient. 
        Note that both graphs have $\Theta(\ell)$ many vertices. 
        Further, by construction each separator of $E_n$ has at least 
        $\ell + 1$ many vertices. Thus, \cite[Theorem 6.4]{DBLP:journals/combinatorica/CaiFI92} 
        implies $\cfiGraph{E_{n}}{\odd} \WLequiv{\ell} \cfiGraph{E_{n}}{\even}$.\footnote{By \cite[Theorem 6.4]{DBLP:journals/combinatorica/CaiFI92}, we obtain $\cfiGraph{E_{n}}{\odd} \equiv_{\mathcal{C}_{\ell+1}} \cfiGraph{E_{n}}{\even}$. By \cite[Theorem 5.2]{DBLP:journals/combinatorica/CaiFI92}, 
        this is equivalent to $\cfiGraph{E_{n}}{\odd} \WLequiv{\ell} \cfiGraph{E_{n}}{\even}$.} Further, \cite[Theorem 6.2]{DBLP:journals/combinatorica/CaiFI92} yields that $\cfiGraph{E_{n}}{\odd}$ and $\cfiFunc(E_n, \even)$ are non-isomorphic. 
        Hence, $\cfiFunc(\cfiGraph{E_{n}}{\odd}) = 1$ and $\cfiFunc(\cfiGraph{E_{n}}{\even}) = 0$.
    \end{itemize}
\end{proof}

\begin{theorem}\label{theo:CFI:problem}
    There exist an decision problem $\Pi = (\Graphs, \cfiFunc \colon \Graphs \to \{0, 1\})$ such that  $\Pi$ is polynomial time solvable and
    for all $\ell$ the problem $\Pi$ is not \emph{$\ell$-WL invariant}.
\end{theorem}
\begin{proof}
    Let $\Pi$ be the problem that given a graph $G$ decides if $\cfiFunc(G) = 1$. By \cref{lem:cfiFunc}, $\Pi$ is polynomial time solvable but for all $\ell$ the problem $\Pi$ is not \emph{$\ell$-WL invariant}.
\end{proof}

\cref{theo:WL:iso} shows that if a problem is $\ell$-WL invariant than it is polynomial time solvable inside the Isomorphic-Priors
model. However, \cref{theo:CFI:problem} demonstrates that there are problems
that are in $\cP$ (and that are therefore polynomial time solvable inside the Isomorphic-Priors
model) but for which our Weisfeiler-Leman approach fails. 

\begin{remark}
    The original CFI graph construction can be viewed as a graph-theoretic encoding of parity information. More specifically, the underlying idea of CFI graphs is closely related to the solvability of linear systems over the finite field $\mathbb{F}_2$, with the gadgets enforcing local parity constraints (see \cite{DBLP:journals/corr/abs-1204-3022}). It follows that solving linear equations over $\mathbb{F}_2$ is not invariant under $\ell$-WL for any $\ell$. This therefore yields another polynomial-time solvable problem that cannot be captured by Weisfeiler-Leman.
\end{remark}

We note that, using similar constructions, the same can be proved for several other hard problems, e.g., for counting certain (induced) subgraphs \cite{Neuen24a,LanzingerB24, DBLP:conf/stoc/DoringMW24, CurticapeanN25, DBLP:conf/soda/DoringMW25} or dominating set \cite{GobelGR24}.

\section{Conditional Lower Bounds}\label{sec:Hardness}\label{sec:Lowerbounds}

In this section, we prove different hardness results for the Isomorphic-Priors model. As our source of hardness, we use the $\giName$ problem.

\begin{definition}[$\giName$ ($\GI$)]
    In the $\giName$ problem we are given a pair of graphs $I = (G_0, G_1)$ and have to decide if $G_0 \cong G_1$.

    Two problem instances $I = (G_0, G_1)$ and $J = (H_0, H_1)$ are isomorphic if $G_0 \cong H_0$ and $G_1 \cong H_1$.
\end{definition}

Up to this date, it is not known whether $\GI$ is in $\cP$ despite a lot of effort from many researchers \cite{DBLP:journals/sigact/Kozen78}. We can therefore use $\GI$ as a source of hardness and show that many problems are not polynomial time solvable inside the Isomorphic-Priors model unless $\GI \in P$, which would be a breakthrough result. 

Furthermore, it is not known if $\GI$ is in  $\coma$.\footnote{The complexity class $\mathsf{MA}$ consists of all languages $L$ for which there exist a polynomial-time deterministic Turing machine $M$ and polynomials $p$, $q$ such that for every input string $x$ of length $n \coloneqq |x|$,
\begin{itemize}
    \item if $x \in L$ then there exist a $z \in \{0, 1\}^{q(n)}$ with $\Pr_{y \in \{0, 1\}^{p(n)}}[M(x, y, z) =  1] \geq 2/3$,
    \item if $x \notin L$ then for all $z \in \{0, 1\}^{q(n)}$ we obtain $\Pr_{y \in \{0, 1\}^{p(n)}}[M(x, y, z) =  0] \geq 2/3$.
\end{itemize}
The class $\coma$ consists of all languages whose complement is in $\mathsf{MA}$.} Again, we can use this as a source of hardness. We show in \cref{sec:GI:hard} that $\GI$ is not polynomial time solvable within our model unless $\GI \in \coma$.

Next, we are utilizing the hardness of $\GI$ to obtain further hardness results within our model for ``normal'' decision problems $\Pi$. To this end, we associate each problem $\Pi$ with a specific promise version of the same problem that we denote as $\mathrm{Iso}\text{-}\Pi$. See \cref{sec:iso:prob} for a formal definition. We start by showing that for $\GI$ the new problem $\mathrm{Iso}\text{-}\GI$ is $\GI$-complete (i.e., as hard as $\GI$). Therefore, $\mathrm{Iso}\text{-}\GI$ can be used as a source hardness within our model.

Having a source of hardness, we introduce in \cref{sec:isoreduction} the notion of an \emph{isomorphism-preserving reduction}. This new kind of reduction allows us to propagate the hardness results from $\mathrm{Iso}\text{-}\GI$ to other $\mathrm{Iso}\text{-}\Pi$ problems as long as there exists an isomorphism-preserving reduction between $\GI$ and $\Pi$. We then use isomorphism-preserving reductions to obtain hardness results for each problem $\mathrm{Iso}\text{-}\Pi$ where $\Pi$ is an element of
\begin{align*}
    \ProbCollection \coloneqq  \{&\LPZEROONE,\; \GAP{\kCenter}_{1, 2}, \Clique, \; \CliqueCover, \; \COL,\; \kCOL{3}, \\
    \; 
    &\DMatching, \; 
    \EXA, \; \FeedbackArcSet, \; \FeedbackVertexSet, \;
    \DirHam, \; \\
    &\UndirHam, \; \HIT,\; \IS,\;  \MaxCut,\; \kMedian, \; \SAT,  \; \kSAT{3},\; \MaxkSAT{2}, \;  \\
    &\SetCover, \; \SetPacking, \; \STEINER, \; \GAP{\TSP_{1, c}}, \; \MetricTSP, \;
    \VC\}.
\end{align*}
The problem definitions can be found in \cref{app:hardness:defitions}. The hardness of most problems can be shown by utilizing standard reductions from the literature. However, to obtain hardness of $\kSAT{3}$, we need to take a non-standard route through $\kCOL{3}$. For this, we need a reduction by Lovász. See \cref{sec:3sat} for more details.
Finally, utilizing the hardness of $\mathrm{Iso}\text{-}\GI$ again, we then show in \cref{sec:hard:iso:pro} that each problem in $\ProbCollection$ cannot be solved efficiently in our model. To be more specific, we obtain two different kind of hardness results.

We show that no problem from $\ProbCollection$ can be solved with fast preprocessing and queries, unless $\GI \in \cP$.
\begin{remark}
Note that in the following theorem, we state only the dependence of
the running time on $n$, although it may also depend on $k$.
Here, ``polynomial in $n$'' means polynomial in $n$ with $k$ held fixed.
The same convention applies to space requirements and to the remaining theorems in this section.
\end{remark}
\begin{theorem}\label{cor:hardness-consequence}
Unless $\GI \in \cP$, for every problem $\Pi \in \ProbCollection$,
the problem $\Pi$ cannot be solved inside the Isomorphic-Priors model with $k$ prior instances each of size at most $n$ with preprocessing and query time polynomial in 
$n$, even when $k=1$ prior instance is given.
\end{theorem}

Under the stronger assumption $\GI \notin \coma$, we can also rule out a data structure of polynomial size and query time.
\begin{theorem}\label{thm:ds-lb}
Unless $\GI \in \coma$, for every problem $\Pi \in \ProbCollection$, the problem $\Pi$ cannot be solved inside the Isomorphic-Priors model with $k$ prior instances each of size at most $n$ with a data structure of size polynomial in $n$, and query time polynomial in $n$, even when $k=1$ prior instance is given.
\end{theorem}

After showing hardness for decision problems, we move on in \cref{subsec:proofhardnessofapprox} to show different hardness results for approximation problems.
To this end, we first show that
$\mathsf{Max\text{-}}k\mathsf{\text{-}Coverage}$ cannot be efficiently approximated better than $1-1/e$ in our model, unless $\MaxkSAT{3}$ can be approximated arbitrarily well.
The assumption on $\MaxkSAT{3}$ serves as our analogue to the PCP theorem.
We have not been able to formulate the PCP theorem directly, because it is unclear how to define an appropriate notion of isomorphism for probabilistically checkable proofs.

\begin{theorem}\label{thm:fastlb-approx}
Assume that there is an $\eps>0$ such that $\mathsf{Gap\text{-}Max\text{-}3\text{-}SAT}_{1,\,1-\eps}$ cannot be solved in the Isomorphic-Priors model with $k$ prior instances each of size at most $n$ with preprocessing time polynomial in $k$ and $n$
and query time polynomial in $n$, even when $k=1$ prior instance is given.

Then there is an $\eps'>0$ such that $\mathsf{Max\text{-}}k\mathsf{\text{-}Coverage}$ cannot be approximated within 
 $(1-1/e+\eps')$
in the Isomorphic-Priors model with preprocessing and query time polynomial in $n$, even when $k=1$ prior instance is given.
Also, there is an $\eps'' >0$ such that $\kMedian$ cannot be approximated within $(1+2/e-\eps'')$
in the Isomorphic-Priors model with preprocessing and query time polynomial in $n$, even when $k=1$ prior instance is given.
\end{theorem}

The analogous implication can be shown for approximating $\mathsf{Max\text{-}}k\mathsf{\text{-}Coverage}$ with a small data structure.
\begin{theorem}\label{thm:smalllb-approx}
Assume that there is an $\eps>0$ such that $\mathsf{Gap\text{-}Max\text{-}3\text{-}SAT}_{1,\,1-\eps}$ cannot be solved in the Isomorphic-Priors model with $k$ prior instances each of size at most $n$
and query time polynomial in $n$, even when $k=1$ prior instance is given.

Then there is an $\eps'>0$ such that $\mathsf{Max\text{-}}k\mathsf{\text{-}Coverage}$ cannot be approximated within $(1-1/e+\eps')$
in the Isomorphic-Priors model with data structure of size polynomial in $n$ and query time polynomial in $n$, even when $k=1$ prior instance is given.
Also, there is an $\eps''$ such that $\kMedian$ cannot be approximated better than $(1+2/e-\eps'')$
in the Isomorphic-Priors model with data structure of size polynomial in $n$ and query time polynomial in $n$, even when $k=1$ prior instance is given.
\end{theorem}

\subsection{Hardness Results for $\GI$}\label{sec:GI:hard}

We show that GI cannot be solved in our setting with a small data structure, unless $\GI \in \mathsf{coMA}$. The following Lemma is inspired by \cite[Corollary 4.4]{ LozanoToran1992NonUniformGI}. 

\begin{lemma}\label{lem:gi-ds-coMA}
Unless $\GI \in \mathsf{coMA}$, $\GI$ cannot be solved in the Isomorphic-Priors model with $k$ prior instances each having at most $n$ many vertices with a data structure of polynomial size in $n$
and query time polynomial in $n$, even when $k=1$ prior instance is given.

This statement is even true for a randomized data structure such that the query is correct with probability $5/6$.
\end{lemma}

\begin{proof}
Assume that there is an algorithm consisting of the preprocessing phase $\textsf{Pre}$ and query phase $\textsf{Query}$ in the Isomorphic-Priors model
that solves $\GI$ with a data structure of size $\poly_k(n)$\footnote{This means that the size is polynomial in $n$ if $k$ is fixed.} and query time $\poly_k(n)$, and in particular
works for $k=1$. 

We show that this implies $\overline{\GI}\in \mathsf{MA}$, and hence $\GI\in\mathsf{coMA}$.

\paragraph{An $\mathsf{MA}$ protocol for $\overline{\GI}$.}
On a tuple of input graphs $(G_0,G_1)$, Merlin sends Arthur a string $D$ of length $\poly(n)$, claimed to be a
valid data structure for the single prior instance
\[
I^\star := (G_0,G_0).
\]
Arthur repeats the following test twice 
\begin{enumerate}
    \item choose a uniform random bit $b\in\{0,1\}$ and a uniform random permutation $\pi$ of $[n]$,
    \item form the query instance
    \[
        J := (G_0,\ \pi(G_b)),
    \]
    \item compute $z := \textsf{Query}(D,J)$ and define $\widehat b:=0$ iff $z=1$ (otherwise $\widehat b:=1$),
    \item accept this round iff $\widehat b=b$.
\end{enumerate}
Arthur accepts iff the majority of the $t$ rounds accept.

\paragraph{Completeness.}
Assume $G_0\not\cong G_1$. Let Merlin send the actual data structure $D=\textsf{Pre}(\{I^\star\})$.
If $b=0$, then there exists an isomorphism $\varphi:G_0\to \pi(G_0)$, hence
$J=(G_0,\pi(G_0))$ is isomorphic to $I^\star$ via $(\mathrm{id},\varphi)$ and $\GI(J)=1$.
Thus correctness forces $\textsf{Query}(D,J)=1$, so $\widehat b=0=b$ with probability at least $5/6$.
If $b=1$, then $G_0\not\cong \pi(G_1)$, so $\GI(J)=0$ and moreover $J$ is not isomorphic to $I^\star$.
By correctness on non-matching queries, $\textsf{Query}(D,J)$ is either $0$ or $\bot$ with probability at least $5/6$, hence $z\neq 1$ and
therefore $\widehat b=1=b$ with probability at least $5/6$. Thus, each round accepts with probability at least $5/6$. Therefore, Arthur accepts with probability at least $25/36>2/3$.

\paragraph{Soundness.}
Assume $G_0\cong G_1$. Then the distributions of $\pi(G_0)$ and $\pi(G_1)$ are identical, hence the random bit $b$
is independent of the pair $(D,J)$, and therefore also independent of $\widehat b$ (which is a function of $(D,J)$ and
the internal randomness of $\textsf{Query}$).
Consequently, in each round we have $\Pr[\widehat b=b]\le \tfrac12$, and Arthur accepts with probability at most $1/4<2/3$.

Therefore $\overline{\GI}\in \mathsf{MA}$, which implies $\GI\in \mathsf{coMA}$.

By contrapositive, unless $\GI\in\mathsf{coMA}$, no such polynomial-size data structure algorithm for $\GI$
exists in the Isomorphic-Priors model (even for $k=1$).
\end{proof}

\subsection{Iso-Problems and $\GI$-Completeness}\label{sec:iso:prob}
We introduce the promise problem $\mathrm{Iso}\text{-}\Pi$ and show properties that will be used in later sections to show hardness results in our model.

\begin{definition}
For a decision problem $\Pi$, define the promise problem
$\mathrm{Iso}\text{-}\Pi=(\mathrm{Iso}\text{-}\Pi_{\textsc{YES}},\mathrm{Iso}\text{-}\Pi_{\textsc{NO}})$
over pairs of instances $(I,J)$ by
\[
\mathrm{Iso}\text{-}\Pi_{\textsc{YES}}:=\{(I,J)\mid I\in\Pi \ \wedge\ J\cong I\},
\qquad
\mathrm{Iso}\text{-}\Pi_{\textsc{NO}}:=\{(I,J)\mid I\in\Pi \ \wedge\ J\notin\Pi\}.
\]
The input is promised to lie in $\mathrm{Iso}\text{-}\Pi_{\textsc{YES}}\cup \mathrm{Iso}\text{-}\Pi_{\textsc{NO}}$. 
\end{definition}

We recall the definition of polynomial-time many-one reductions.
\begin{definition}
    We say that there is a \emph{polynomial-time many-one reduction} from $\Pi$ to $\Pi'$ if there exist a polynomial-time computable function $f$ such that $X \in \Pi$ if and only if $f(X) \in \Pi$.

    We say that $\Pi$ is $\GI$-hard if there is a polynomial-time many-one reduction from $\GI$ to $\Pi$. We say that $\Pi$ is $\GI$-complete if there $\Pi$ is $\GI$-hard and there is a polynomial-time many-one reduction from $\Pi$ to $\GI$.
\end{definition}

We show that for $\Pi = \GI$, $\mathrm{Iso}\text{-}\Pi$ is just a reformulation of $\GI$.
\begin{lemma}\label{lem:iso-gi-complete}
$\mathrm{Iso}\text{-}\GI$ is $\GI$-complete.
\end{lemma}
\begin{proof}
Recall that $\mathrm{Iso}\text{-}\GI$ is the promise problem on pairs $(I,J)$ of $\GI$-instances such that
$I\in \GI$ and either $J\cong I$ (YES) or $J\notin \GI$ (NO).

\paragraph{(1) $\mathrm{Iso}\text{-}\GI \in \GI$.}
Given an input $(I,J)$, where $I=(G_0,G_1)$ and $J=(H_0,H_1)$, run a $\GI$-test to decide whether
$I\cong J$, i.e.\ whether $G_0\cong H_0$ and $G_1\cong H_1$.
If $I\cong J$ output \textsc{YES}, otherwise output \textsc{NO}.
This is correct under the promise: in a \textsc{YES}-instance we have $J\cong I$, while in a promised
instance with $J\not\cong I$ we cannot be in $\mathrm{Iso}\text{-}\GI_{\textsc{YES}}$, hence must be in
$\mathrm{Iso}\text{-}\GI_{\textsc{NO}}$.

\paragraph{(2) $\GI$-hardness.}
We give a polynomial-time many-one reduction from $\GI$ to $\mathrm{Iso}\text{-}\GI$.
Given a $\GI$-instance $I=(G,H)$, output the pair
\[
f(G,H) \;:=\; \bigl( (G,G),\ (G,H) \bigr).
\]
Note that $(G,G)\in \GI$ always. If $G\cong H$, then $(G,H)\cong (G,G)$, so $f(G,H)\in
\mathrm{Iso}\text{-}\GI_{\textsc{YES}}$. If $G\not\cong H$, then $(G,H)\notin \GI$, so
$f(G,H)\in \mathrm{Iso}\text{-}\GI_{\textsc{NO}}$. Hence $f$ is a valid reduction and $\GI \le_m
\mathrm{Iso}\text{-}\GI$.

\paragraph{Conclusion.}
We have shown $\mathrm{Iso}\text{-}\GI\in \GI$ and $\GI\le_m \mathrm{Iso}\text{-}\GI$, therefore
$\mathrm{Iso}\text{-}\GI$ is $\GI$-complete.
\end{proof}

We conclude this section by showing that the promise problem is at most as hard as $\GI$ independent of the problem it is based on.
\begin{lemma}\label{lem:isoPi-in-GI}
Let $\Pi$ be a decision problem. Then $\mathrm{Iso}\text{-}\Pi$ is in $\GI$.
\end{lemma}

\begin{proof}
We give a polynomial-time many-one reduction from $\mathrm{Iso}\text{-}\Pi$ to $\GI$.

On input $(I,J)$, output the pair $(I,J)$ as an instance of $\GI$ (i.e., the question whether $I\cong J$).
This mapping is polynomial time.

We claim it is correct under the promise defining $\mathrm{Iso}\text{-}\Pi$.
If $(I,J)\in \mathrm{Iso}\text{-}\Pi_{\textsc{YES}}$, then $J\cong I$, hence $(I,J)$ is a \textsc{YES}-instance of $\GI$.
If $(I,J)\in \mathrm{Iso}\text{-}\Pi_{\textsc{NO}}$, then $I\in\Pi$ and $J\notin\Pi$.
Since $\Pi$ is isomorphism-invariant, $J\cong I$ would imply $J\in\Pi$, a contradiction; hence $J\not\cong I$,
and $(I,J)$ is a \textsc{NO}-instance of $\GI$.

Therefore, deciding $(I,J)\in \GI$ decides $(I,J)\in \mathrm{Iso}\text{-}\Pi$, and thus
$\mathrm{Iso}\text{-}\Pi \le_m \GI$, i.e.\ $\mathrm{Iso}\text{-}\Pi$ is in $\GI$.
\end{proof}

\subsection{\texorpdfstring{Isomorphism-Preserving Reductions}{Isomorphism-Preserving Reductions}}\label{sec:isoreduction}

A tool we use to show the hardness is isomorphism-preserving reductions
\begin{definition}\label{def:iso:preserving:redution}
We call a function $R$ that maps instances of $\Pi$ to instances of $\Pi'$ an \emph{isomorphism-preserving reduction} if:
\begin{itemize}
    \item $R$ is polynomial time computable.
    \item $R$ is \emph{correct}, that is $I \in \Pi$ is a YES-instance of $\Pi$ if and only if $R(I)$ is a YES-instance of $\Pi'$.
    \item $R$ is \emph{isomorphism-preserving}, that is if $I_1,I_2 \in \Pi$ are isomorphic in $\Pi$ then
    $R(I_1),R(I_2)$ are isomorphic in $\Pi'$.
\end{itemize}
We write $\Pi \isored \Pi'$ if such a reduction exits. Further, we say that a problem $\Pi'$ is $\GI$-hard if $\GI \isored \Pi'$.
\end{definition}

The property of isomorphism-preserving reductions is that they can be used to show $\GI$-hardness of promise problems.
\begin{lemma}\label{lem:iso-hardness-transfer}
Let $\Pi$ and $\Pi'$ be decision problems. If $\mathrm{Iso}\text{-}\Pi$ is $\GI$-hard and
$\Pi \isored \Pi'$, then $\mathrm{Iso}\text{-}\Pi'$ is $\GI$-hard.
\end{lemma}

\begin{proof}
Since $\mathrm{Iso}\text{-}\Pi$ is $\GI$-hard, there exists a polynomial-time reduction
$g:\GI \le_m \mathrm{Iso}\text{-}\Pi$.

It therefore suffices to give a polynomial-time reduction
$f:\mathrm{Iso}\text{-}\Pi \le_m \mathrm{Iso}\text{-}\Pi'$.

Let $(I,J)$ be an instance of $\mathrm{Iso}\text{-}\Pi$. Since $\Pi \isored \Pi'$, there exists an
isomorphism-preserving many-one reduction $R$ from $\Pi$ to $\Pi'$.
Define
\[
f(I,J) \;:=\; (R(I),\,R(J)).
\]

We check that $f$ maps promised instances to promised instances and preserves YES/NO.

\paragraph{YES case.}
If $(I,J)\in \mathrm{Iso}\text{-}\Pi_{\textsc{YES}}$, then $I\in \Pi$ and $J\cong I$.
By correctness, $R(I)\in \Pi'$, and by isomorphism preservation, $R(J)\cong R(I)$.
Hence $f(I,J)\in \mathrm{Iso}\text{-}\Pi'_{\textsc{YES}}$.

\paragraph{NO case.}
If $(I,J)\in \mathrm{Iso}\text{-}\Pi_{\textsc{NO}}$, then $I\in \Pi$ and $J\notin \Pi$.
By correctness of $R$, we have $R(I)\in \Pi'$ and $R(J)\notin \Pi'$.
Hence $f(I,J)\in \mathrm{Iso}\text{-}\Pi'_{\textsc{NO}}$.

Thus $f$ is a valid many-one reduction $\mathrm{Iso}\text{-}\Pi \le_m \mathrm{Iso}\text{-}\Pi'$.
By transitivity,
\[
\GI \le_m \mathrm{Iso}\text{-}\Pi \le_m \mathrm{Iso}\text{-}\Pi',
\]
so $\mathrm{Iso}\text{-}\Pi'$ is $\GI$-hard.
\end{proof}

Thus, by using isomorphism-preserving reductions, we can obtain $\GI$-completeness for many different $\mathrm{Iso}\text{-}\Pi$ problems. To this end, we use the reductions shown in \Cref{fig:reductions}. We use this property to show that the promise problems based on the problems in $\ProbCollection$ are $\GI$-complete.

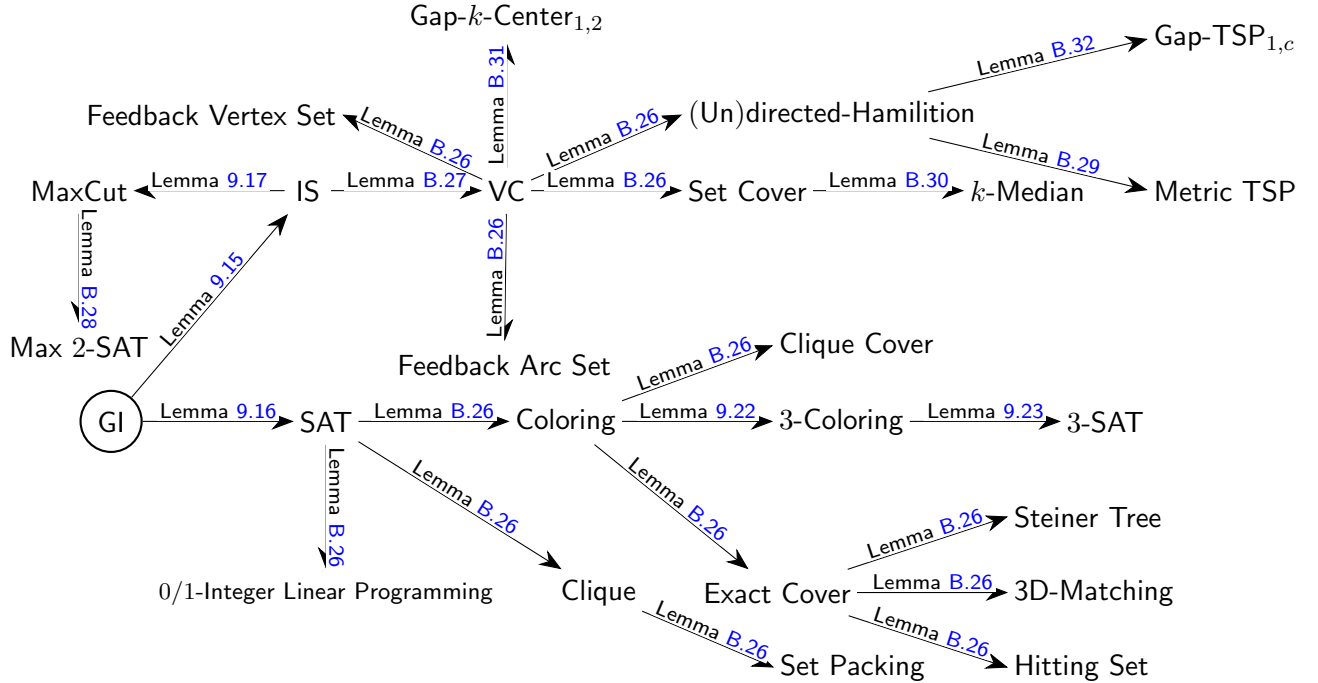
\begin{figure}[h]
    \centering
    \begin{tikzpicture}[
    >={Stealth[length=3mm]},
    node distance=1.65cm and 2.0cm,
    reduction/.style={rectangle, minimum size=6mm, inner sep=2pt, font=\sffamily},
    gi_node/.style={circle, draw, thick, minimum size=8mm, font=\sffamily\bfseries},
    edge_label/.style={font=\scriptsize\sffamily, fill=white, inner sep=1pt}
]

    \node[gi_node] (GI) {$\GI$};
    \node[reduction, right=of GI] (SAT) {$\SAT$};
    \node[reduction, right=of SAT] (Coloring) {$\COL$};
    
    \node[reduction, right=of Coloring] (3Coloring) {$\kCOL{3}$};
    \node[reduction, right=of 3Coloring] (3SAT) {$\kSAT{3}$};
    \node[reduction, right=of Coloring, yshift=10mm] (CliqueCov) {$\CliqueCover$};
    
    \node[reduction, above right=of GI, yshift=8mm] (IS) {$\IS$};
    \node[reduction, left=of IS] (MaxCut) {$\MaxCut$};
    \node[reduction, right=of IS] (VC) {$\VC$};
    \node[reduction, right=of VC, yshift=10mm] (Ham) {$\mathsf{(Un)directed}\text{-}\textsf{Hamilition}$};
    \node[reduction, right=of VC] (SetCover) {$\SetCover$};
    \node[reduction, right=of SetCover] (kMedian) {$\kMedian$};

    \node[reduction, right=of Ham, yshift=10mm, xshift=2mm] (GapTSP) {$\GAP{\TSP_{1, c}}$};
    \node[reduction, right=of Ham, yshift=-10mm, xshift=2mm] (MetricTSP) {$\MetricTSP$};

    \node[reduction, above=of VC] (GapCenter) {$\GAP{\kCenter_{1, 2}}$};

    \node[reduction, below=of MaxCut, yshift=2mm] (MaxSat) {$\MaxkSAT{2}$};

    \node[reduction, below right=of SAT, xshift=6mm] (Clique) {$\Clique$};
    \node[reduction, below=of SAT] (IP) {\footnotesize $\LPZEROONE$\hspace{0.5cm}};
    \node[reduction, below right=of SAT, yshift=-10mm, xshift=35mm] (SetPacking) {$\SetPacking$};

    \node[reduction, below right =of Coloring, xshift=-10mm] (EC) {$\EXA$};
    \node[reduction, right=of EC, yshift=10mm] (Steiner) {$\STEINER$};
    \node[reduction, right=of EC, yshift=-10mm] (Hitting) {$\HIT$};
    \node[reduction, right=of EC] (DMatch) {$\DMatching$};
    \node[reduction, above right=of GI, yshift=18mm, xshift=-27mm] (FVS) {$\FeedbackVertexSet$};
    \node[reduction, right=of IS, yshift=-23mm, xshift=-12mm] (FAS) {$\FeedbackArcSet$};

    \draw[->] (IS) -- (MaxCut) node[midway, sloped, above, edge_label] {\cref{lem:is-to-maxcut}};
    \draw[->] (GI) -- (IS) node[midway, sloped, above, edge_label] {\cref{lem:gi-to-is}};
    \draw[->] (GI) -- (SAT) node[midway, sloped, above, edge_label] {\cref{lem:gi-to-sat}};
    \draw[->] (SAT) -- (Clique) node[midway, sloped, above, edge_label] {\cref{lem:karp:red}};
    \draw[->] (SAT) -- (IP) node[midway, sloped, above, edge_label] {\cref{lem:karp:red}};
    
    \draw[->] (IS) -- (VC) node[midway, sloped, above, edge_label] {\cref{lem:iso:red:VC}};
    \draw[->] (VC) -- ([yshift=-3mm]Ham.north west) node[midway, sloped, above, edge_label] {\cref{lem:karp:red}};
    \draw[->] (VC) -- (SetCover) node[midway, sloped, above, edge_label] {\cref{lem:karp:red}};
    \draw[->] (VC) -- ([yshift=-3mm]FVS.north east) node[midway, sloped, above, edge_label] {\cref{lem:karp:red}};
    \draw[->] (VC) -- (FAS) node[midway, sloped, above, edge_label] {\cref{lem:karp:red}};

    \draw[->] (SetCover) -- (kMedian) node[midway, sloped, above, edge_label] {\cref{lem:Set:red:Mean}};

    \draw[->] (Clique) -- ([yshift=-3mm]SetPacking.north west) node[midway, sloped, above, edge_label] {\cref{lem:karp:red}};
    
    \draw[->] (SAT) -- (Coloring) node[midway, sloped, above, edge_label] {\cref{lem:karp:red}};
    \draw[->] (Coloring) -- (3Coloring) node[midway, sloped, above, edge_label] {\cref{lem:col:isored:3col}};
    \draw[->] (Coloring) -- ([yshift=-3mm]CliqueCov.north west) node[midway, sloped, above, edge_label] {\cref{lem:karp:red}};
    \draw[->] (3Coloring) -- (3SAT) node[midway, sloped, above, edge_label] {\cref{lem:3col:isored:3sat}};
    
    \draw[->] (Coloring) -- (EC) node[midway, sloped, above, edge_label] {\cref{lem:karp:red}};
    \draw[->] (EC) -- ([yshift=-3mm]Steiner.north west) node[midway, sloped, above, edge_label] {\cref{lem:karp:red}};
    \draw[->] (EC) -- ([yshift=-3mm]Hitting.north west) node[midway, sloped, above, edge_label] {\cref{lem:karp:red}};
    \draw[->] (EC) -- (DMatch) node[midway, sloped, above, edge_label] {\cref{lem:karp:red}};

    \draw[->] (MaxCut) -- (MaxSat) node[midway, sloped, above, edge_label] {\cref{lem:maxcut:isored:max2sat}};

    \draw[->] (VC) -- (GapCenter) node[midway, sloped, above, edge_label] {\cref{lem:vc:isored:GapCenter}};

    \draw[->] (Ham) -- ([yshift=-3mm]GapTSP.north west) node[midway, sloped, above, edge_label] {\cref{lem:undirham:isored:gaptsp}};

    \draw[->] (Ham) -- ([yshift=-3mm]MetricTSP.north west) node[midway, sloped, above, edge_label] {\cref{lem:isored:MetricTSP}};

\end{tikzpicture}
    \caption{A tree that shows the different isomorphism-preserving reductions. An edge from $A$ to $B$ means $A \isored B$.} 
    \label{fig:reductions}
\end{figure}

\begin{theorem}\label{thm:iso-pi-gi-complete}
For every $\Pi \in \ProbCollection$, the promise problem $\mathrm{Iso}\text{-}\Pi$
is $\GI$-complete.
\end{theorem}
\begin{proof}
Fix $\Pi \in \ProbCollection$. By Lemma~\ref{lem:isoPi-in-GI}, $\mathrm{Iso}\text{-}\Pi \in \GI$. Next, we use the fact that isomorphism-preserving reductions are transitive together with \cref{lem:gi-to-is}, \cref{lem:gi-to-sat}, \cref{lem:is-to-maxcut}, \cref{lem:col:isored:3col}, \cref{lem:3col:isored:3sat}, \cref{lem:karp:red}, \cref{lem:iso:red:VC}, \cref{lem:maxcut:isored:max2sat}, \cref{lem:isored:MetricTSP}, \cref{lem:Set:red:Mean}, \cref{lem:vc:isored:GapCenter} and \cref{lem:undirham:isored:gaptsp} . Thus, we obtain $\GI \isored \Pi$ for each $\Pi \in \ProbCollection$ (see \cref{fig:reductions}). By \cref{lem:iso-gi-complete}, $\mathrm{Iso}\text{-}\GI$ is $\GI$-hard. Hence, the theorem follows from \cref{lem:iso-hardness-transfer}.
\end{proof}

The proof of \cref{thm:iso-pi-gi-complete} relies on having all necessary isomorphism-preserving reductions. We moved to proofs of most of our reductions to \cref{app:hard} 
since these reductions are standard reductions form the literature. However, we exemplify how our isomorphism-preserving reductions work by present three different ones in this section.

\begin{lemma}\label{lem:gi-to-is}
$\GI \isored \IS$.
\end{lemma}

\begin{proof}
We give an isomorphism-preserving polynomial-time reduction from $\GI$ to the decision problem
$\IS=\{(K,k)\mid K\text{ has an independent set of size at least }k\}$.

Let $(G,H)$ be an instance of $\GI$. If $|V(G)|\neq |V(H)|$, output a fixed NO-instance of $\IS$
(e.g.\ $(K,2)$ where $K$ is a single edge). Hence assume $|V(G)|=|V(H)|=n$.

\paragraph{Reduction $R$.}
We construct the conflict graph $K=K(G,H)$ with
\[
V(K):=V(G)\times V(H),
\]
and for distinct $(u,v),(u',v')\in V(K)$ put an edge between them iff at least one of the following holds:
\begin{enumerate}
    \item (\emph{Function constraint}) $u=u'$ or $v=v'$.
    \item (\emph{Adjacency inconsistency})
    $\mathbf{1}[(u,u')\in E(G)] \oplus \mathbf{1}[(v,v')\in E(H)] = 1$.
\end{enumerate}
Define
\[
R(G,H) := (K(G,H),\, n).
\]
This is computable in polynomial time.

\paragraph{Correctness.}
We claim that
\[
(G,H)\in \GI \iff (K(G,H),n)\in \IS.
\]
If $\varphi:V(G)\to V(H)$ is an isomorphism, then
$S:=\{(u,\varphi(u)) : u\in V(G)\}$ is an independent set of size $n$ in $K(G,H)$:
type-(1) edges are avoided because $\varphi$ is a bijection, and type-(2) edges are avoided because
$\varphi$ preserves adjacency. Hence $(K(G,H),n)\in \IS$.

Conversely, if $K(G,H)$ has an independent set $S$ of size $n$, then constraint (1) implies that $S$
defines a bijection $\varphi:V(G)\to V(H)$ by $\varphi(u)=v$ iff $(u,v)\in S$.
Independence of $S$ implies no type-(2) edge occurs between $(u,\varphi(u))$ and $(u',\varphi(u'))$,
hence $(u,u')\in E(G)$ iff $(\varphi(u),\varphi(u'))\in E(H)$ for all $u\neq u'$.
Therefore $\varphi$ is an isomorphism and $(G,H)\in\GI$.

\paragraph{Isomorphism preservation.}
View $\GI$ instances as attributed graphs $\langle G,H\rangle$ (disjoint union with L/R vertex colors).
If $(G,H)\cong (G',H')$, there exist isomorphisms $\psi_G:G\to G'$ and $\psi_H:H\to H'$.
Define $\Psi:V(K(G,H))\to V(K(G',H'))$ by
\[
\Psi(u,v):=(\psi_G(u),\psi_H(v)).
\]
Then $\Psi$ preserves edges of type (1) (it preserves equality of coordinates) and of type (2)
(since $\psi_G,\psi_H$ preserve adjacency), hence $K(G,H)\cong K(G',H')$.
Also $n=|V(G)|=|V(G')|$, so the parameter $k=n$ is preserved.
Thus $R(G,H)\cong R(G',H')$ as $\IS$ instances.

Therefore $R$ is an isomorphism-preserving reduction and $\GI \isored \IS$.
\end{proof}

\begin{lemma}\label{lem:gi-to-sat}
$\GI \isored \SAT$.
\end{lemma}

\begin{proof}
We define a polynomial-time, isomorphism-preserving many-one reduction from $\GI$ to $\SAT$.

\paragraph{Reduction.}
Let $(G,H)$ be an instance of $\GI$. If $|V(G)|\neq |V(H)|$, output a fixed unsatisfiable CNF
(e.g.\ $(x)\wedge(\neg x)$). Hence assume $|V(G)|=|V(H)|=n$.

Introduce Boolean variables
\[
x_{u,v} \quad\text{for each } u\in V(G),\ v\in V(H),
\]
with the intended meaning “$\varphi(u)=v$”.

We construct a CNF formula $\Phi(G,H)$ consisting of the following clauses.

\paragraph{Bijection constraints.}
\begin{enumerate}
\item \emph{Each $u$ maps somewhere:}
\[
\bigvee_{v\in V(H)} x_{u,v} \qquad\text{for all } u\in V(G).
\]
\item \emph{Each $u$ maps to at most one $v$:}
\[
(\neg x_{u,v}\ \vee\ \neg x_{u,v'}) \qquad\text{for all } u\in V(G),\ v\neq v'\in V(H).
\]
\item \emph{Injectivity (each $v$ is used by at most one $u$):}
\[
(\neg x_{u,v}\ \vee\ \neg x_{u',v}) \qquad\text{for all } v\in V(H),\ u\neq u'\in V(G).
\]
\end{enumerate}

\paragraph{Adjacency consistency constraints.}
For every pair of distinct vertices $u,u'\in V(G)$ and distinct vertices $v,v'\in V(H)$ such that
\[
\mathbf{1}[(u,u')\in E(G)] \oplus \mathbf{1}[(v,v')\in E(H)] = 1,
\]
add the clause
\[
(\neg x_{u,v}\ \vee\ \neg x_{u',v'}).
\]

This construction is polynomial-time and yields a CNF formula of size polynomial in $n$.

\paragraph{Correctness.}
We show that $\Phi(G,H)$ is satisfiable if and only if $G\cong H$.

\smallskip
\noindent\emph{($\Rightarrow$)} Suppose $\varphi:V(G)\to V(H)$ is an isomorphism. Set
$x_{u,\varphi(u)}=\text{true}$ for all $u\in V(G)$ and all other variables to false.
Then the bijection constraints are satisfied. For adjacency consistency, for any $u\neq u'$ let
$v=\varphi(u)$ and $v'=\varphi(u')$. Since $\varphi$ is an isomorphism,
\[
(u,u')\in E(G) \iff (v,v')\in E(H),
\]
so the XOR condition is false and no adjacency-consistency clause is violated. Hence the assignment
satisfies $\Phi(G,H)$.

\smallskip
\noindent\emph{($\Leftarrow$)} Conversely, let $A$ satisfy $\Phi(G,H)$. By constraints (1) and (2), for each
$u\in V(G)$ there is exactly one $v\in V(H)$ with $A(x_{u,v})=\text{true}$; define $\varphi(u)=v$.
Constraint (3) implies that $\varphi$ is injective, and since $|V(G)|=|V(H)|$ it is a bijection.
Now fix distinct $u,u'\in V(G)$ and let $v=\varphi(u)$, $v'=\varphi(u')$.
If exactly one of $(u,u')\in E(G)$ and $(v,v')\in E(H)$ holds, then the corresponding clause
$(\neg x_{u,v}\vee \neg x_{u',v'})$ would be falsified (both literals false), contradicting satisfaction.
Therefore
\[
(u,u')\in E(G) \iff (\varphi(u),\varphi(u'))\in E(H),
\]
so $\varphi$ is an isomorphism $G\cong H$.

\paragraph{Isomorphism preservation.}
We take ``isomorphism'' between $\SAT$ instances to mean: there is a variable renaming (and clause permutation)
mapping one CNF to the other.
If $(G,H)\cong (G',H')$ as $\GI$ instances, there exist isomorphisms $\psi_G:V(G)\to V(G')$ and
$\psi_H:V(H)\to V(H')$. Define the variable renaming $\rho$ by
\[
\rho(x_{u,v}) := x_{\psi_G(u),\,\psi_H(v)}.
\]
Then $\rho$ maps the three families of bijection constraints for $(G,H)$ to the corresponding families for
$(G',H')$, and it maps each adjacency-consistency clause to the corresponding clause in $\Phi(G',H')$.
Hence $\Phi(G,H)\cong \Phi(G',H')$ as $\SAT$ instances, so the reduction is isomorphism-preserving.

Therefore $\GI \isored \SAT$.
\end{proof}

\begin{lemma}\label{lem:is-to-maxcut}
$\IS \isored \MaxCut$.
\end{lemma}

\begin{proof}
We use the standard gadget reduction from $\isName$ to $\MaxCut$ from Barak's online textbook~\cite{Barak-IntroTCS-v095}
and verify that it is also isomorphism-preserving.

\paragraph{Decision versions.}
We view $\IS$ as the decision problem on instances $(G,k)$ asking whether $G$ has an independent set of size
at least $k$, and $\MaxCut$ as the decision problem on instances $(H,T)$ asking whether $H$ has a cut of size
at least $T$.

\paragraph{Reduction.}
Given an instance $(G,k)$ of $\IS$, let $G=(V,E)$ and $m:=|E|$.
We construct a graph $H=R(G)$ as follows (cf.~\cite{Barak-IntroTCS-v095}):
add a distinguished vertex $s^\star$ adjacent to every $v\in V$; and for each edge $e=\{u,v\}\in E$,
attach a constant-size gadget by adding two new vertices $e_u,e_v$ and the five edges
\[
\{e_u,e_v\},\quad \{s^\star,e_u\},\quad \{s^\star,e_v\},\quad \{u,e_u\},\quad \{v,e_v\}.
\]
(No other edges are added; in particular, we do not keep the original edges of $G$ in $H$.)
We set the $\MaxCut$ threshold to
\[
T := k + 4m,
\]
and output the instance $(H,T)$. Barak proves that $G$ has an independent set of size at least $k$ if and only if
$H$ has a cut of size at least $T$, so this is a valid polynomial-time many-one reduction~\cite{Barak-IntroTCS-v095}.

\paragraph{Isomorphism preservation.}
Let $\varphi:G\to G'$ be a graph isomorphism and note that it preserves the parameter $k$ and the number of edges $m$.
We extend $\varphi$ to an isomorphism $\psi:R(G)\to R(G')$ by mapping $s^\star\mapsto s^\star$,
each original vertex $u\mapsto \varphi(u)$, and for each edge $e=\{u,v\}\in E(G)$ mapping the gadget vertices
$e_u\mapsto (\varphi(e))_{\varphi(u)}$ and $e_v\mapsto (\varphi(e))_{\varphi(v)}$, where
$\varphi(e)=\{\varphi(u),\varphi(v)\}\in E(G')$.
Since the gadget adjacencies are defined purely in terms of incidences to $s^\star$, $u$, $v$, and the internal
gadget edge $\{e_u,e_v\}$, this mapping preserves all edges. Hence $R(G)\cong R(G')$ and the threshold $T$ is preserved.

Therefore the reduction is isomorphism-preserving, i.e.\ $\IS \isored \MaxCut$.
\end{proof}

\subsection{The Hardness of $\kSAT{3}$}\label{sec:3sat}

To establish the hardness of $\kSAT{3}$, we need a reduction other than the one introduced by Karp in \cite{DBLP:conf/coco/Karp72}, as Karp’s reduction is not isomorphism-preserving. Instead, we prove hardness through the following sequence of reductions: $\COL \isored \kCOL{3} \isored \kSAT{3}$

\subsubsection{Lovász's Reduction for Colorability}

The following construction goes back to László Lovász (see \cite{lovasz_graph}). We give a full description of Lovász's construction here since the paper is hard to find.\footnote{We refer the interested reader to the following website for further information: \url{https://users.encs.concordia.ca/~chvatal/notes/color.html}}  See \cref{fig:lov} for an example of the construction. 

\begin{definition}[See \cite{lovasz_graph}]
    Given a graph $G$ with vertex set $V(G)$
    and edge set $E(G)$. Further, let $k \geq 3$ be an integer. We construct a new graph $\lovG{G}{k}$ with:
    \begin{itemize}
        \item $V(\lovG{G}{k}) \coloneqq V \uplus C \uplus W \uplus \{G_g, G_b\}$ with 
        \begin{itemize}
            \item $V \coloneq \{v^{(i)} \, | \, v \in V(G), i \in [k]\}$
            \item $C \coloneq \{c_v^{(i)} \, | \, v \in V(G), i \in [k]\}$
            \item $W \coloneq \{e_u^{(i)}, e_v^{(i)}, e_g^{(i)} \, | \, e = \{u, v\} \in E(G), i \in [k] \}$
        \end{itemize}
        \item $E(\lovG{G}{k}) \coloneqq E_b \uplus  E_g  \uplus E_C \uplus E_W $ with 
        \begin{itemize}
            \item $E_b \coloneq \{G_b, G_g\} \; \cup \; \{\{v^{(i)}, G_b\} \, | \, v \in V(G), i \in [k]\}$
            \item $E_g \coloneq  \{\{e_g^{(i)}, G_g\}  \, | \, e \in E(G), i \in [k]\}$
            \item $E_C \coloneq \{\{c_v^{(i)}, c_v^{(i+1 \mod k)} \} \, | \, \, v \in V(G), i \in [k]\} \; \cup \;  \{ \{v^{(i)}, c_v^{(i)} \} \, | \, v \in V(G), i \in [k]   \}$
            \item \[E_W = \bigcup_{\substack{e \in E(G) \\ e = \{u, v\}}} \bigcup_{i = 1}^k \left\{\{e_u^{(i)}, e_v^{(i)}\}, \{e_v^{(i)}, e_g^{(i)}\}, \{e_g^{(i)}, e_u^{(i)}\}, \{e_u^{(i)}, u^{(i)}\}, \{e_v^{(i)}, v^{(i)}\}\right\} \]
        \end{itemize}
    \end{itemize}
\end{definition}

\begin{figure}
    \centering
    
    \begin{subfigure}[b]{0.8\textwidth}
        \centering

        \includegraphics[]{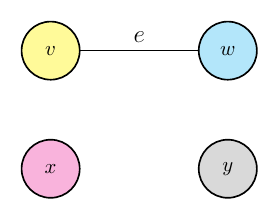} 
        \caption{A undirected graph $G$ with $4$ vertices. Note that the vertices are actually uncolored.}
        \label{fig:lov:a}
    \end{subfigure}
    
    \vspace{1em}

    \begin{subfigure}[b]{1.1\textwidth}
        \centering
        \includegraphics[width=\textwidth]{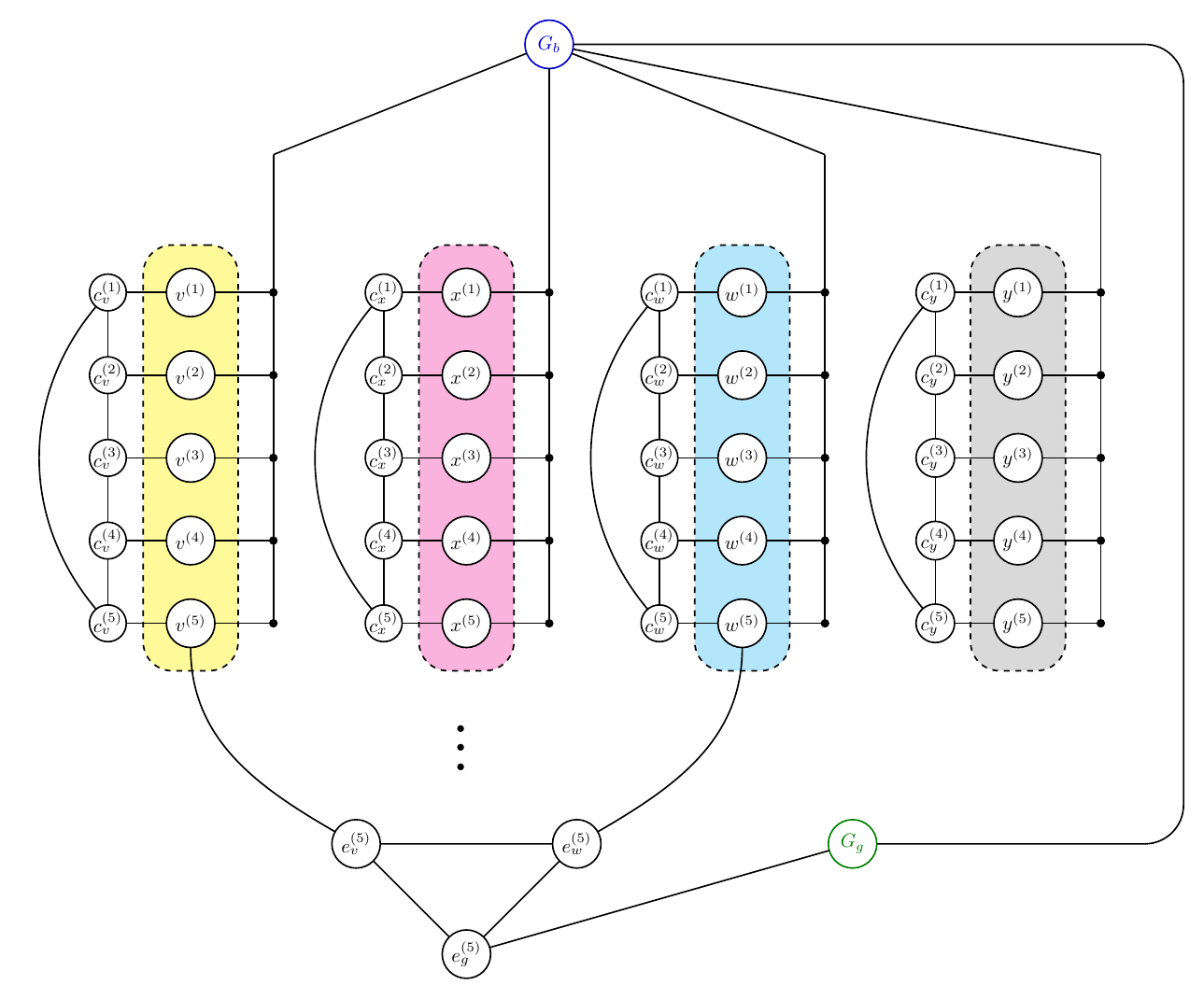}
        \caption{The graph $\lovG{G}{5}$. Note that we only show one of the five edge gadgets for the edge $e \in E(G)$.}
        \label{fig:lov:b}
    \end{subfigure}
    \caption{Example of $G$ and $\lovG{G}{5}$}\label{fig:lov}
\end{figure}

We show that $\lovG{\cdot}{\cdot}$ describes a correct polynomial-time many-one reduction.
\begin{lemma}\label{lem:lovG}
    Given a graph $G$ and an odd integer $k \geq 3$, then $\lovG{G}{k}$ satisfies the following properties:
    \begin{itemize}
        \item $\lovG{G}{k}$ has $O(k \cdot (|V(G)| + |E(G)|))$ many vertices and edges.
        \item $\lovG{G}{k}$ can be constructed in time $O(k \cdot (|V(G)| + |E(G)|))$.
        \item $\lovG{G}{k}$ is $3$-colorable if and only if $G$ is $k$-colorable.
    \end{itemize}
\end{lemma}
\begin{proof}
    Write $n \coloneq |V(G)|$ and $m \coloneqq |E(G)|$.
    It is easy to check that $\lovG{G}{k}$ has $O(k \cdot (n + m))$ many vertices and edges and that we can construct $\lovG{G}{k}$ in time $O(k \cdot (n + m))$. In the following, we assume that $\lovG{G}{k}$ is $3$-colorable with a function $\tilde{c} \colon V(\lovG{G}{k}) \to \{r, g, b\}$ (i.e., the colors red, green, blue). We show that $G$ is $k$-colorable.

    Without loss of generality, we assume that the vertex $G_g \in \lovG{G}{k}$ is colored $g$ and $G_b \in \lovG{G}{k}$ is colored $b$ (we can do this since $G_g$ and $G_b$ are connected). Next, observe that $\tilde{c}(v^{(i)}) \in \{r, g\}$ since each $v^{(i)}$ is connected to $G_b$. Next, we show the following claim.
    \begin{claim}\label{claim:v:g}
        For every $v\in V(G)$ there exists an index $i \in [k]$ such that $\tilde{c}(v^{(i)}) = g$.
    \end{claim}
    \begin{proof}
        Assume otherwise, then $\tilde{c}(v^{(i)}) = r$ for all $i \in [k]$.
        Since $k$ is odd, we obtain that the vertices $\{c_v^{(i)} \, | \, i \in [k]\}$ induce an odd cycle of size at least $3$. Now, each odd cycle of size $3$ is 3-colorable but not two colorable. Thus there is an index $i \in [k]$ with $\tilde{c}(c^{(i)}) = r$. However, this is not possible since $v^{(i)}$ is adjacent to $c^{(i)}$. Thus, our assumption that $\tilde{c}(v^{(i)}) = r$,  for all $i \in [k]$, is wrong and therefore there exist an $i \in [k]$ with $\tilde{c}(v^{(i)}) = g$.
    \end{proof}
    Now, \cref{claim:v:g} yields a function $f \colon V(G) \to [k]$ that maps each vertex $v$ to the smallest value $i \in [k]$ with $\tilde{c}(v^{(i)}) = g$. We show that $f$ is a $k$-coloring for $G$. Assume that $f$ is not a proper $3$-coloring then there is an edge
    $e = \{u, v\} \in E(G)$ with $f(u) = f(v) = i$. By construction, we obtain $\tilde{c}(v^{(i)}) = g = \tilde{c}(u^{(i)})$. Further, note that $\tilde{c}(G_g) = g$. Thus, each vertex in the triangle $\{e^{(i)}_u, e^{(i)}_v, e_g^{(i)}\}$ is adjacent to a vertex with color $g$. Hence, we can only use the colors $\{r, b\}$ to color $\{e^{(i)}_u, e^{(i)}_v, e_g^{(i)}\}$ which is not possible since the vertices form a triangle in $\lovG{G}{k}$. Therefore, $\tilde{c}$ is not a proper 3-coloring which is a contradiction. 

    \paragraph{}
    Next, let $f \colon V(G) \to [k]$ be a $k$-coloring for $G$. We define the the following mapping on the subset $V' \coloneqq \{G_g, G_b \} \cup \{v^{(i)} \, | \, v\in V(G), i \in  [k]\}$ of vertices of $\lovG{G}{k}$.
    \[c' \colon V'  \,\to  \, \{r, g, b\},  \, x \mapsto \begin{cases}
       g & \text{ if $x = G_g$ } \\
       b & \text{ if $x = G_b$ } \\
       g & \text{ if $x = v^{(i)}$ and $f(v) = i$} \\
       r & \text{ if $x = v^{(i)}$ and $f(v) \neq i$} \\
    \end{cases} \]
    We show that $c'$ can be extended to a proper $3$-coloring that is defined on all vertices of $\lovG{G}{k}$. First note that under $c'$ no two adjacent vertices in $V'$ have the same color (the only relevant case is between $G_g$ and $v^{(i)}$ with $f(v) = i$ which are non-adjacent by construction). 
    
    Let $v \in V(G)$ then $\{c_v^{(i)} \, | \, \in [k]\}$ induces an odd cycle in $\lovG{G}{k}$. By construction there are vertices $v^{(s)}$ that are colored $r$ and a vertex $v^{(t)}$ that is colored $g$. Thus, we have three colors available to color the odd cycle $\{c_v^{(i)} \, | \, \in [k]\}$. 

    Now, let $e = \{u, v\} \in E(G)$ and $i \in [k]$ then $\{e^{(i)}_u, e^{(i)}_v, e^{(i)}_g\}$ induces a triangle in $\lovG{G}{k}$. Now, by construction not every node in $\{e^{(i)}_u, e^{(i)}_v, e^{(i)}_g\}$ can be adjacent to a vertex that is colored $g$. Otherwise, $f(v) = i = f(u)$ which is not possible. Thus, we have three colors available to color the triangle $\{e^{(i)}_u, e^{(i)}_v, e^{(i)}_g\}$. Meaning that we can extend $c'$.

\end{proof}

The reduction $\lovG{\cdot}{\cdot}$ is isomorphism-preserving.
\begin{lemma}\label{lem:lovG:iso}
    Let $G$ and $H$ be two graphs and $k \geq 3$ be an odd positive integer. If $G \cong H$ then $\lovG{G}{k} \cong \lovG{H}{k}$.
\end{lemma}
\begin{proof}
    Let $\varphi \colon V(G) \to V(H)$ be an isomorphism from $G$ to $H$. Write $V(\lovG{G}{k}) \coloneqq V \uplus C \uplus W \uplus \{G_g, G_b\}$. For $v \in V(G)$, we write $\tilde{v} \coloneqq \varphi(v)$. For $e \in \{u, v\}$, we write $\tilde{e} \coloneqq \{\varphi(u), \varphi(v)\} \in E(H)$. Note that $v\to\tilde{v}$ and $e\to\tilde{e}$ both define one-to-one mappings between $V(G)$ and $V(H)$, and between $E(G)$ and $E(H)$, respectively. We define 
    \[\psi \colon V(\lovG{G}{k}) \to V(\lovG{H}{k}) , x \mapsto \begin{cases}
       H_g & \text{ if $x = G_g$ } \\
       H_b & \text{ if $x = G_b$ } \\
       \tilde{v}^{(i)} & \text{ if $x = v^{(i)}$} \\
       c_{\tilde{v}}^{(i)} & \text{ if $x = c_v^{(i)}$} \\
       \tilde{e}_{\tilde{v}}^{(i)} & \text{ if $x = e_v^{(i)}$} \\
       \tilde{e}_{b}^{(i)} & \text{ if $x = e_b^{(i)}$} \\
    \end{cases}. \]
    We show that $\psi$ defines an isomorphism. First, note that $\psi$ is a bijection. Let us show that $\{x, y\} \in E(\lovG{G}{k})$ if and only if $\{\psi(x), \psi(y)\} \in E(\lovG{H}{k})$. By construction we already know that $|E(\lovG{G}{k})| = |E(\lovG{H}{k})|$, thus it is actually enough to show that $\{x, y\} \in E(\lovG{G}{k})$ implies $\{\psi(x), \psi(y)\} \in E(\lovG{H}{k})$.\footnote{This implies that the mapping $\psi_E \colon E(\lovG{G}{k}) \to E(\lovG{H}{k}), \{x, y\} \mapsto \{\psi(x), \psi(y)\}$ is well-defined. Further, $\psi_E$ is injective since $\psi$ is. Lastly,  $|E(\lovG{G}{k})| = |E(\lovG{H}{k})|$ yields that $\psi_E$ is a bijection and thus $\psi$ is an isomorphism.}
    
    We consider a case distinction.
    \begin{itemize}
        \item Case $x = G_b$ or $y = G_b$: Without loss of generality, we assume $x = G_b$. If $\{G_b, y\} \in \lovG{G}{k}$ then either $y = G_g$ or $y = v^{(i)}$ thus $\{H_g, \psi(y)\}  \in E(\lovG{H}{k})$.

        \item Case ($x = G_g$ or $y = G_g$) and $x \neq G_b$ and $y \neq G_b$: Without loss of generality, we assume  $x = G_g$. If $\{G_g, y\} \in \lovG{G}{k}$ then $y = e_g^{(i)}$ thus $\{\psi(G_g), \psi(y)\}  = \{H_g, \tilde{e}_g^{(i)}\} \in E(\lovG{H}{k})$.

        \item Case $x, y \in V \uplus C \uplus W$ and ($x \in V$, $y \in C$ or $x \in C$, $y \in Y$): Without loss of generality, we assume $x = v^{(i)}$ and $y = c_u^{(j)}$ for some $u, v \in V$ and $i, j \in [k]$. If $\{v^{(i)}, c_u^{(j)}\} \in \lovG{G}{k}$ then $v = u$ and $i = j$. Thus, $\{\psi(x), \psi(y)\} = \{\tilde{v}^{(i)}, c_{\tilde{v}}^{(i)}\} \in \lovG{H}{k}$.

        \item Case $x, y \in V \uplus C \uplus W$ and ($x \in V$, $y \in W$  or $x \in W$, $y \in V$): Without loss of generality, we assume $x = v^{(i)}$ and $y = e_c^{(j)}$ for some $v \in V$, $e = \{a, b\} \in E(G)$ and $i, j \in [k]$. If $\{v^{(i)}, e_c^{(j)}\} \in \lovG{G}{k}$ then $v = c$, $i = j$ and either $a = v$ or $b = v$. Thus, $\{\psi(x), \psi(y)\} = \{\tilde{v}^{(i)}, \tilde{e}_{\tilde{v}}^{(i)}\} \in \lovG{H}{k}$.

        \item Case $x, y \in V \uplus C \uplus W$ and $x \in C$ and $y \in C$:  Without loss of generality, we assume $x = c_v^{(i)}$ and $y = c_u^{(j)}$ for $u, v \in V$ and $i, j \in [k]$. If $\{c_v^{(i)}, c_u^{(j)}\} \in \lovG{G}{k} $ then $u = v$ and either $j = i+1$ or $i = j+1$. Without loss of generality, assume $j = i+1$, now 
        $\{\psi(x), \psi(y)\} = \{c_{\tilde{v}}^{(i)}, c_{\tilde{v}}^{(i+1)}\} \in \lovG{H}{k}$.

        \item Case $x, y \in V \uplus C \uplus W$ and $x \in E$ and $y \in E$:  Without loss of generality, we assume $x = e_{c}^{(i)}$ and $y = {d}_{c'}^{(j)}$ for some $e, d \in E(G)$, $i, j \in [k]$. If 
        $\{e_{c}^{(i)}, {d}_{c'}^{(j)} \} \in \lovG{G}{k}$ then $e = d$ and $i = j$ and $c \neq c'$. Thus, $\{\psi(x), \psi(y)\} = \{\tilde{e}_{c}^{(i)}, \tilde{e}_{c'}^{(i)}\} \in \lovG{H}{k}$.

        \item Otherwise: In all remaining cases we have $x, y \in V \uplus C \uplus W$ with $\{x, y \} \notin E(\lovG{G}{k})$. For example, if $x, y \in V$ then $\psi(x) = \tilde{v}^{(i)}$ and $\psi(x) = \tilde{u}^{(j)}$ for $v, u \in V$ and $i, j \in [k]$. Note that $\{\tilde{v}^{(i)}, \tilde{u}^{(j)}\} \notin E(\lovG{H}{k})$.
    \end{itemize}
    This shows that $\psi$ is an isomorphism which proves that lemma.
\end{proof}

We extend Lovász's construction by handling even values of $k$ and conclude a isomorphism-preserving reduction to $\kCOL{3}$.
\begin{lemma}\label{lem:col:isored:3col}
    $\COL \isored \kCOL{3}$.
\end{lemma}
\begin{proof}
    We first define a helper function $A \colon \Graphs \to \Graphs$ that takes as input an graph $G$ and outputs a graph $G'$ that is obtained by adding a new vertex $x$ to $V(G)$ and connecting $x$ to all vertices in $G$. Observe that $A(G)$ is computable in time $O(|V(G)| + |E(G)|)$. Further, note that $A(G)$ is $(k+1)$-colorable if and only if $G$ is $k$-colorable. Also note that $G \cong H$ if and only if $A(G) \cong A(H)$.  Next, we define 
    \[f \colon \Graphs \times \mathbb{N}_{\geq 3}  \to \Graphs, (G, k) \mapsto \begin{cases}
        \lovG{G}{k} &\text{ if $k$ is odd} \\ 
        \lovG{A(G)}{k+1} &\text{ if $k$ is even}
    \end{cases}. \]
    We show that $f$ defines a polynomial time many-one reduction from $\COL$ to $\kCOL{3}$ is isomorphism-preserving. First, observe that $f$ is polynomial time computable due to \cref{lem:lovG}. Next, let $G$ be a graph and $k$ be a number, we show $(G, k) \in \COL$ if and only if $f(G, k) \in \kCOL{3}$. If $k$ is odd then this directly follows from \cref{lem:lovG}. If $k$ is even then $G$ is $k$-colorable if and only if $A(G)$ is  $(k+1)$-colorable. Now, \cref{lem:lovG} yields that $A(G)$ is $(k+1)$-colorable if and only if $\lovG{A(G)}{k+1}$ is 3-colorable, proving the claim.

    Lastly, we show that $f$ is isomorphism-preserving. To this end, let $k \geq 3$ be an integer and $G, H$ be graphs with $G \cong H$. If $k$ is odd then \cref{lem:lovG:iso} implies $f(G, k) \cong f(H, k)$. If $k$ is even then we obtain $A(G) \cong A(H)$. Again, \cref{lem:lovG:iso} yields $f(G, k) \cong f(H, k)$, proving the lemma.
\end{proof}

\subsubsection{Reduction from 3-Coloring to $\kSAT{3}$}
We give an isomorphism-preserving reduction from $\kCOL{3}$ to $\kSAT{3}$.
\begin{lemma}\label{lem:3col:isored:3sat}
    $\kCOL{3} \isored \kSAT{3}$.
\end{lemma}
\begin{proof}
    We define a polynomial-time, isomorphism-preserving many-one reduction from $\kCOL{3} \isored \kSAT{3}$.
    
    \paragraph{Reduction.}
    let $G$ be a graph, we define a $3$-CNF $f(G)$ in the following way:
    \begin{itemize}
        \item For each $v \in V(G)$ we define variables $v_r, v_g, v_b$. For each edge $e \in E(G)$, we define variables  $d^1_e, d_e^2, d^3_e$. Thus, $f(G)$ has $3(|V(G)|  + |E(G)|) $ many variables.
        \item For each $v \in V(G)$, we define the clause $(v_r \lor v_g \lor v_b)$. 

        For each $e = \{u, v\} \in V(G)$ and $c \in \{r, g, b\}$, we define the clause $(\lnot v_c  \lor \lnot u_c \lor d^1_e)$. Further, we define the clauses 
        \[(\lnot d^1_e \lor d^2_e \lor d_e^3 ), \;\; (\lnot d^1_e \lor d^2_e \lor \lnot d_e^3 ), \;\; (\lnot d^1_e \lor \lnot d^2_e \lor d_e^3 ), \;\; (\lnot d^1_e \lor \lnot d^2_e \lor \lnot d_e^3).\]
        Note that theses clauses are constructed in a way that ensures $d^1_e = 0$. This is done to reduce $(\lnot v_c  \lor \lnot u_c \lor d^1_e)$ to $(\lnot v_c  \lor \lnot u_c)$.
        In total $f(G)$ has $|V(G)| + 7|E(G)|$ many clauses.
    \end{itemize} 

    Note that $f(G)$ is a $3$-CNF that is obviously computable in time $O(|V(G)| + |E(G)|)$.

    \paragraph{Correctness.}
    We show that $G$ is 3-colorable if and only if $f(G)$ is solvable. 
    
    Let $\tilde{c} \colon V(G) \to \{r, g, b\}$ be a valid $3$-coloring of $G$. We use $\tilde{c}$ to construct a satisfying assignment for $f(g)$. We define an assignment $\mathbf{x}$ via
    \[x \mapsto \begin{cases}
        1 & \text{ if $x = v_c$ for $c \in \{r, g, b\}$ and $\tilde{c}(v) = c$} \\ 
        0 &\text{ otherwise} \\ 
    \end{cases}.\]
    Now, let $e = \{u, v\} \in V(G)$ and assume that $\mathbf{x}$ does not satisfy $(\lnot v_c \lor \lnot u_c \lnot d^1_e)$. Since $d^1_e = 0$, this is only possible if $v_c = 1$ and $u_c = 1$. However, this would imply that $\tilde{c}(v) = \tilde{c}(u)$ which is not possible because $\tilde{c}$ is a proper 3-coloring. Thus,  $\mathbf{x}$ satisfies $(\lnot v_c \lor \lnot u_c \lnot d^1_e)$. Further, it is easy to check that $\mathbf{x}$ satisfies all other constraints. Hence, $f(g)$ is satisfiable. 

    Next, let $\mathbf{x}$ be a satisfying assignment for $f(g)$. Note that $\mathbf{x}$ can only satisfy the constraints 
    \[(\lnot d^1_e \lor d^2_e \lor d_e^3 ), \;\; (\lnot d^1_e \lor d^2_e \lor \lnot d_e^3 ), \;\; (\lnot d^1_e \lor \lnot d^2_e \lor d_e^3 ), \;\; (\lnot d^1_e \lor \lnot d^2_e \lor \lnot d_e^3),\]
    by setting $d^1_e = 0$ for each $e \in E(G)$.
    We show how to construct a $3$-coloring for $G$ using $\mathbf{x}$. Due to the constraint $(v_r \lor v_g \lor v_b)$, for each $v \in V(G)$ there is a $c \in \{r, g, b\}$ with $v_c = 1$ with respect to $\mathbf{x}$. We define $\tilde{c}(v) = c$ (if $v_c = 1$ is true for multiple $c \in  \{r, g, b\}$ then we choose arbitrarily). We show that $\tilde{c}$ defines a proper $3$-coloring on $V(G)$. Assume otherwise, then there is a edge $e = \{u, v\} \in E(G)$ with $\tilde{c}(v) = c = \tilde{c}(u)$. However, since $d_e^1 = 0$, this would yield that the constraint $(\lnot v_c \lor \lnot u_c \lor d^1_e)$ is false which is a contradiction. Thus, $\tilde{c}$ defines a proper 3-coloring.
 
    \paragraph{Isomorphism preservation.}
    We show that if $G, H$ are two graphs with $G \cong H$ then $f(G) \cong f(H)$. 
    
    Let $\varphi \colon V(G) \to V(H)$ be an isomorphisms between $G$ and $H$. For $v \in V(G)$, we write $\tilde{v} \coloneqq \varphi(v)$. For $e \in \{u, v\}$, we write $\tilde{e} \coloneqq \{\varphi(u), \varphi(v)\} \in E(H)$. Note that $v\to\tilde{v}$ and $e\to\tilde{e}$ both define one-to-one mappings between $V(G)$ and $V(H)$, and between $E(G)$ and $E(H)$, respectively. We write $V(f(G))$ for the set of variables in $f(G)$.  Let
    \[\psi \colon V(f(G)) \to V(f(H)), x \mapsto \begin{cases}
        \tilde{v}_c &\text{ if $x = v_c$} \\
        d_{\tilde{e}}^i &\text{ if $x = d^i_e$} 
    \end{cases}.\]
    It is now easy to check that $\psi$ is a isomorphism between $E(G)$ and $E(H)$.
    
    Thus, $f$ defines a polynomial time many-one reduction from $\kCOL{3}$ to $\kSAT{3}$ is isomorphism-preserving. Hence, $\kCOL{3} \isored \kSAT{3}$.
\end{proof}

\subsection{Hardness in the Isomorphism-Priors Model}\label{sec:hard:iso:pro}

We show that the $\GI$-hardness of $\mathrm{Iso}\text{-}\Pi$ implies that $\Pi$ cannot be solved in our model with fast preprocessing, unless $\GI\in \cP$.

\begin{lemma}\label{lem:isoPi-hard-model}
Let $\Pi$ be a decision problem such that $\mathrm{Iso}\text{-}\Pi$ is $\GI$-hard. Unless $\GI \in \cP$, $\Pi$ cannot be solved inside the Isomorphic-Priors model with $k$ prior instances each of size at most $n$ with preprocessing and query time polynomial in $n$, even when $k=1$ prior instance is given.
\end{lemma}

\begin{proof}
Assume that there exists an algorithm for $\Pi$ in the Isomorphic-Priors model with
preprocessing time $\poly_k(n)$ and query time $\poly_k(n)$, and that it is correct on all valid queries.
In particular, consider the case $k=1$.

We show that this implies $\mathrm{Iso}\text{-}\Pi \in \cP$.
Given an input $(G,H)$ to $\mathrm{Iso}\text{-}\Pi$, run the model's preprocessing phase on the single
prior instance $G$ together with $L=\Pi(G)$ (i.e., $L = 1$ if any only if $G$ is a YES-instances), obtaining a data structure $D$.
Then query the model on $H$ and obtain an output
\[
z := \textsf{Query}(D,H)\in\{0,1,\bot\}.
\]
We decide $\mathrm{Iso}\text{-}\Pi$ by outputting \textsc{YES} iff $z=1$, and \textsc{NO} otherwise
(in particular, treat $z=\bot$ as \textsc{NO}).

Correctness holds under the promise defining $\mathrm{Iso}\text{-}\Pi$:
\begin{itemize}
\item If $(G,H)\in \mathrm{Iso}\text{-}\Pi_{\textsc{YES}}$, then $G\in\Pi$ and $H\cong G$.
Since $\Pi$ is invariant under isomorphism, $\Pi(H)=\Pi(G)=1$, so the model must output $z=\Pi(H)=1$,
and we accept.
\item If $(G,H)\in \mathrm{Iso}\text{-}\Pi_{\textsc{NO}}$, then $G\in\Pi$ and $H\notin\Pi$, i.e.\ $\Pi(H)=0$.
By correctness of the model algorithm on non-matching queries, on input $H$ it outputs either $z=\Pi(H)=0$
or the special symbol $\bot$. In both cases we reject.
\end{itemize}
The running time is polynomial, since preprocessing on one instance and one query are polynomial in $n$.

Thus $\mathrm{Iso}\text{-}\Pi \in \cP$. Since by assumption $\mathrm{Iso}\text{-}\Pi$ is $\GI$-hard,
this implies $\GI \in \cP$.
Therefore, unless $\GI\in\cP$, no such algorithm exists for $\Pi$ in the Isomorphic-Priors model (even for $k=1$).
\end{proof}

We show that a problem, that can be reduced to $\GI$ with an isomorphism-preserving reduction, cannot be solved in our model with a small data structure, unless $\GI \in \coma$.
\begin{lemma}\label{lem:Pi-ds-coMA-via-GI}
Let $\Pi$ be a decision problem such that $\GI \isored \Pi$.
Assume $\GI \notin \coma$. Then $\Pi$ cannot be solved inside the Isomorphic-Priors model with $k$ prior instances each of size at most $n$ with a data structure of size polynomial in $n$
and query time polynomial in $n$, even when $k=1$ prior instance is given.
\end{lemma}

\begin{proof}
Assume for contradiction that $\Pi$ can be solved in the Isomorphic-Priors model with preprocessing and query
time polynomial and with a data structure of size $\poly_k(n)$.

Since $\GI \isored \Pi$, there exists an isomorphism-preserving many-one reduction $R$ from $\GI$-instances
to $\Pi$-instances such that for every $\GI$-instance $I$,
\[
I\in \GI \iff R(I)\in \Pi.
\]

We build an algorithm that solves $\GI$ in the Isomorphic-Priors model with a polynomial-size data structure (already for $k=1$),
contradicting Lemma~\ref{lem:gi-ds-coMA} unless $\GI \in \coma$.

\paragraph{GI preprocessing (one prior instance).}
Given a single prior $\GI$-instance $I=(G_1,G_2)$ together with $L=\GI(I)$,\footnote{That is, $L = 1$ if $G_1 \cong G_2$ and zero, otherwise.}
compute the corresponding $\Pi$-instance $X := R(I)$.
Since $R$ is a many-one reduction, we have $\Pi(X) = \GI(I) = L$,\footnote{Here, $\Pi(x) = 1$ if $x \in \Pi$ and zero otherwise.} so we can use the same value $L$ for $X$.
Run $\Pi$-preprocessing on the singleton database $\{(X,L)\}$ to obtain a data structure $D$.

\paragraph{GI query.}
Given a query $\GI$-instance $J$, compute the $\Pi$-instance $Y:=R(J)$ and query the $\Pi$-algorithm to get
\[
z := \textsf{Query}_{\Pi}(D,Y)\in\{0,1,\bot\}.
\]
Output $z$.

\paragraph{Correctness.}
If $J\cong I$, then by isomorphism preservation $Y\cong X$, hence the query is matching and the
$\Pi$-algorithm must output $z=\Pi(Y)=\GI(J)$, so we output the correct answer.
If $J\not\cong I$, then either the $\Pi$-algorithm outputs $\Pi(Y)$ or $\bot$; in either case, the
above rule outputs the correct value for $\GI(J)$ or $\bot$, both are correct outputs for $J$ in the Isomorphic-Priors model.

Thus we obtain a GI solver in the Isomorphic-Priors model with a polynomial-size data structure, contradicting
Lemma~\ref{lem:gi-ds-coMA} unless $\GI\in\mathsf{coMA}$.
\end{proof}

\begin{proof}[Proof of \Cref{cor:hardness-consequence}]
Fix $\Pi \in \ProbCollection$ and assume for contradiction that $\Pi$ can be solved
in the Isomorphic-Priors model with preprocessing time polynomial in $k$ and $n$ and query time polynomial in $n$,
even for $k=1$.

By \Cref{thm:iso-pi-gi-complete}, the promise problem $\mathrm{Iso}\text{-}\Pi$ is $\GI$-hard.
Applying Lemma \Cref{lem:isoPi-hard-model}, we conclude that
$\GI\in\cP$.

Taking the contrapositive, unless $\GI\in\cP$ no such Isomorphic-Priors model algorithm for $\Pi$ exists
(even for $k=1$). Since $\Pi$ was arbitrary in the set $\ProbCollection$, the claim
follows for all three problems.
\end{proof}

\begin{proof}[Proof of \Cref{thm:ds-lb}]
Fix $\Pi\in \ProbCollection$ and assume for contradiction that $\Pi$ can be solved
in the Isomorphic-Priors model with a data structure of size $\poly(k,n)$.

By Lemma~\ref{lem:Pi-ds-coMA-via-GI}, it suffices to show that $\GI \isored \Pi$.
For $\Pi=\IS$ this holds by \Cref{lem:gi-to-is}, and for $\Pi=\mathsf{SAT}$ it holds by
Lemma~\ref{lem:gi-to-sat}.
For $\Pi = \MaxCut$, we have $\GI \isored \IS$ by \Cref{lem:gi-to-is} and
$\IS \isored \MaxCut$ by \Cref{lem:is-to-maxcut}; by composition of
isomorphism-preserving reductions, $\GI \isored \MaxCut$.

Hence in all three cases $\GI \isored \Pi$, and \Cref{lem:Pi-ds-coMA-via-GI} implies $\GI\in\mathsf{coMA}$,
contradicting the assumption. Therefore, unless $\GI\in\mathsf{coMA}$, no such polynomial-size data structure algorithm
for $\Pi$ exists (even for $k=1$).
\end{proof}

\subsection{Proof of Hardness for $\mathsf{Gap\text{-}Max\text{-}}k\mathsf{\text{-}Coverage}$}\label{subsec:proofhardnessofapprox}
In this section, we show that $\mathsf{Max\text{-}}k\mathsf{\text{-}Coverage}$ cannot be approximated better than $(1-1/e)$ in our model, unless $\MaxkSAT{3}$ can be approximated arbitrarily well in our model.
We achieve this by first observing that Feige's $k$-prover system works for general $\MaxkSAT{3}$ and is not limited to bounded-occurrences.
This is necessary because bounded occurrences would imply that the underlying isomorphism problem can be solved fast.
We then go on and show that this implies a isomorphism-preserving reduction.
\begin{lemma}[Feige's $k$-prover system for general $\MaxkSAT{3}$]\label{lem:feige-kprover-unbounded}
Consider the $k$-prover proof system of Section~2.3 in~\cite{Feige1998Threshold} applied to an arbitrary
3CNF formula $\phi$ (no bounded-occurrence assumption). Let $\eps\in(0,1)$.
\begin{enumerate}
\item If $\phi$ is satisfiable, then there is a prover strategy that causes the verifier to always
\emph{strongly} accept.
\item If $\opt(\phi)\le 1-\eps$ (i.e., no assignment satisfies more than a $(1-\eps)$-fraction of clauses),
then the verifier \emph{weakly} accepts with probability at most $k^2\cdot 2^{-c\ell}$, where $c>0$ depends only on
$\eps$ (and on the constant answer length), and $\ell$ is the repetition parameter of the verifier.
\end{enumerate}
\end{lemma}

\begin{proof}
The verifier samples $\ell$ clauses uniformly and independently and chooses one distinguished variable per clause,
exactly as in~\cite[Sec.~2.3]{Feige1998Threshold}. The weak predicate is “some pair of provers is consistent”
(on the induced assignment to the $\ell$ distinguished-variable positions), and the strong predicate is
“every pair is consistent”~\cite[Sec.~2.3]{Feige1998Threshold}.

\paragraph{Completeness.}
If $\phi$ is satisfiable, fix any satisfying assignment and let every prover answer consistently with it.
Then the induced assignments to the distinguished-variable positions are identical for all provers, so the verifier
always strongly accepts (identical to Feige’s argument)~\cite[Lem.~2.3.1]{Feige1998Threshold}.

\paragraph{Soundness: reducing to two provers and parallel repetition.}
Assume $\opt(\phi)\le 1-\eps$ and suppose the verifier weakly accepts with probability at least $\delta$.
Then there exist two provers whose answers are consistent with probability at least $\delta/k^2$ (by averaging over the
$\binom{k}{2}$ pairs), exactly as in Feige~\cite[Lem.~2.3.1]{Feige1998Threshold}.
By the Hamming-distance property of the code used to distribute queries, there are at least $\ell/6$ coordinates in which
one of these provers receives a clause and the other receives the distinguished variable from that clause
(again exactly as in Feige)~\cite[Lem.~2.3.1]{Feige1998Threshold}.
Fix the questions on the remaining $5\ell/6$ coordinates to maximize acceptance; the acceptance probability remains at least
$\delta/k^2$. Now ignore those fixed coordinates. This yields a strategy for a two-prover system that succeeds with probability
at least $\delta/k^2$ on $\ell/6$ independent coordinates, i.e., on $\ell/6$ parallel repetitions of the underlying two-prover test.

\paragraph{Base two-prover error does not need bounded occurrence.}
Feige’s Proposition~2.2.1 (stated for 3CNF-5) bounds the optimal acceptance probability of the base two-prover test by
$1-\eps/3$ when at most a $(1-\eps)$-fraction of clauses are simultaneously satisfiable~\cite[Prop.~2.2.1]{Feige1998Threshold}.
Its proof uses only: (i) the second prover’s strategy induces a global assignment $x$, (ii) a random clause is unsatisfied by $x$
with probability at least $\eps$, and (iii) conditioned on such a clause, if the first prover makes it satisfied by changing at least
one literal, then a uniformly random distinguished variable from that clause disagrees with probability at least $1/3$.
This argument does not depend on how many times a variable appears in the formula (variables may occur unboundedly often).

Hence the base two-prover system has constant error at most $1-\eps/3$, independent of $n$ and of occurrence bounds.

\paragraph{Parallel repetition.}
By Raz’s parallel repetition theorem (Feige’s Theorem~2.2.2), repeating a one-round two-prover system $t$ times independently in
parallel reduces the acceptance probability to at most $2^{-c t}$ for some constant $c>0$ that depends only on the base error and
the constant answer length~\cite[Thm.~2.2.2]{Feige1998Threshold}.
Taking $t=\ell/6$, we obtain
\[
\delta/k^2 \;\le\; 2^{-c\ell},
\]
after renaming constants, and thus $\delta \le k^2\cdot 2^{-c\ell}$.

This proves the claimed soundness bound for arbitrary 3CNF formulas, without any bounded-occurrence assumption.
\end{proof}

This $k$-prover systems can be used for an isomorphism-preserving reduction.
\begin{lemma}\label{lem:gap3sat-to-maxkcover}
For every $\eps,\eps'\in(0,1/3)$, there is an isomorphism-preserving reduction
\[
\mathsf{Gap\text{-}Max\text{-}3\text{-}SAT}_{1,\,1-\eps}\ \isored\ \mathsf{Gap\text{-}Max\text{-}}k\mathsf{\text{-}Coverage}_{1,\,1-1/e+\eps'}.
\]
\end{lemma}
\begin{proof}
Fix $\eps\in(0,1/3)$. We construct an isomorphism-preserving mapping from 3CNF formulas to
$\mathsf{Max\text{-}}k\mathsf{\text{-}Coverage}$ instances, following Feige~\cite{Feige1998Threshold}, and use
Lemma~\ref{lem:feige-kprover-unbounded} (Feige’s Lemma~2.3.1 without bounded-occurrence) as the only
nontrivial black box about the underlying $k$-prover system.

\paragraph{No regularization.}
We do \emph{not} apply any bounded-occurrence / regularization preprocessing: such steps typically require
non-canonical choices (e.g., ordering occurrences), and are therefore not isomorphism-preserving in general.

\paragraph{Step 1: The reduction $\Phi \mapsto (U(\Phi),\mathcal{S}(\Phi),k')$.}
Let $\Phi$ be a 3CNF instance. Consider the $k$-prover verifier from~\cite[Sec.~2.3]{Feige1998Threshold}
with repetition parameter $\ell$ (to be chosen later as a function of $\eps$).
Let $\mathcal{R}$ be the set of random strings of the verifier; write $R:=|\mathcal{R}|$.
For each $r\in\mathcal{R}$, the verifier determines (for each prover $i\in[k]$) a question
$q_i(r)$ (a tuple of clause/variable identifiers as in~\cite[Sec.~2.3]{Feige1998Threshold}).
Let $Q$ be the number of possible questions a prover can receive (as in Feige).

\medskip
\noindent\emph{Partition systems.}
Fix parameters $(m,L,k,d)$ as in Feige’s partition-system lemma (Lemma~3.2 in~\cite{Feige1998Threshold}),
with $L=2^\ell$. For each $r\in\mathcal{R}$, create an independent copy of the partition system
\[
B_r := B(m,L,k,d)
\]
on a ground set of size $m$, equipped with $L$ partitions indexed by strings $p\in\{0,1\}^\ell$.
Each such partition consists of $k$ disjoint parts, which we denote
\[
\mathcal{P}_{r,p}=\bigl\{B(r,p,1),\dots,B(r,p,k)\bigr\},
\qquad
\bigcup_{i=1}^k B(r,p,i)=B_r.
\]

\medskip
\noindent\emph{Universe.}
Define the universe as the disjoint union of all these ground sets:
\[
U(\Phi)\;:=\;\biguplus_{r\in\mathcal{R}}\bigl(\{r\}\times B_r\bigr),
\]
so $|U(\Phi)|=mR$.

\medskip
\noindent\emph{Sets.}
For each prover $i\in[k]$, each question $q$ that prover $i$ may receive, and each possible answer $a$ to that
question, define a set $S(q,a,i)\subseteq U(\Phi)$ by
\[
S(q,a,i)
\;:=\;
\bigcup_{\substack{r\in\mathcal{R}:\\ q_i(r)=q}}
\bigl(\{r\}\times B(r,p_r(a),i)\bigr),
\]
where $p_r(a)\in\{0,1\}^\ell$ is the induced assignment to the $\ell$ distinguished variables selected under
randomness $r$ (exactly as in Feige’s reduction).
Let
\[
\mathcal{S}(\Phi)\;:=\;\{\,S(q,a,i): i\in[k],\, q\in\mathcal{Q}_i,\, a\in\mathcal{A}_{i,q}\,\}.
\]

\medskip
\noindent\emph{Budget.}
Set
\[
k'\;:=\;k^Q,
\]
corresponding to choosing, for each prover, one answer for each of its $Q$ possible questions
(as in Feige’s Max-$k$-Cover hardness, Section~5 of~\cite{Feige1998Threshold}).

We output the $\mathsf{Max\text{-}}k\mathsf{\text{-}Coverage}$ instance
\[
R(\Phi)\;:=\;\bigl(U(\Phi),\mathcal{S}(\Phi),k'\bigr).
\]

\paragraph{Step 2: Gap correctness (black-box).}
We argue informally, following Feige.

\smallskip
\noindent\emph{YES case.}
If $\opt(\Phi)=1$, then by Lemma~\ref{lem:feige-kprover-unbounded} there is a prover strategy that
causes the verifier to \emph{strongly} accept on every random string. Fix that strategy and choose the $k'$
sets that correspond to the provers’ answers on all possible questions (one answer per question per prover,
hence $k^Q$ choices total). For each $r$, strong acceptance implies that the induced strings $p_r(\cdot)$ are
consistent across provers, and therefore the chosen sets cover, inside the block $\{r\}\times B_r$, all $m$
points (because the $k$ parts of $\mathcal{P}_{r,p}$ partition $B_r$). Hence the chosen $k'$ sets cover all of
$U(\Phi)$, so $\opt(R(\Phi))=1$.

\smallskip
\noindent\emph{NO case.}
If $\opt(\Phi)\le 1-\eps$, then by Lemma~\ref{lem:feige-kprover-unbounded} the verifier weakly accepts
with probability at most $\delta$, where $\delta$ can be made exponentially small in $\ell$.
Feige’s analysis of the reduction (using the partition-system lemma) shows that any family of $k'$ sets
corresponds to a (possibly inconsistent) prover strategy, and that the fraction of universe elements covered is
at most $1-1/e+O(\delta)+O(1/d)$ (cf.~\cite[Sec.~4--5]{Feige1998Threshold}). Choosing $\ell$ and $d$ large enough
as functions of $\eps$ makes the additive terms $O(\delta)+O(1/d)$ arbitrarily small; in particular, we can
choose parameters so that the covered fraction is at most $1-1/e+\eps'$.
Therefore $\opt(R(\Phi))\le 1-1/e+\eps'$ in the NO case.

Altogether, $\Phi\in \mathsf{Gap\text{-}Max\text{-}3\text{-}SAT}_{1,\,1-\eps'}$ implies
$R(\Phi)\in \mathsf{Gap\text{-}Max\text{-}}k\mathsf{\text{-}Coverage}_{1,\,1-1/e+\eps'}$.

\paragraph{Step 3: Isomorphism preservation.}
Let $\alpha:\Phi\to\Phi'$ be an isomorphism of 3CNF instances (renaming variables and permuting clauses,
preserving literal signs). We build an isomorphism between $R(\Phi)=(U,\mathcal{S},k')$ and
$R(\Phi')=(U',\mathcal{S}',k')$.

The verifier’s random string $r\in\mathcal{R}$ specifies a pattern of queries (which clause indices are chosen,
which variables are distinguished, etc.). Under $\alpha$, clause/variable identifiers are renamed, and hence
each question $q$ and answer $a$ is renamed canonically to $\alpha(q)$ and $\alpha(a)$.
Define a bijection $\beta_U:U\to U'$ by
\[
\beta_U(r,x):=(r,\;x)\qquad\text{for }x\in B_r,
\]
i.e.\ we keep the partition-system point and randomness index, but interpret the clause/variable identifiers
inside the indexing data via $\alpha$ (equivalently: we identify $\mathcal{R}$ across $\Phi$ and $\Phi'$ by the
same randomness and query pattern, differing only by renamed identifiers).
Define a bijection $\beta_{\mathcal{S}}:\mathcal{S}\to\mathcal{S}'$ by
\[
\beta_{\mathcal{S}}\bigl(S(q,a,i)\bigr)\;:=\;S(\alpha(q),\alpha(a),i).
\]
By construction, for every universe element $(r,x)$ and every set $S(q,a,i)$, we have
\[
(r,x)\in S(q,a,i)\quad\Longleftrightarrow\quad
(r,x)\in S(\alpha(q),\alpha(a),i),
\]
because membership is defined solely by (i) whether $q_i(r)=q$ and (ii) which part $B(r,p_r(a),i)$ is selected,
and both are invariant under renaming clause/variable identifiers. Hence $(\beta_U,\beta_{\mathcal{S}})$ preserves
incidence, so it is an isomorphism of set systems. Therefore $R(\Phi)\cong R(\Phi')$.

\paragraph{Conclusion.}
The mapping $\Phi\mapsto R(\Phi)$ is polynomial-time, isomorphism-preserving, and maps YES-instances of
$\mathsf{Gap\text{-}Max\text{-}3\text{-}SAT}_{1,\,1-\eps}$ to YES-instances of
$\mathsf{Gap\text{-}Max\text{-}}k\mathsf{\text{-}Coverage}_{1,\,1-1/e+\eps'}$ and NO-instances to NO-instances of
$\mathsf{Gap\text{-}Max\text{-}}k\mathsf{\text{-}Coverage}_{1,\,1-1/e+\eps'}$. This proves the lemma.
\qedhere
\end{proof}

We use this isomorphism-preserving reduction to conclude the proofs of \Cref{thm:fastlb-approx,thm:smalllb-approx}.
\begin{proof}[Proof of \Cref{thm:fastlb-approx,thm:smalllb-approx}]
    We show the contraposition.
    Assume there is an $\eps>0$ such that $\mathsf{Gap\text{-}Max\text{-}}k\mathsf{\text{-}Coverage}_{1,\,1-1/e+\eps}$ can be solved in the Isomorphic-Priors model with preprocessing and query time polynomial in $n$ even when $k=1$ prior instance is given.

    Then for every $\eps'>0$, $\mathsf{Gap\text{-}Max\text{-}3\text{-}SAT}_{1,\,1-\eps'}$ can be solved in the Isomorphic-Priors model with preprocessing and query time polynomial in $n$ even when $k=1$ prior instance is given.

    This algorithm just computes for all formulas $\Phi$ that are given in the preprocessing and query phase the $\mathsf{Gap\text{-}}k\mathsf{\text{-}Coverage}$ instance $R(\Phi)$, using the reduction from \Cref{lem:gap3sat-to-maxkcover}, and applies the preprocessing and query algorithm for $\mathsf{Gap\text{-}Max\text{-}}k\mathsf{\text{-}Coverage}_{1,\,1-1/e+\eps}$ to them.

    The same argument can be made for data structures of polynomial size.

    Because in \cite{DBLP:journals/jal/GuhaK99} it is argued that the reduction in \Cref{lem:Set:red:Mean} is also a reduction from $\mathsf{Gap\text{-}Max\text{-}}k\mathsf{\text{-}Coverage}_{1,\,1-1/e+\eps}$ and $\mathrm{Gap\text{-}}\kMedian_{1,1+2/e+2\eps}$. 
    We can repeat both arguments above for $\kMedian$.
\end{proof}

\section{Conclusion}\label{sec:hypothesis}\label{sec:conclusion}

We conclude by presenting the proofs of our main theorems, referring to the relevant results established throughout the paper, and by discussing some hypotheses within our model and their consequences.

\subsection{Proof of the Main Theorems}

\begin{proof}[Proof of \cref{thm:main:NP-hard Upper Bound}]
    By \cref{cor:csp_with_preprocessing}, we can solve $\csp$ inside the Isomorphic-Priors model with $\widetilde{O}(km)$ preprocessing time and $\widetilde{O}(m)$ query time. By \cref{cor:constrained_spanning_tree_preprocessing}, we can solve $\cst$ randomized inside the Isomorphic-Priors model with $\widetilde{O}(kn^\omega)$ preprocessing time and $\widetilde{O}(n^\omega)$ query time. 

    The extensions from \emph{$\cspName$} to \emph{$\cspName$ and $\lspName{p}$}, and from \emph{$\cstName$} to \emph{$\cstName$ and $\lstName{p}$}, are both straightforward generalizations of the presented algorithm. To this end, note that the technique of walk lifting (see \cref{lem:walk_lifting}) also works in this setting.
\end{proof}

\begin{proof}[Proof of \cref{thm:main:Triangle Upper Bound}]
    By \cref{thm:replacement_path_with_preprocessing} there is an algorithm that solves $\repPathName$ inside the Isomorphic-Priors model using $\widetilde{O}(km)$ preprocessing time and $\widetilde{O}(m)$ query time.
    By \cref{cor:neg_triangle_preprocessing}, there is an algorithm that solves $\negName$ inside the Isomorphic-Priors model randomized using $\widetilde{O}(k n^\omega)$ preprocessing time and $\widetilde{O}(n^\omega)$ query time. 
\end{proof}

\begin{proof}[Proof of \cref{thm:main:Infinite Games Upper Bounds}]
    By \cref{cor:energy_preprocessing}, we can solve $\energyName$ inside the Isomorphic-Priors model with $\widetilde{O}(km)$ preprocessing time and $\widetilde{O}(m)$ query time. Our Weisfeiler-Leman techniques also yield the same result for $\parityName$, $\meanName$, and $\stochasticName$. To be more specific, we show in \cref{thm:value_determined_game_with_preprocessing} that this can be extended to value-determined games (a general concept including the above games).
\end{proof}

\begin{proof}[Proof of \cref{thm:intro:hardness}]
    The proof is a direct consequence of \cref{cor:hardness-consequence} and \cref{thm:ds-lb}. The approximation results follow from the corresponding gap results. 
\end{proof}

\subsection{ $\kClique{k}$ Hypothesis and Consequence}

One of the foundational hypotheses in fine-grained complexity \cite{Williams2018ICM} states that, for every fixed integer $k\geq 3$, no randomized algorithm can detect a $k$-clique in an $n$-vertex graph in $O(n^{\frac{\omega k}{3}- \epsilon})$ time for any constant $\epsilon > 0$. It is natural to ask whether this fundamental computational barrier translates to the Isomorphic-Priors model. We observe that if this hypothesis extends to our setting,  then $\giName$ is not polynomial-time solvable. 

Recall our hypothesis from earlier:

\kclique*

\begin{observation}
Under \Cref{conj:iso-k-clique}, $\giName$  cannot be solved in polynomial time (i.e. $\GI\not\in \cP$).
\end{observation}
\begin{proof}[Proof sketch]
Assume for contradiction that $\giName$ can be solved in $O(n^b)$ time for some constant $b \ge 1$. Let $k = \lceil 10b \rceil$. We construct an isomorphic-priors algorithm for  $\kClique{k}$ to handle $O(1)$ prior graphs that requires $O(1)$ preprocessing time and $O(n^b)$ query time. 

During preprocessing, the algorithm does nothing. Upon receiving a query graph $H$, it simply runs the $O(n^b)$-time $\GI$ algorithm to test whether $H \cong G_i$ for any $i$. If such $G_i$ exists, it outputs the precomputed answer $\opt(G_i)$; otherwise, it correctly declares $H \not\cong G_i$ for any $i$. This satisfies the correctness guarantees of the Isomorphic-Priors model.

Since we chose $k = \lceil 10b \rceil$, our query time is $O(n^b) = O(n^{k/10})$. Because the matrix multiplication exponent $\omega \ge 2$, we have $\frac{\omega k}{3} \ge \frac{2k}{3}$. Since $k = \lceil 10b \rceil\geq 10$, $n^{k/10}$ is asymptotically faster than $O(n^{\frac{\omega k}{3} - 0.001})$. This implies we can solve  $\kClique{k}$ in our model in $O(n^{\frac{\omega k}{3} - 0.001})$ time, which  contradicts \Cref{conj:iso-k-clique}.
\end{proof}

Note that the implication $\GI \notin \cP$ in the observation above remains true under weaker variants of the hypothesis, specifically as long as  $\kClique{k}$ requires $\Omega(n^{f(k)})$ time in our model for any monotonically increasing, unbounded function $f(k)$. 
Consequently, refuting \Cref{conj:iso-k-clique} would serve as an intermediate step toward a polynomial-time algorithm for $\giName$. Conversely, resolving the hypothesis, or any of its weaker variants, in the affirmative would prove that no such algorithm exists (i.e., $\GI \notin \mathsf{P}$).

\subsection{Acknowledgments and AI Disclosure}

The authors thank Maximilian Thiessen for helpful comments on an
earlier draft of this paper.

The authors used Gemini and ChatGPT as interactive writing and proofreading assistants, including for figure generation.
While preparing the camera-ready version for FOCS 2026, ChatGPT was also used to suggest references, discuss 
technical formulations and possible extensions, and assist with
the formulation and proof of the observation relating the
Isomorphic-Priors model to invariant graph predictors (\Cref{sec:related}). These
interactions informed revisions to the abstract, introduction,
related-work and open-problems discussions, and the observation
mentioned above. The authors substantially revised the generated
material and take full responsibility for the manuscript,
including all mathematical claims, proofs, and references.

Separately, after the FOCS 2026 submission, the authors used the
ChatGPT models \texttt{5.6-Sol} and \texttt{6-Pro} to explore
alternative proofs of results in this paper and extensions to
other problems. No additional proofs or results from these
post-submission experiments are included in this paper.

\bibliography{References}

\appendix

\section{Further Problems that are Invariant under 1-WL}
\label{app:1wl}

We first provide the deferred proof of Lemma~\ref{lem:color-consistent-strategy} concerning color-consistent strategies for $\energyName$. We then give further examples of problems that are invariant under 1-WL and can therefore be solved in near-linear time in the Isomorphic-Priors model.

\subsection{Color-Consistent Strategies for $\energyName$}
\label{app:color-consistent-strategies}

\begin{proof}[Proof of Lemma~\ref{lem:color-consistent-strategy}]
Let $h(v)$ be the minimum initial energy required to win from $v$. By assumption,~$\sigma$ wins from every starting vertex $v$ with initial energy $h(v)$. The shadow-play argument using Neighborhood Preservation in the proof of Theorem~\ref{thm:energy-game-1-wl} shows that vertices of the same stable 1-WL color have the same minimum initial energy. Thus $h$ is constant on each color class.

For each color class $C$, choose an arbitrary representative $r_C\in C$ and define $H(C):=h(r_C)$.
For each stable color class $C\subseteq V_1$ and each $u\in C$, choose
\[
\widehat{\sigma}(u)\in\left\{v\in V \;\middle|\; (u,v)\in E,\ w(u,v)=w(r_C,\sigma(r_C)),\ \chi_\infty^G(v)=\chi_\infty^G(\sigma(r_C))\right\}.
\]
This set is nonempty by Neighborhood Preservation.

Let $(u,v)\in E$ be an edge consistent with $\widehat{\sigma}$, meaning that either $u\in V_2$ or $v=\widehat{\sigma}(u)$. Denote $C=\chi_\infty^G(u)$ and $D=\chi_\infty^G(v)$. There exists a $\sigma$-consistent edge $(r_C,z)$ with $w(r_C,z)=w(u,v)$ and $\chi_\infty^G(z)=D$: for $u\in V_1$, take $z=\sigma(r_C)$; for $u\in V_2$, use Neighborhood Preservation. Since $\sigma$ wins from $r_C$ with energy $H(C)$, we have $h(r_C) + w(u,v) \ge h(z)$ and therefore $w(u,v)\ge H(D)-H(C)$.

Now fix any initial vertex $v_0$ and energy $E_0\ge h(v_0)$, and consider any resulting infinite walk $\pi=v_0v_1v_2\dots$ consistent with $\widehat{\sigma}$. Applying the preceding inequality to each edge $(v_k,v_{k+1})$ yields a telescoping sum
\[
E_k=E_0+\sum_{j=0}^{k-1}w(v_j,v_{j+1})\ge H(\chi_\infty^G(v_k))\ge0
\qquad\text{for all }k\ge0.
\]
Hence $\widehat{\sigma}$ is optimal from every starting vertex and has the required consistency property.
The construction of $\widehat{\sigma}$ takes linear time.
\end{proof}

\subsection{More Games on Graphs}
\label{app:hardness-graphs}
We have shown in \cref{sec:algorihtms via WL} that $\energyName$ can be solved efficiently in the Isomorphic-Priors model using the 1-WL algorithm. In the following, we extend them to more general games.

\begin{definition}[Value-determined Game]\label{def:value-determined-game}
    An instance of Value-determined Game is defined by a tuple $G = (V, E, V_1, V_2, V_r, \delta, w, v_0, W, p)$, where 
    \begin{itemize}
        \item $(V, E)$ is a finite directed graph with no terminal states, 
        \item $(V_1, V_2, V_r)$ is a partition of $V$, where Player 1 controls $V_1$, Player 2 controls $V_2$, and $V_r$ are random nodes,
        \item $\delta$ is a transition probability function for random nodes: for each $u\in V_r$, $\delta(u)$ is a probability distribution over $Succ(u) = \{v \mid (u,v) \in E \}$.
        \item $w: V\rightarrow \Sigma$ is a label function.
        \item $v_0$ is the designated initial vertex.
        \item $W \subseteq \Sigma^\omega$ is the Borel objective for Player 1.
        \item $p \in [0,1]$ is a probability threshold. 
    \end{itemize}
    A play $\pi = v_0v_1v_2\dots$ is formed by moving a token: at step $k\geq 0$, if $v_k \in V_i, i\in \{1,2\}$, Player $i$ chooses an outgoing edge $(v_k, v_{k+1}) \in E$ and moves the token to $v_{k+1}$; if $v_k \in V_r$, then the arena samples a node $v_{k+1}$ from $\delta(v_k)$, and moves the token to $v_{k+1}$. We slightly abuse the label function by denoting $w(\pi) := w(v_0)w(v_1)w(v_2)\dots$. We denote $\pi(\sigma, \tau)$ as the (randomized) play where Player 1 uses strategy $\sigma$ and Player 2 uses strategy $\tau$. We say 
    \begin{itemize}
        \item Player 1 wins the game from $v_0$ if $\sup_\sigma \inf_\tau \Pr(w(\pi(\sigma, \tau))\in W) \geq p$,
        \item Player 2 wins the game from $v_0$ if $\inf_\tau \sup_\sigma \Pr(w(\pi(\sigma, \tau))\in W) < p$.
    \end{itemize}
    The game is value-determined, if $\sup_\sigma \inf_\tau \Pr(w(\pi(\sigma, \tau))\in W) = \inf_\tau \sup_\sigma \Pr(w(\pi(\sigma, \tau))\in W)$, which means that exactly one of the Players wins. For a value-determined game $G$, we define $\opt(G)=\text{YES}$ if Player~1 wins from $v_0$, and $\opt(G)=\text{NO}$ otherwise.
    
\end{definition}

We view a value-determined game as an attributed directed graph in the sense of Section~\ref{sec:model}. All instance-specific information can be encoded using vertex and edge attributes. In particular, the threshold $p$, which is a global parameter of the instance, can be encoded by including $p$ in the attribute of every vertex. We regard the Borel objective $W$ as a fixed global parameter. Thus, Value-determined Game fits the model of~\ref{sec:model}.

\begin{theorem}[Value-Determined Game in the Isomorphic-Priors Model]
\label{thm:value_determined_game_with_preprocessing}
    There is an algorithm that solves the Value-determined Game problem with $k$ prior instances each having at most $n$ vertices and $m$ edges in $O(k(n+m)\log n)$ preprocessing time and $O((n+m)\log n)$ query time.

\end{theorem}

Similar to solving $\energyName$ in the Isomorphic-Priors setting, we only need to show that if two games are 1-WL equivalent, then the winning player is the same. 

We initialize 1-WL with the vertex types and labels, and use
transition probabilities as labels on edges leaving random nodes.
All other edges receive a distinguished dummy label.
Colors are compared using a common refinement scheme.
Strategies may be history-dependent and randomized.

\begin{theorem}
    \label{thm:value-determined-game-1wl}
    Let $G$ and $G'$ be two instances of Value-determined Games
    with the same Borel objective $W$.
    If $G \WLequiv{1} G'$ and
    $\chi_\infty^G(v_0)=\chi_\infty^{G'}(v_0')$,
    then the winners of $G$ and $G'$ are the same.
\end{theorem}

\begin{proof}
    By symmetry, it suffices to show that the value of $G'$
    is at least the value of $G$.
    Fix an arbitrary strategy $\sigma$ for Player~1 in $G$.
    For simplicity, we use $w$ for the label functions of both games.

    For every pair of equally colored vertices $u,u'$, fix a
    color-preserving bijection
    $b_{u,u'}:\operatorname{Succ}(u)\to\operatorname{Succ}(u')$.
    At random nodes, additionally require
    $
        \delta(u)(v)
        =
        \delta'(u')\bigl(b_{u,u'}(v)\bigr).
    $
    Such bijections exist by Neighborhood Preservation,
    since the refinement also includes transition probabilities
    as edge labels.

    Fix a common infinite random tape, viewed as a sequence
    $(U_k)_{k\geq 0}$ of independent uniform random variables
    on $[0,1]$, independent of the players' randomization.
    At corresponding random nodes $u,u'$, order the successors
    of $u$ arbitrarily and order those of $u'$ according to
    $b_{u,u'}$.
    At step $k$, both games use $U_k$ to sample from the
    corresponding cumulative probability intervals.
    The two transitions then have the correct marginal
    distributions and lead to successors matched by $b_{u,u'}$.
    This random tape is used only for the coupling and is
    not observed by the players.

    We define $\sigma'$ dynamically by maintaining a simulated
    ``shadow play'' in $G$.
    Let the current finite prefix of a play in $G'$ be
    $\pi'=(v_0',v_1',\dots,v_k')$ and the shadow play in $G$
    be $\pi=(v_0,v_1,\dots,v_k)$.
    We show by induction on $k$ that
    $\chi_\infty^G(v_i)=\chi_\infty^{G'}(v_i')$
    for all $i\leq k$.

    \emph{Base case.}
    For $k=0$, the claim holds by assumption.

    \emph{Inductive step.}
    Assume the claim holds for all $i\leq k$.
    Since $v_k'$ and $v_k$ have the same color,
    they have the same vertex type.
    \begin{itemize}
        \item
        If $v_k$ and $v_k'$ are Player~1 nodes, we first move
        the token in $G$: the strategy $\sigma$ selects a
        successor $v_{k+1}$ according to $\sigma(\pi)$.
        Then we move the token in $G'$ to
        $
            v_{k+1}'=b_{v_k,v_k'}(v_{k+1}).
        $
        By the choice of the bijection,
        $\chi_\infty^{G'}(v_{k+1}')=\chi_\infty^G(v_{k+1})$.

        \item
        If $v_k$ and $v_k'$ are Player~2 nodes, we first move
        the token in $G'$: Player~2 selects an arbitrary
        successor $v_{k+1}'$ of $v_k'$.
        We then move the token in $G$ to
        $
            v_{k+1}=b_{v_k,v_k'}^{-1}(v_{k+1}').
        $
        Again,
        $\chi_\infty^G(v_{k+1})=\chi_\infty^{G'}(v_{k+1}')$.

        \item
        If $v_k$ and $v_k'$ are random nodes, the arena samples $v_{k+1}'\sim\delta'(v_k')$. Upon observing $v_{k+1}'$, the strategy $\sigma'$ appends to the shadow play the vertex $v_{k+1}:=b_{v_k,v_k'}^{-1}(v_{k+1}')$. This implements the coupling above: since $\delta'(v_k')(b_{v_k,v_k'}(u))=\delta(v_k)(u)$ for every $u\in\operatorname{Succ}(v_k)$, the induced shadow transition has distribution $\delta(v_k)$. Moreover, color preservation gives $\chi_\infty^G(v_{k+1})=\chi_\infty^{G'}(v_{k+1}')$.
        Thus $\sigma'$ reconstructs the shadow move from the observed move in $G'$, without observing the random tape.
    \end{itemize}

    The fixed local bijections give a bijection between finite
    histories in the two games.
    Both directions of this correspondence depend only on
    the histories already observed.
    Therefore, $\sigma'$ is a legal history-dependent strategy.

    Now fix an arbitrary strategy $\tau'$ for Player~2 in $G'$.
    Using the history correspondence, define a strategy $\tau$
    in $G$ by transporting the choices of $\tau'$ through
    the inverse local bijections.
    This is also a legal history-dependent strategy.

    Under these strategy pairs, corresponding finite histories
    have equal probabilities.
    At player nodes, this follows from the construction of
    the strategies; at random nodes, it follows from the
    probability-preserving coupling.
    Recall that colors determine the labels of the nodes
    along a play.
    Thus the two games induce the same distribution on label
    sequences, and
    \[
        \Pr\bigl(w(\pi'(\sigma',\tau'))\in W\bigr)
        =
        \Pr\bigl(w(\pi(\sigma,\tau))\in W\bigr).
    \]
    Here $\pi(\sigma,\tau)$ and $\pi'(\sigma',\tau')$ denote
    the resulting infinite plays in $G$ and $G'$, respectively.

    Since this holds for every $\tau'$, we have
    \[
        \inf_{\tau'}
        \Pr\bigl(w(\pi'(\sigma',\tau'))\in W\bigr)
        \ge
        \inf_\tau
        \Pr\bigl(w(\pi(\sigma,\tau))\in W\bigr).
    \]
    Since $\sigma$ was arbitrary, it follows that
    \[
        \sup_{\sigma'}\inf_{\tau'}
        \Pr\bigl(w(\pi'(\sigma',\tau'))\in W\bigr)
        \ge
        \sup_\sigma\inf_\tau
        \Pr\bigl(w(\pi(\sigma,\tau))\in W\bigr).
    \]
    By symmetry, the values of $G$ and $G'$ are equal.
    Since both games are value-determined and have the same
    threshold $p$, their winners are the same.
\end{proof}

\begin{proof}[Proof of \Cref{thm:value_determined_game_with_preprocessing}]
    The first part follows directly from \Cref{theo:1WL:iso,thm:value-determined-game-1wl}.
    For the second part, observe that \Cref{thm:value-determined-game-1wl} provides an explicit method for constructing a winning strategy on $H$ given a strategy for a prior graph $G_i$ such that $G_i \WLequiv{1} H$.
\end{proof}

\paragraph{Generalized B\"uchi games.}
A generalized B\"uchi game is equipped with a family of sets $F_1,\ldots,F_r \subseteq V$, and Player~1 wins if every set $F_i$ is visited infinitely often. Generalized B\"uchi games can be viewed as a special case of value-determined games. Indeed, if the sets $F_1,\ldots,F_r$ come with fixed identifiers, we may label every vertex
$v$ by the set
\[
    w(v) := \{i \in [r] \mid v \in F_i\},
\]
and choose the Borel objective that requires every index $i \in [r]$ to occur infinitely often. Hence, Theorem~\ref{thm:value_determined_game_with_preprocessing} immediately applies to this representation.

There is, however, a subtlety in this formulation. Encoding the indices $i \in [r]$ as part of the vertex labels makes the identities of the B\"uchi sets part of the input attributes. Consequently, an isomorphism is required to preserve every set $F_i$ individually. This is stronger than the arguably more natural notion of isomorphism for generalized B\"uchi games, under which the B\"uchi sets themselves may be permuted. Namely, one would only require the existence of a graph isomorphism~$\varphi$ and a permutation $\pi$ of $[r]$ such that
\[
    v \in F_i
    \quad\Longleftrightarrow\quad
    \varphi(v) \in F'_{\pi(i)}
\]
for every $v \in V$ and $i \in [r]$.

A natural attempt to model this stronger symmetry is to introduce one unlabeled auxiliary vertex for each B\"uchi set and encode membership by incidence edges. In such a representation, graph isomorphisms may freely permute the auxiliary vertices, exactly as desired. However, our current 1-WL argument does not immediately extend to this representation: the shadow-play construction preserves local colors and transition behavior, but does not provide a single globally consistent matching 
between the B\"uchi sets of the two games. Such a consistent matching appears necessary in order to conclude that visiting every B\"uchi set infinitely often is preserved.

\paragraph{Rabin games.}
A similar distinction arises for Rabin games. A Rabin game is equipped
with a family of pairs $(E_1,F_1),\ldots,(E_r,F_r)$, where
$E_i,F_i \subseteq V$, and Player~1 wins if there exists an index
$i \in [r]$ such that $E_i$ is visited only finitely often and $F_i$
is visited infinitely often. If the pairs come with fixed identifiers,
we may label every vertex $v$ by
\[
    w(v) :=
    \bigl(\{i \in [r] \mid v \in E_i\},
          \{i \in [r] \mid v \in F_i\}\bigr).
\]
The Rabin condition is then a Borel objective on the resulting
label sequence, so
Theorem~\ref{thm:value_determined_game_with_preprocessing}
also applies to this representation.

As in the generalized B\"uchi case, however, this encoding requires
isomorphisms to preserve every set $E_i$ and $F_i$ individually.
Under the more natural notion of isomorphism, the Rabin pairs may
be permuted, while the pairing and the roles of their two components
are preserved. More precisely, one requires a game-graph
isomorphism~$\varphi$ and a permutation $\pi$ of $[r]$ such that
\[
    \varphi(E_i)=E'_{\pi(i)}
    \qquad\text{and}\qquad
    \varphi(F_i)=F'_{\pi(i)}
\]
for every $i \in [r]$.
Encoding the pairs and their two types of membership using auxiliary
vertices and incidence edges allows this symmetry to be represented.
However, the same obstacle remains: our current shadow-play
construction does not provide a single globally consistent
correspondence between the Rabin pairs, and therefore does not
establish preservation of the Rabin objective in this representation.

\paragraph{Open questions.}
Under the natural notions of isomorphism described above, we leave
open whether the Isomorphic-Priors model admits algorithms with
$O(\poly(k,n))$ preprocessing time and
\begin{enumerate}
    \item $O(n^{2.99})$ query time for generalized B\"uchi games,
    when the B\"uchi sets may be arbitrarily permuted; and
    \item $O(\poly(n))$ query time for Rabin games,
    when the Rabin pairs may be arbitrarily permuted.
\end{enumerate}

\subsection{The $\repPathName$ Problem}

The work of Vassilevska Williams and Williams~\cite{DBLP:conf/focs/WilliamsW10} is widely regarded as an early landmark in fine-grained complexity. In particular, they showed that a number of fundamental problems with best known subcubic algorithms are tightly related via subcubic reductions, yielding an equivalence web among path, matrix, and triangle problems.

In this appendix, we show that, in our model, $\repPathName$ can be solved in $O((n+m)\log n)$ time.

\begin{definition}[$\repPathName$]
    An instance of $\repPathName$ is defined by a tuple $G = (V, E,c,P)$, where 
    
    \begin{itemize}
        \item $(V, E)$ is a finite directed graph, 
        \item $c\colon E \to \mathbb{R}$ is a cost function assigning a non-negative cost to each edge,
        \item $P = v_1 v_2 \ldots v_k$ is a shortest path.
    \end{itemize}
    The solution for $\repPathName$ is a vector $\opt = (\opt_1,\ldots, \opt_{k-1})$, where for each $i\in \{1,\ldots,k-1\}$, $\opt_i$ denotes the cost of the shortest $v_1,v_k$-path in $G_i = (V, E\setminus \{(v_i,v_{i+1}) \})$.
\end{definition}

\begin{remark}\label{rem:replacement_path_no_neg_cycles}
We may assume without loss of generality that $G$ contains no negative cycles reachable from $v_1$ and to $v_k$. Indeed, since $P$ is a shortest path, the distance between its endpoints is finite.
\end{remark}

We view an instance $G = (V, E,c,P)$, where $P = v_1 v_2 \ldots v_k$, of $\repPathName$ as an \attributed graph with edge attributes $\att_E(e) = c(e)$ for all $e \in E$ and vertex attributes $\att_V(v_i) = i$ for all $i \in \{1,\dots,k\}$, and $\alpha_V(w) = 0$ for all other vertices $w \in V \setminus \{v_1,\dots,v_k\}$.

\begin{theorem}[$\repPathName$ in the Isomorphic-Priors Model]
    \label{thm:replacement_path_with_preprocessing}
    There is an algorithm that solves $\repPathName$ with $k$ prior instances each having at most $n$ vertices and $m$ edges in $O(k(n+m)\log n)$ preprocessing time and $O((n+m)\log n)$ query time. 
    Additionally, if each prior $G_i$ is equipped with a set of replacement paths, then the algorithm can also compute a set of replacement paths for $H$ in $O((n+m)\log n)$ time.
\end{theorem}

In the following theorem, we prove that $\repPathName$ is invariant under 1-WL equivalence.

\begin{theorem}\label{thm:replacement_path_problem}
Let $G = (V,E,c\colon E\to \mathbb{R}, P)$ and $G' = (V',E',c'\colon E'\to \mathbb{R}, P')$ be two instances of $\repPathName$.
If $G \WLequiv{1} G'$, then $G$ and $G'$ have the same optimal replacement path vector.
\end{theorem}

\begin{proof}
    Let $P=v_1,\ldots,v_k$ and $P'=v'_1,\ldots,v'_k$ be the designated shortest paths in $G$ and $G'$, respectively.
    Observe that the initial coloring $\chi_0$ assigns color $i$ to $v_i$ and to $v'_i$ for every $i\in\{1,\ldots,k\}$, and assigns color $0$ to all remaining vertices. Let $\opt$ and $\opt'$ be the optimal replacement vectors for $G$ and $G'$, respectively.  We show that for every $i\in\{1,\ldots,k-1\}$, $\opt_i = \opt'_i$.

    Fix $i \in \{1,\ldots,k-1\}$. Let $Q=u_1,\ldots,u_m$ be a shortest $v_1,v_k$-path in $G_i=(V,E\setminus\{(v_i,v_{i+1})\})$. We construct a $v'_1,v'_k$-walk $Q'=u'_1,\ldots,u'_m$ in $G'_i=(V',E'\setminus\{(v'_i,v'_{i+1})\})$ with the same sequence of stable colors as $Q$, and hence with the same total cost.

    We define $Q'$ inductively. Since $u_1=v_1$, its initial color is $1$. Because $G \WLequiv{1} G'$, there exists a vertex $u'_1 \in V'$ with $\chi_\infty^G(u_1)=\chi_\infty^{G'}(u'_1)$. By the Refinement property of 1-WL (see Lemma~\ref{thm:1wl_properties}), $u'_1$ has the same initial color as $u_1$, and therefore $\chi^{G'}_0(u'_1)=1$, so $u'_1=v'_1$.

    Now suppose that $u'_j$ has already been chosen so that $\chi_\infty^G(u_j)=\chi_\infty^{G'}(u'_j)$. Since $(u_j,u_{j+1}) \in E$, by Neighborhood Preservation (see Lemma~\ref{thm:1wl_properties}), there exists an outgoing edge $(u'_j,u'_{j+1}) \in E'$ such that $\chi_\infty^G(u_{j+1})=\chi_\infty^{G'}(u'_{j+1})$ and $(u_j,u_{j+1})$ and $(u'_j,u'_{j+1})$ have the same edge label, i.e., the same cost under $c$ and $c'$, respectively. Continuing in this way, we obtain a walk $Q'$ in $G'$ from $v'_1$ to some vertex $u'_m$ with $\chi_\infty^G(v_k)=\chi_\infty^{G'}(u'_m)$. By Refinement, $\chi^{G'}_0(u'_m)=\chi^G_0(v_k)=k$, and hence $u'_m=v'_k$. 

    Thus, $Q'$ is a $v'_1,v'_k$-walk in $G'$ and $c'(Q')=c(Q)$. We claim that $Q'$ does not use the edge $(v'_i,v'_{i+1})$. Suppose otherwise. Then there exists $j$ such that $(u'_j,u'_{j+1})=(v'_i,v'_{i+1})$. Since $v'_i$ and $v'_{i+1}$ are the unique vertices of initial colors $\chi^{G'}_0(v'_i) = i$ and $\chi^{G'}_0(v'_{i+1}) = i+1$, it follows from $\chi_\infty^G(u_j)=\chi_\infty^{G'}(u'_j)$ and $\chi_\infty^G(u_{j+1})=\chi_\infty^{G'}(u'_{j+1})$, together with another application of Refinement, that $u_j=v_i$ and $u_{j+1}=v_{i+1}$. Hence $(u_j,u_{j+1})=(v_i,v_{i+1})$, contradicting the fact that $Q$ is a path in $G_i$.

    Therefore, $Q'$ is a valid $v'_1,v'_k$-walk in $G'_i$, and so, by Remark~\ref{rem:replacement_path_no_neg_cycles}, $\opt'_i \le c'(Q') = c(Q) = \opt_i$. By symmetry, $\opt_i \le \opt'_i$. Hence $\opt_i=\opt'_i$, as required.
\end{proof}

\begin{proof}[Proof of \Cref{thm:replacement_path_with_preprocessing}]
    The first part follows directly from \Cref{theo:1WL:iso,thm:replacement_path_problem}.
    For the second part, observe that \Cref{thm:replacement_path_problem} provides an explicit method for constructing replacement paths on $H$ given replacement paths for a prior graph $G_i$ such that $G_i \WLequiv{1} H$.
\end{proof}

\subsection{Maximum Flow}

Next, we show that $\maxFlowName$ can be solved in near-linear time in our model, slightly improving over existing worst-case algorithms \cite{DBLP:conf/focs/ChenKLPGS22,DBLP:conf/focs/Brand0PKLGSS23,DBLP:conf/stoc/BrandLLSS0W21,DBLP:conf/focs/BrandLNPSS0W20}.
Towards this end, we again argue that 1-WL equivalent graphs admit the same maximum flow.
We note that this also improves over a result from \cite{DBLP:journals/siamcomp/AtseriasM13} where the same is shown for 2-WL.

\begin{definition}
    An instance of $\maxFlowName$ is a tuple $G = (V,E,c,s,t)$ where $(V,E)$ is a directed graph, $c\colon E(G) \to \mathbb{R}_{\geq 0}$ are edge capacities, and $s,t \in V$ are the source and sink vertex.
    A \emph{flow} is a mapping $f\colon E(G) \to \mathbb{R}_{\geq 0}$ such that
    \begin{itemize}
       \item $f(e) \leq c(e)$ for all $e \in E(G)$, and
        \item $\sum_{u \in N_{in}(v)} f(u,v) = \sum_{w \in N_{out}(v)} f(v,w)$ for all $v \in V\setminus \{s,t\}$.
    \end{itemize}
    The objective is to find a flow $f$ that maximizes $\operatorname{val}(f) \coloneqq \sum_{w \in N_{out}(s)} f(s,w)$.
\end{definition}

\begin{theorem}[$\maxFlowName$ in the Isomorphic-Priors Model]
    \label{thm:max_flow_with_preprocessing}
    There is an algorithm that solves $\maxFlowName$ with $k$ prior instances each having at most $n$ vertices and $m$ edges in  $O(k(n+m)\log n)$ preprocessing time and $O((n+m)\log n)$ query time. 
    Additionally, if each prior $G_i$ is equipped with a maximum flow, then the maximum flow for the query graph can be found in $O((n+m)\log n)$ query time.
\end{theorem}

For the proof, we require some additional basic tools.
First, we generalize the neighborhood preservation property from \Cref{thm:1wl_properties}.

\begin{lemma}[Folklore]
\label{thm:1wl_properties_extended}
For every two (\attributed) graphs $G = (V_G,E_G,\alpha_G^V,\alpha_G^E)$ and $H = (V_H,E_H,\alpha_H^V,\alpha_H^E)$ such that $G \WLequiv{1} H$, the following holds:
\begin{enumerate}
    \item \textbf{Neighborhood Multiset Preservation:} For every $v \in V_G$ and $v' \in V_H$, if $\chi_\infty^G(v) = \chi_\infty^{H}(v')$, then for every stable color $C$ and edge \attribute $\ell$, it holds that
    \[\left|\left\{u \in V_G \mid \att_G^E(u,v) = \ell, \chi_\infty^G(u) = C\right\}\right| = \left|\left\{u' \in V_H \mid \att_H^E(u',v') = \ell, \chi_\infty^H(u') = C\right\}\right|.\]
    The analogous statement holds for outgoing edges.
\end{enumerate}
\end{lemma}

Now let $G = (V,E,c,s,t)$ be an instance of $\maxFlowName$, and let $\chi_{\infty}^G$ denote the coloring computed by 1-WL.
We define an auxiliary instance $G/\chi_\infty = (V^*,E^*,c^*,s^*,t^*)$ via
\begin{itemize}
    \item $V^* = \{\chi^G_\infty(v) \mid v \in V\}$,
    \item $E^* = \{(\chi^G_\infty(v),\chi^G_\infty(w)) \mid (v,w) \in E\}$,
    \item $c^*(a,b) = \sum_{v \in \chi^{-1}(a),w \in \chi^{-1}(b), (v,w) \in E} c(v,w)$,
    \item $s^* = \chi^G_\infty(s)$ and $t^* = \chi^G_\infty(t)$.
\end{itemize}

The next two lemmas prove that the value of the maximum flow is identical in $G$ and $G/\chi_\infty$.

\begin{lemma}
    \label{lem:max_flow_factor_1}
    Let $G = (V,E,c,s,t)$ be an instance of $\maxFlowName$ and let $f$ be a flow in $G$.
    Then there exists a flow $f^*$ in $G/\chi_\infty$ such that $\operatorname{val}(f^*) \geq \operatorname{val}(f)$.
\end{lemma}

\begin{proof}
    For $(a,b) \in E(G/\chi_\infty)$, we define
    \[f^*(a,b) \coloneqq \sum_{v \in \chi^{-1}(a),w \in \chi^{-1}(b), (v,w) \in E} f(v,w).\]
    By definition, it immediately follows that $f^*(a,b) \leq c^*(a,b)$.
    Moreover, for $a \in V^* \setminus \{s^*,t^*\}$, it holds that
    \begin{align*}
        \sum_{b \in N_{in}(a)} f^*(b,a) &= \sum_{b \in N_{in}(a)} \sum_{v \in \chi^{-1}(a),u \in \chi^{-1}(b), (u,v) \in E} f(u,v)\\
        &= \sum_{v \in \chi^{-1}(a)} \sum_{u \in N_{in}(v)} f(u,v)\\
        &= \sum_{v \in \chi^{-1}(a)} \sum_{w \in N_{out}(v)} f(v,w)\\
        &= \sum_{b \in N_{out}(a)} \sum_{v \in \chi^{-1}(a),w \in \chi^{-1}(b), (v,w) \in E} f(v,w)\\
        &= \sum_{b \in N_{out}(a)} f^*(a,b).
    \end{align*}
    Hence, $f^*$ is a flow in $G/\chi_\infty$.
    Finally,
    \[\operatorname{val}(f^*) = \sum_{a \in N_{out}(s^*)} f(s^*,a) = \sum_{a \in N_{out}(s^*)} \sum_{w \in \chi^{-1}(a), (s,w) \in E} f(s,w) = \sum_{w \in N_{out}(s)} f(s,w) = \operatorname{val}(f).\]
\end{proof}

\begin{lemma}
    \label{lem:max_flow_factor_2}
    Let $G = (V,E,c,s,t)$ be an instance of $\maxFlowName$, and let $f^*$ be a flow in $G/\chi_\infty$.
    Then there is a flow $f$ in $G$ such that $\operatorname{val}(f) \geq \operatorname{val}(f^*)$.
\end{lemma}

\begin{proof}
    For $(a,b) \in E(G/\chi_\infty)$, we define $E_{a,b} \coloneqq \{(v,w) \in E \mid \chi^G_\infty(v) = a, \chi^G_\infty(w) = b\}$.
    We define
    \[f(v,w) \coloneqq \frac{f^*(a,b)}{c^*(a,b)} \cdot c(v,w),\]
    for all $(v,w) \in E_{a,b}$.
    Clearly, $f(v,w) \leq c(v,w)$ for all $(v,w) \in E_{a,b}$, because $f^*(a,b) \leq c^*(a,b)$ due to $f^*$ being a flow.

    Now let $v \in \chi^{-1}(a)$ where $a$ is an arbitrary color in the image of $\chi^G_\infty$. We claim that
    \begin{equation}
        \sum_{u \in N_{in}(v)} f(u,v) = \frac{1}{|\chi^{-1}(a)|}\sum_{b \in N_{in}(a)} f^*(b,a).
    \end{equation}
    Indeed, we have that
    \[\sum_{u \in N_{in}(v)} f(u,v) = \sum_{u' \in N_{in}(v')} f(u',v')\]
    for all $v,v' \in \chi^{-1}(a)$ by \Cref{thm:1wl_properties_extended}.
    It follows that
    \begin{align*}
        |\chi^{-1}(a)| \cdot \sum_{u \in N_{in}(v)} f(u,v) &= \sum_{v' \in \chi^{-1}(a)}\sum_{u' \in N_{in}(v')} f(u',v')\\
        &= \sum_{v' \in \chi^{-1}(a)} \sum_{b \in N_{in}(a)} \sum_{u' \in N_{in}(v'), \chi(u') = b} f(u',v')\\
        &= \sum_{v' \in \chi^{-1}(a)} \sum_{b \in N_{in}(a)} \sum_{u' \in N_{in}(v'), \chi(u') = b} \frac{f^*(b,a)}{c^*(b,a)} \cdot c(u',v')\\
        &= \sum_{b \in N_{in}(a)} \sum_{v' \in \chi^{-1}(a), u' \in \chi^{-1}(b), (u',v') \in E} \frac{f^*(b,a)}{c^*(b,a)} \cdot c(u',v')\\
        &= \sum_{b \in N_{in}(a)} f^*(b,a).
    \end{align*}
    Similarly, we obtain that
    \begin{equation}
        \sum_{w \in N_{out}(v)} f(v,w) = \frac{1}{|\chi^{-1}(a)|}\sum_{b \in N_{out}(a)} f^*(a,b).
    \end{equation}
    Together, we conclude that
    \[\sum_{u \in N_{in}(v)} f(u,v) = \frac{1}{|\chi^{-1}(a)|}\sum_{b \in N_{in}(a)} f^*(b,a) = \frac{1}{|\chi^{-1}(a)|}\sum_{b \in N_{out}(c)} f^*(a,b) = \sum_{w \in N_{out}(v)} f(v,w)\]
    where the second equality holds since $f^*$ is a flow.
    Also, we obtain that
    \[\operatorname{val}(f) = \sum_{w \in N_{out}(s)} f(s,w) = \sum_{b \in N_{out}(\chi(s))} f^*(\chi(s),b) = \operatorname{val}(f^*)\]
    since the color class containing $s$ has size one.
\end{proof}

\begin{corollary}\label{cor:max_flow_1wl}
    Let $G = (V,E,c\colon E\to \mathbb{R},s,t)$ and $G' = (V',E',c'\colon E'\to \mathbb{R}, s',t')$ be two instances of $\maxFlowName$. If $G \WLequiv{1} G'$, then $G$ and $G'$ have the same maximum flow value.
\end{corollary}

\begin{proof}
    If $G \WLequiv{1} G'$, then we obtain that $G/\chi_\infty = G'/\chi_\infty$.
    Now, the statement follows from \Cref{lem:max_flow_factor_1,lem:max_flow_factor_2}.
\end{proof}

\begin{proof}[Proof of \Cref{thm:max_flow_with_preprocessing}]
    The first part follows directly from \Cref{theo:1WL:iso,cor:max_flow_1wl}.
    For the second part, observe that \Cref{lem:max_flow_factor_1,lem:max_flow_factor_2} provide an explicit method for obtaining a flow on $H$ given the flow for a prior graph $G_i$ such that $G_i \WLequiv{1} H$.
\end{proof}

\section{More Conditional Lower Bounds}\label{app:hard}

\subsection{Additional Problem Definitions}\label{app:hardness:defitions}

\begin{definition}[$\LPZEROONE$]
    In the $\LPZEROONE$ we are given an integer matrix $C$ and an integer $d$ and. We have to decide if there exist a 0/1-vector $x$ with $C x = d$.

    Two problem instances $(C, d)$ and $(C', d')$ are isomorphic if there are permutation matrices $P$ and $Q$ with $P C Q = C'$ and $P d = d'$.
\end{definition}

Note that $\LPZEROONE$ is closed under isomorphism since if $x$ is a solution to $C x = d$ if and only if $x' = Q^{-1}x$ is a solution to $C' x' = d'$. For the remaining problem checking that they are closed under isomorphism is easy to see.

\begin{definition}[$\kCenter$]\label{def:kcenter}
    In the $\kCenter$ problem we are given a simple graph $G=(V,E,w)$ with edge weights $w\colon E\rightarrow \mathbb{N}$, a number $k\in \mathbb{N}$, and a value $c\in \mathbb{N}$. We have to decide if there is a set $X\subseteq V$ of size at most $k$ such that every $v\in V$ has distance at most $c$ to the closest point in $X$.

Two problem instances $(G=(V,E,w),k,c)$ and $(G'=(V',E',w'),k',c')$ are isomorphic if $k=k'$, $c=c'$ and $G$ is isomorphic to $G'$.
\end{definition}

\begin{definition}[$\Clique$]
    In the $\Clique$ problem we are given a simple graph $G$ and a value $k$. We have to decide if $G$ contains $k$-clique (i.e., a subgraph with $k$ vertices such that each vertex is connected to each other vertex)

    Two problem instances $(G, k)$ and $(G', k')$ are isomorphic if $G$ and $G'$ are isomorphic as graphs and $k = k'$.
\end{definition}

\begin{definition}[$\CliqueCover$]
    In the $\CliqueCover$ problem we are given a simple graph $G$ and a value $k$. We have to decide if $G$ contains clique cover of size at most $k$ (i.e., there is a collection of at most $k$ many cliques that are all subgraphs of $G$ and each edge of $G$ is covered by at least one clique)

    Two problem instances $(G, k)$ and $(G', k')$ are isomorphic if $G$ and $G'$ are isomorphic as graphs and $k = k'$. 
\end{definition}

\begin{definition}[$\COL$]
    In the $\COL$ problem we are given a simple graph $G$ and a value $k$. We have to decide if $G$ is $k$-colorable.

    Two problem instances $(G, k)$ and $(G', k')$ are isomorphic if $G$ and $G'$ are isomorphic as graphs and $k = k'$.
\end{definition}

\begin{definition}[$\kCOL{k}$]
    In the $\kCOL{k}$ problem we are given a simple graph $G$. We have to decide if $G$ is $k$-colorable.

    Two problem instances $G$ and $G'$ are isomorphic if $G$ and $G'$ are isomorphic as graphs.
\end{definition}

\begin{definition}[$\DMatching$]
    In the $\DMatching$ problem we are given a set $U \subseteq [n]^3$. We have to decide if there exist a set $W \subseteq U$ with $|W| = n$ and no two elements of $W$ agree in any coordinate.  

    Two problem instances $(U, n)$ and $(U', n')$ are isomorphic if there exist a bijection $f \colon [n] \to [n']$ with $U' = \{(f(x), f(y), f(z)) \, | \, (x, y, z) \in [n]^3\}$.
\end{definition}

\begin{definition}[$\EXA$]
    In the $\EXA$ problem we are given a number $n$ and a values of subsets $S_1, \dots, S_m \subseteq [n]$. We have to decide if there are indices $j_1, \dots j_s$ such that $S_{j_1}, \dots, S_{j_s}$ are disjoint and $\bigcup_{t = 1}^s S_{j_t} = \bigcup_{i = 1}^m S_i$.

    Two problem instances $(S_1, \dots, S_m, n)$ and $(S'_1, \dots, S'_{m'}, n')$ are isomorphic if $m = m'$, $n = n'$ and there are bijective functions $f \colon [n] \to [n']$ and $g \colon [m] \to [m']$ such that for each $i \in [m]$ we obtain $f(S_i) = \{f(s) \, | \, s \in S_i\} = S'_{g(j)}$
\end{definition}

\begin{definition}[$\FeedbackArcSet$]
    In the $\FeedbackArcSet$ problem we are given a directed graph $H$ and a number $k$. We have to decide if there exist a set of edges $S \subseteq E(G)$ with $|S| \leq k$ such that every cycle in $G$ contains at least on edge in $S$.

    Two problem instances $(H, k)$ and $(H', k')$ are isomorphic if $k = k'$ and $H \cong H'$ are isomorphic as graphs.
\end{definition}

\begin{definition}[$\FeedbackVertexSet$]
    In the $\FeedbackVertexSet$ problem we are given a directed graph $H$ and a number $k$. We have to decide if there exist a set of vertices $R \subseteq V(G)$ with $|R| \leq k$ such that every cycle in $G$ contains at least on vertex in $R$.

    Two problem instances $(H, k)$ and $(H', k')$ are isomorphic if $k = k'$ and $H \cong H'$ are isomorphic as graphs.
\end{definition}

\begin{definition}[$\DirHam$]
    In the $\DirHam$ problem we are given a directed unweighted graph $G$. We have to decide if $G$ contains a directed cycle which includes each node exactly once.

    Two problem instances $G$ and $G'$ are isomorphic if $G$ and $G'$ are isomorphic as graphs. 
\end{definition}

\begin{definition}[$\UndirHam$]
    In the $\UndirHam$ problem we are given a undirected unweighted graph $G$. We have to decide if $G$ contains a directed cycle which includes each node exactly once.

    Two problem instances $G$ and $G'$ are isomorphic if $G$ and $G'$ are isomorphic as graphs. 
\end{definition}

\begin{definition}[$\HIT$]
    In the $\HIT$ problem we are given a number $n$ and a values of subsets $S_1, \dots, S_k \subseteq [n]$. We have to decide if there is a set $W \subseteq [n]$ such that $|W \cap S_i| = 1$.

    Two problem instances $(S_1, \dots, S_k, n)$ and $(S'_1, \dots, S'_{k'}, n')$ are isomorphic if $k = k'$, $n = n'$ and there are bijective functions $f \colon [n] \to [n']$ and $g \colon [k] \to [k']$ such that for each $i \in [k]$ we obtain $f(S_i) = \{f(s) \, | \, s \in S_i\} = S'_{g(j)}$
\end{definition}

\begin{definition}[$\isName$ ($\IS$)]
    In the $\Clique$ problem we are given a simple graph $G$ and a value $k$. We have to decide if $G$ contains an independent set of size $k$ (i.e., a subgraph with $k$ vertices such that no vertex is connected to any other vertex)

    Two problem instances $(G, k)$ and $(G', k')$ are isomorphic if $G$ and $G'$ are isomorphic as graphs and $k = k'$.
\end{definition}

\begin{definition}[$\MaxCut$]
    In the $\MaxCut$ problem we are given a simple graph $G$ and a value $k$. We have to decide if there exist a subset of vertices $S \subseteq V(G)$ such that there are at least $k$ many edges between $S$ and $V(G) \setminus S$.

    Two problem instances $(G, k)$ and $(G', k')$ are isomorphic if $G$ and $G'$ are isomorphic as graphs and $k = k'$.
\end{definition}

\begin{definition}[$\kMedian$]
    In the $\kMedian$ problem we are given a set of facilities $F = \{f_1, \dots, f_t\}$, a set of clients $C = \{c_1, \dots ,c_s\}$, a function $d \colon F \times C \to \mathbb{R}_{\geq 0}$, that satisfies the triangle inequality, a number $k$ and a number $c$. We have to decide if there are exist a set $F' \subseteq F$ of size $k$ such that 
    \[\sum_{i = 1}^s \min_{f \in F'} d(f, c_i) \le c. \]

    Two problem instances $(F, C, d, k, c)$ and $(F', C', d', k', c')$ are isomorphic if $c = c'$, $k=k'$, and there are bijections $f \colon F \to F'$ and $g \colon C \to C'$ such that $d(x, y) = d'(f(x), g(y))$.
\end{definition}

\begin{definition}[Satisfiability ($\SAT$)]
    In the Satisfiability problem we are given a formula $\Phi$ in conjunctive normal form (CNF). We have to decide if $\Phi$ is decidable.

    Two CNF formulas $\Phi$, $\Phi'$  are isomorphic if $\Phi'$ can be obtained from $\Phi$ by a renaming of variables. Note that this renaming is consistent with negation (i.e., if we rename $x$ to $y$ then $\lnot x$ is mapped to $\lnot y$).
\end{definition}

\begin{definition}[3-Satisfiability ($\kSAT{3}$)]
    In the 3-Satisfiability problem we are given a formula $\Phi$ in conjunctive normal form (CNF) such that each clause contains exactly three literals. We have to decide if $\Phi$ is decidable.

    Two CNF formulas $\Phi$, $\Phi'$  are isomorphic if $\Phi'$ can be obtained from $\Phi$ by a renaming of variables. Note that this renaming is consistent with negation (i.e., if we rename $x$ to $y$ then $\lnot x$ is mapped to $\lnot y$).
\end{definition}

\begin{definition}[$\MaxkSAT{k}$]
    In the $\MaxkSAT{k}$ we are given a formula $\Phi$ in conjunctive normal form (CNF) such that each clause contains exactly $k$ literals. Given an assignment $x$ of variables we write $|\Phi(x)|$ for the number of clauses that are satisfied by $x$. We have to return $\max_{x} |\Phi(x)|$.

    Two CNF formulas $\Phi$, $\Phi'$  are isomorphic if $\Phi'$ can be obtained from $\Phi$ by a renaming of variables. Note that this renaming is consistent with negation (i.e., if we rename $x$ to $y$ then $\lnot x$ is mapped to $\lnot y$).
\end{definition}

\begin{definition}[$\SetCover$]
    In the $\SetCover$ problem we are given a number $n$, a list of subsets $S_1, \dots, S_m \subseteq [n]$, and a number $k$. We have to decide if there are indices $j_1, \dots j_s$ such that $s \leq k$ and $\bigcup_{t = 1}^s S_{j_t} = \bigcup_{i = 1}^m S_i$.

    Two problem instances $(S_1, \dots, S_m, n, k)$ and $(S'_1, \dots, S'_{m'}, n', k')$ are isomorphic if $m = m'$, $k = k'$, $n = n'$ and there are bijective functions $f \colon [n] \to [n']$ and $g \colon [k] \to [k']$ such that for each $i \in [m]$ we obtain $f(S_i) = \{f(s) \, | \, s \in S_i\} = S'_{g(i)}$.
\end{definition}

\begin{definition}[$\SetPacking$]
    In the $\SetPacking$ problem we are given a family of sets $\mathcal{S} = \{S_1, \dots, S_m\}$ with $S_i \subseteq [n]$ and a number $k$. We have to decide if there exist $k$ many mutually disjoint sets in $\mathcal{S}$. 

    Two problem instances $(\mathcal{S} = (S_1, \dots, S_m), n, k)$ and $(\mathcal{S} = (S_1, \dots, S_m), n, k)$ are isomorphic if $m = m'$, $k = k'$, $n = n'$ and there are bijective functions $f \colon [n] \to [n']$ and $g \colon [k] \to [k']$ such that for each $i \in [m]$ we obtain $f(S_i) = \{f(s) \, | \, s \in S_i\} = S'_{g(i)}$.
\end{definition}

\begin{definition}[$\STEINER$]
    In the $\STEINER$ problem we are given a weighted undirected graph $G$, set of vertices $R$, and a value $k$. We have to decide if $G$ contains a subtree that contains each vertex in $R$ and has cost (i.e. sum of edges) at most $k$.

    Two problem instances $(G, R, k)$ and $(G', R', k')$ are isomorphic if $k = k'$ and there is an isomorphism $\varphi \colon V(G) \to V(G')$ from $G$ to $G'$ that respects the edge weights and $\varphi$ maps $R$ to $R'$.
\end{definition}

\begin{definition}[$\tspName$ ($\TSP$)]
    In the $\tspName$ we are given a weighted clique with $n$ vertices where each edge-weight is non-negative. We have to return the length of the humiliation cycle with the smallest length.

    Two problem instances $G$ and $G'$ are isomorphic if the underlying graphs are isomorphic.
\end{definition}

\begin{definition}[$\metrictspName$ ($\MetricTSP$)]
    \phantom{ssssssssssssssssssssssssssssssss}  In the $\metrictspName$ we are given a weighted clique with $n$ vertices where each edge-weight is non-negative. Further the edge-weight satisfy the triangle inequality. We have to return the length of the humiliation cycle with the smallest length.

    Two problem instances $G$ and $G'$ are isomorphic if the underlying graphs are isomorphic.
\end{definition}

\begin{definition}[Vertex-Cover ($\VC$)]
    In the Vertex-Cover problem we are given a simple graph $G$ and a value $k$. We have to decide if $G$ contains a vertex cover of size at most $k$.

    Two problem instances $(G, k)$ and $(G', k')$ are isomorphic if $G$ and $G'$ are isomorphic as graphs and $k = k'$.
\end{definition}

\subsection{Additional Isomorphism-Preserving Reductions}

In the following we show many different isomorphism-preserving reductions. However, many of them are simple modifications of already known reductions.

\begin{lemma}\label{lem:karp:red}
    All of the following reductions do hold:
    \begin{multicols}{2}
    \begin{itemize}
        \item $\SAT \isored \COL$
        \item $\VC \isored \DirHam$
        \item $\VC \isored \SetCover$
        \item $\VC \isored \FeedbackVertexSet$
        \item $\VC \isored \FeedbackArcSet$
        \item $\SAT \isored \LPZEROONE$
        \item $\SAT \isored \Clique$
        \item $\COL \isored \EXA$
        \item $\EXA \isored \STEINER$
        \item $\EXA \isored \HIT$
        \item $\EXA \isored \DMatching$
        \item $\COL \isored \CliqueCover$
        \item $\Clique \isored \SetPacking$
    \end{itemize}
    \end{multicols}
    and $\DirHam \isored \UndirHam$.
\end{lemma}
\begin{proof}
    All of these reductions can be obtained by using the corresponding reduction in \cite{DBLP:conf/coco/Karp72} and checking that the reduction is isomorphism-preserving. 
\end{proof}

\begin{lemma}\label{lem:iso:red:VC}
    $\IS \isored \VC$
\end{lemma}
\begin{proof}
    A folklore result states that $S$ is an independent set of a graph $G$ if and only if $\bar{S} \coloneqq V(G) \setminus S$ is a vertex cover of $G$. Thus, a function $R$ that maps $(G, k)$ to $(G, |V(G)| - k)$ is an isomorphism-preserving reduction from $\IS$ to $\VC$.
\end{proof}

\begin{lemma}\label{lem:maxcut:isored:max2sat}
    $\MaxCut \isored \MaxkSAT{2}$
\end{lemma}
\begin{proof}
    Let $(G=(V,E),k)$ be an instance of $\MaxCut$. Let $R$ be the function that maps $(G,k)$ to $\Phi$ where $\{x_v\mid v\in V\}$ are the variables of $(\Phi,|E|+k)$ and $\bigcup_{\{v,u\}\in E}\{(x_v\vee  x_u),(\neg x_v \vee \neg x_u)\}$.

    If there is a set $S\subseteq V$ that cuts at least $k$ edges, then the assignment 
    \begin{align*}
        x_v = \begin{cases}
            1&\text{, if } v\in S\\
            0&\text{, if } v\notin S
        \end{cases}
    \end{align*}
    satisfies both clauses for a cut edge and one clause for every other edge. This implies that $|E|+k$ clauses are satisfied for this assignment.

    If there is an assignment $x_v$ that satisfies $|E|+k$ clauses, then the set $S=\{v\in V\mid x_v=1\}$ cuts at least $k$ edges because there are $k$ edges for which both constraints are satisfied. These edges are cut by $S$.

    Let $(G_1=(V_1,E_1),k_1)$ and $(G_1=(V_1,E_1),k_2)$ be isomorphic instances of $\MaxCut$ with isomorphism $\phi\colon V\rightarrow V$. Then, the formulas in  $R(G_1,|E_1|+k_1)$ and $R(G_2,|E_2|+k_2)$ are isomorphic to each other due to renaming the variables using $\phi$. Also it holds that
    \begin{align*}
        |E_1|+k_1 = |E_1|+k_2=|E_2|+k_2.
    \end{align*}
\end{proof}

\begin{lemma}\label{lem:isored:MetricTSP}
    $\UndirHam \isored \MetricTSP$
\end{lemma}
\begin{proof}
    Let $G=(V,E)$ be an instance of $\UndirHam$.
    Let $R$ be the function that maps $G$ to $(G'=(V,\binom{V}{2},w),|V|)$ such that
    \begin{align*}
        w(e)=\begin{cases}
            1&\text{, if } e\in E\\
            2 &\text{, if }e \notin E
        \end{cases}
    \end{align*}
    for all $e\in \binom{V}{2}$.
    Note that this weight functions satisfies the triangle inequality because it has only values from $\{1,2\}$.

    If $G$ has a Hamiltonian Cycle, then it has $G'$ has a Hamiltonian Cycle of cost at most $|V|$.

    If $G$ has no Hamiltonian Cycle, then every Hamiltonian Cycle in $G'$ uses an edge $e\notin E$ and thus is of cost $w(e)=2$. Hence, every Hamiltonian Cycle in $G'$ has cost at least $|V|+1$.

    Let $G_1$ and $G_2$ be isomorphic instances of $\UndirHam$ with isomorphism $\phi\colon V\rightarrow V$. Then, $\phi$ is also an isomorphism for $R(G_1)$ and $R(G_2)$.
\end{proof}

This Lemma uses the reduction from \cite{DBLP:journals/jal/GuhaK99}.
\begin{lemma}\label{lem:Set:red:Mean}
    $\SetCover \isored \kMedian$.
\end{lemma}
\begin{proof}
    Let $(S_1, \dots, S_m,n,k)$ be an instance of $\SetCover$. Let $R$ be the function that maps $(S_1, \dots, S_m,n,k)$ to
    $([m],[n],d,k,n)$ and 
    \begin{align*}
        d(i,j) =\begin{cases}
            1 &\text{, if } j\in S_i\\
            3 &\text{, otherwise}
        \end{cases}
    \end{align*}
    for all $i\in [m]$ and $j\in[n]$.
    In \cite{DBLP:journals/jal/GuhaK99}, they show that $R(S_1, \dots, S_m,n,k)$ is a YES-Instance of $\kMedian$ if and only if $(S_1, \dots, S_m,n,k)$ is a YES-Instance of $\SetCover$.

    It remains to show that the reduction $R$ is isomorphism-preserving. To this end,
    we start by considering two isomorphic instances $(S_1, \dots, S_m,n,k)$ and $(S_1', \dots, S_{m'}',n',k')$. By definition, there are two bijective functions $f\colon [n]\rightarrow [n']$ and $g\colon [m]\rightarrow[m']$ with $S'_{g(i)} = f(S_i)$.
    Then $R(S_1, \dots, S_m,n,k)=([m],[n],d,k,n)$ and $R(S_1', \dots, S_{m'}',n',k')=([m'],[n'],d',k',n')$ are isomorphic due to the same bijective functions $f$ and $g$ because for all $i\in [m]$ and $j\in[n]$ it holds that,
    \begin{align*}
        d(i,j)=1\Leftrightarrow j\in S_i \Leftrightarrow f(j)\in S_{g(i)}'\Leftrightarrow d(f(j),g(i)) = 1.
    \end{align*}
\end{proof}

\begin{lemma}\label{lem:vc:isored:GapCenter}
    $\VC \isored \mathsf{Gap\text{-}k\text{-}Center}_{1,2}$
\end{lemma}
\begin{proof}
    Let $((V,E),k)$ be an instance of $\VC$.
    Let $R$ be the function that maps $((V,E),k)$ to $((V,E,w),k)$ such that $w(e)=1$ for all $e\in E$.
    
    If there is a vertex cover $X$ of size $k$ in $V$, then it holds that
    \begin{align*}
        \max_{v\in V}d_{(V,E),w} (v,X) \le 1.
    \end{align*}
    If $(V,E)$ does not contain a vertex cover of size $k$, then for all subsets $X\subseteq V$ of size $k$ it holds that 
    \begin{align*}
        \max_{v\in V}d_{(V,E),w} (v,X) \ge 2.
    \end{align*}

    This is an isomorphism-preserving reduction because the graph is not changed when $R$ is applied.
\end{proof}

\begin{lemma}\label{lem:undirham:isored:gaptsp}
    $\UndirHam \isored \GAP{\TSP_{1, c}}$ For all $c > 1$.
\end{lemma}
\begin{proof}
    Let $G=(V,E)$ be an instance of $\UndirHam$.
    Let $R$ be the function that maps $G$ to $G'=(V,\binom{V}{2},w)$ such that
    \begin{align*}
        w(e)=\begin{cases}
            1/|V|&\text{, if } e\in E\\
            C &\text{, if }e \notin E
        \end{cases}
    \end{align*}
    for all $e\in \binom{V}{2}$.

    If $G$ has a Hamiltonian Cycle, then it has $G'$ has a Hamiltonian Cycle of cost at most $1$.

    If $G$ has no Hamiltonian Cycle, then every Hamiltonian Cycle in $G'$ uses an edge $e\notin E$ and thus is of cost $w(e)=C$. Hence, every Hamiltonian Cycle in $G'$ has cost at least $C$.

    Let $G_1$ and $G_2$ be isomorphic instances of $\UndirHam$ with isomorphism $\phi\colon V\rightarrow V$. Then, $\phi$ is also an isomorphism for $R(G_1)$ and $R(G_2)$.
\end{proof}

\end{document}